\documentclass[12pt]{article}

\usepackage{amsmath}
\usepackage{amssymb}
\usepackage{amsthm}
\usepackage{mathtools}
\usepackage[margin=1in]{geometry}
\usepackage{setspace}
\usepackage[colorlinks=true,allcolors=blue]{hyperref}
\usepackage{cleveref}
\usepackage{tikz}
\usepackage{bm}
\usepackage{soul}
\usepackage[authoryear,comma]{natbib}
\usepackage{booktabs}
\usepackage{multirow}
\usepackage{subcaption}
\usepackage{threeparttable}
\usepackage{makecell}
\usepackage{arydshln}
\usepackage{tabularx,booktabs}
\usepackage{comment}
\newcolumntype{Y}{>{\centering\arraybackslash}X}

\theoremstyle{plain}

\newtheorem{corollary}{Corollary}
\newtheorem{lemma}{Lemma}

\newtheorem{proposition}{Proposition}
\newtheorem{theorem}{Theorem}

\theoremstyle{definition}
\newtheorem{definition}{Definition}
\newtheorem{example}{Example}

\newtheoremstyle{named}{}{}{\itshape}{}{\bfseries}{.}{.5em}{\thmnote{#3's }#1 1}
\theoremstyle{named}

\crefalias{namedcorollary}{corollary}

\DeclareMathOperator{\bd}{bd}

\DeclareMathOperator{\Elas}{Elas}
\DeclareMathOperator{\ext}{ext}
\DeclareMathOperator{\Int}{int}

\DeclareMathOperator{\NC}{NC}
\DeclareMathOperator{\supp}{supp}
\DeclareMathOperator{\sgn}{sgn}
\DeclareMathOperator{\TR}{TR}
\DeclareMathOperator{\MR}{MR}

\newcommand{\dd}{\mathop{}\!\mathrm{d}}

\title{A Tale of Two Monopolies\thanks{We thank Simon Loertscher, Debasis Mishra, Benny Moldovanu, as well as conference/seminar participants from ACM EC 2025, APIOC 2025, ESWC 2025, Academia Sinica, Fudan University, Monash University, National Taiwan University, Shanghai University of Finance and Economics, University of Melbourne, and Wuhan University for stimulating discussions. Enze Huang provided excellent research assistance on numerical simulations.}}
\author{
Yi-Chun Chen\thanks{National University of Singapore. Email: \url{yichun@nus.edu.sg}}
\and
Zhengqing Gui\thanks{National University of Singapore. Email: \url{zgui@nus.edu.sg}}
}
\date{\today}

\begin{document}

\maketitle

\begin{abstract}
%We apply marginal analysis à la \citet{bulow1989simple} to characterize revenue-maximizing selling mechanisms for a multiproduct monopoly. We derive marginal revenue from price perturbations over arbitrary sets of bundles and show that optimal mechanisms admit no revenue-increasing perturbation for bundles with positive demand, nor revenue-decreasing perturbations for bundles with zero demand. For any symmetric two-dimensional type distribution under mild regularity, this analysis fully characterizes the optimal mechanism across independence, substitutability, and complementarity. For general type distributions and allocation spaces, our approach identifies bundles that must carry positive demand and provides conditions under which pure bundling or separate selling is suboptimal.

We apply marginal analysis à la \citet{bulow1989simple} to characterize revenue-maximizing selling mechanisms for a multiproduct monopoly. We derive marginal revenue from price perturbations over a set of bundles. For any symmetric two-dimensional type distribution under mild regularity, this analysis fully characterizes the optimal mechanism. The approach can also be used to study regulation or quantity constraints for multiproduct firms, and the sale of substitutes and complements. 
For general type distributions and allocation spaces, our approach identifies bundles that must be included in every optimal mechanism and provides conditions under which pure bundling or separate selling is suboptimal.

\bigskip


\end{abstract}

\clearpage

\section{Introduction}
\label{Sec_Intro}

Revenue maximization for a multiproduct monopoly is notoriously difficult. Even seemingly simple cases involving just two items remain poorly understood, despite decades of research \citep[p. 735; henceforth DDT]{daskalakis2017strong}. Known proofs of optimality of even a two-dimensional uniform distribution are involved, while optimal mechanisms for other simple distributions may require probabilistic bundling with a menu of infinitely many bundles.\footnote{See \citet{manelli2006bundling,pavlov2011optimal}, and \citet{giannakopoulos2018duality} for proofs regarding uniform distribution. DDT demonstrate distributions under which optimal mechanisms have infinite menu size.} In contrast, the revenue-maximizing mechanism for a single-product monopoly constrained by incentive compatibility (IC) is simply a take-it-or-leave-it price, which can be solved with the introductory textbook monopoly model. Building on this ``simple economics'' insight from \citet{bulow1989simple}, we revisit the analysis of multiproduct monopoly.

We formulate the bundle-pricing problem for an $N$-dimensional type space and a general convex allocation space $\mathcal{Q}$, with the two-product additive model as our leading application. We adopt the taxation principle to recast the revenue-maximization problem as solving nonlinear pricing for bundles. The taxation principle, like the revelation principle, bears no loss in the class of implementable mechanisms, but unlike the revelation principle, it flips the (IC-)constrained optimization problem into pricing bundles self-selected by buyer types drawn from a given distribution. A taxation mechanism attaches prices to bundles of any fraction of each good and thereby generalizes what \citet{manelli2006bundling} called a price schedule which is only defined on bundles of goods in whole units. As in nonlinear pricing, a fractional unit can be flexibly interpreted as quantity, quality, or probability, but unlike nonlinear pricing, a taxation mechanism prices absolute rather than incremental units; see \citet[pp. 51-53]{wilson1993nonlinear}.

We solve the bundle pricing problem by performing standard demand analysis. Different from \citet{bulow1989simple}, we focus on determining the marginal revenue (MR) with respect to a price perturbation rather than a quantity perturbation.\footnote{\citet[Theorem 1]{manelli2006bundling} pioneered the study of marginal revenue with respect to price perturbations but restrict attention to \emph{deterministic} mechanisms and first-order conditions. Alternatively, \citet{jehiel2007mixed} analyze the marginal effects of the weights attached to each partition of the objects among the bidders on welfare and revenue, in the spirit of \citet{roberts1979characterization}, to obtain mixed-bundling auctions.}
Given any optimal taxation mechanism, holding the prices of other bundles fixed, total revenue must not increase with any small price change for bundles with positive demand (FOC), nor with a small price decrease for bundles with zero demand (SOC). Together, they constitute our main result, \Cref{Thm_MR}. In a single-product monopoly, each allocation is demanded by an interval of buyer types, and MR with respect to an interior allocation is the integral of an ``MR density'', which equals the derivative of the buyer's density-weighted virtual value. In a multiproduct monopoly, each allocation is demanded by a convex set of buyers, and the MR density is equal to the negative of the density of the transformed measure in DDT.\footnote{The transformed measure of DDT captures the marginal effect on revenue with respect to an increase in the rent surrendered to a subset of buyer types, whereas MR captures the marginal effect on revenue with respect to an increase in the price of a subset of bundles. The two effects are of the same magnitude but opposite in sign, as the price functional and the buyer's indirect utility function are convex conjugates.}

Despite the connection, the MR analysis differs from DDT's optimal transport approach in several aspects. First, FOC and SOC are directly formulated on the primal problem based on their economic interpretations, whereas DDT focus on the dual problem, which requires finding a measure that convexly dominates the transformed measure and constructing an optimal multidimensional transport plan based on this new measure.
Second, FOC and SOC apply to general convex type and allocation spaces satisfying our standing regularity assumptions, including settings with substitution or complementarity, whereas DDT's framework requires additional extensions to accommodate such environments.\footnote{See \citet{kash2016optimal}. For additive valuations drawn from $[0,1]^N$, \citet{kleiner2019strong} provides a shorter alternative proof for the strong duality result established by DDT.} Third, while solving DDT's dual problem certifies the optimality of a given mechanism, our FOC and SOC can pin down optimal mechanisms directly from primitives, especially in the two-product case. They also identify bundles that must be sold to a positive mass of buyers when the number of products is arbitrary.

To illustrate the power of FOC and SOC, we first examine symmetric two-dimensional joint distributions with positive MR density (strictly regular and symmetric, or SRS distributions), which may exhibit independence or correlation. Strict regularity corresponds to a strict version of the hazard rate condition introduced by \citet{mcafee1988multidimensional}, the earliest class of ``regular distributions'' proposed in the literature.\footnote{The hazard rate condition was imposed to extend the ``no-haggling result'' of \citet{riley1983optimal} to selling two items, but it was discovered later by \citet{thanassoulis2004haggling} and \citet{manelli2006bundling} among others that revenue maximization for some strictly regular distributions still requires randomization.} For SRS distributions, FOC and SOC produce an analogue of ``no distortion at the top'': every active nonzero bundle allocates a whole unit of at least one good. %which was documented by \citet{mcafee1988multidimensional} using a different argument. 
As it turns out, the demand for these bundles is easy to pin down. As the distribution is symmetric, it suffices to consider a symmetric mechanism that sells bundles taking the form of $(q,1)$ to types above the 45-degree line, which we focus on subsequently, and bundles taking the form of $(1,q)$ to types below the 45-degree line. The indirect utility of type $(x_1,x_2)$ above the 45-degree line can then be expressed in terms of the indirect utility function $\overline{u}$ of $(x_1,1)$ according to $\max (\overline{u}(x_1)-(1-x_2),0)$. By the envelope theorem, a boundary type $(x_1,1)$ consumes the bundle $(\overline{u}'(x_1),1)$; hence, type $(x_1,x_2)$ also consumes the bundle $(\overline{u}'(x_1),1)$ up to where their participation constraint is binding, namely, $x_2=1-\overline{u}(x_1)$. As a result, solving the optimal mechanism amounts to solving the one-dimensional indirect utility function $\overline{u}$.

FOC and SOC offer a unified treatment of optimal mechanisms for SRS distributions. Specifically, each SRS distribution identifies a \emph{type-dependent posted price} $\zeta $, which is a pointwise solution to the cumulative $\MR=0$. For the two-dimensional uniform distribution on $[0,1]^2$, $\zeta $ is constantly equal to $2/3$. Moreover, the cumulative MR at $x $ is in deficit (surplus) when $\max\{1-\overline{u}(x ),x \}$ lies above (below) $\zeta (x )$.
Our first observation is that the types with nearly maximal values run an MR deficit, so SOC requires that the grand bundle be active. Moreover, it is impossible for bundle(s) $(q,1)$ with $q\in (0,1)$ to have a positive demand; otherwise, its MR is zero (FOC), which requires, by concavity of $1-\overline{u}$, an MR deficit on the left to cancel the MR surplus on the right. Then, selling any quantity of good 1 smaller than $q$, which will be bought by types to the left of those purchasing $(q,1)$, can have neither zero demand (by SOC) nor positive demand (by FOC), as illustrated by panels (a) and (b) of Figure \ref{Fig_Uniform}. FOC for the grand bundle then pins down a cutoff $(2-\sqrt{2})/3$ at which the MR surplus exactly offsets the MR deficit; see panel (c) of Figure \ref{Fig_Uniform}. To the left of the cutoff, FOC for bundle $(0,1)$ implies that $1-\overline{u}$ must coincide with $\zeta$. As a result, the marginal analysis fully identifies the optimal mechanism for the uniform distribution.

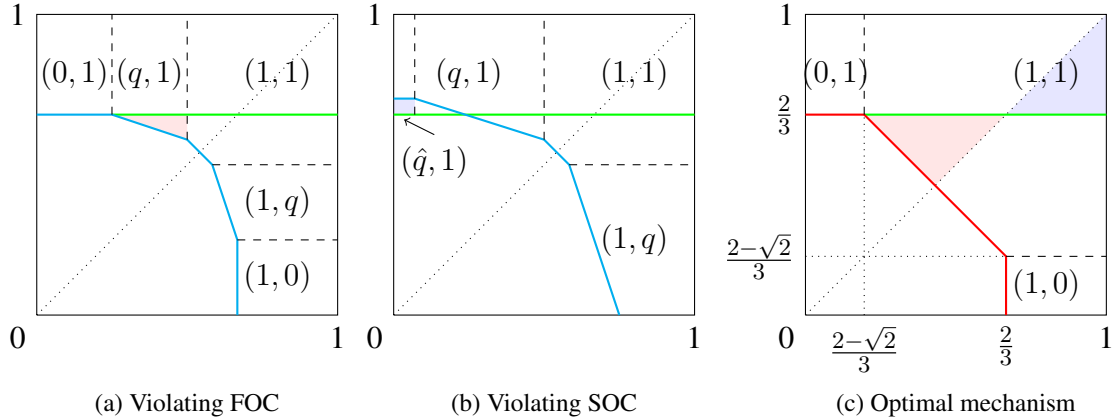
\begin{figure}[ht]
\centering
\begin{tikzpicture}[baseline={(0,0)},x=4.5cm,y=4.5cm]
\fill [red!10] (1/4,2/3) -- (1/2,2/3) -- (1/2,7/12);
%\fill [blue!10] (2/3,2/3) -- (1,1) -- (1,2/3);
\draw (0,0) node [below left] {0} -- (1,0) node [below] {1} -- (1,1) -- (0,1) node [left] {1} -- (0,0);
\draw [dotted] (0,0) -- (1,1);
\draw [thick, green] (1/4,2/3) -- (1,2/3);
\draw [thick, cyan] (0,2/3) -- (1/4,2/3) -- (1/2,7/12) -- (7/12,1/2) -- (2/3,1/4) -- (2/3,0);
\draw (1/8,0.8) node {$(0,1)$};
\draw (0.8,1/8) node {$(1,0)$};
\draw (3/8,0.8) node {$(q,1)$};
\draw (0.8,3/8) node {$(1,q)$};
\draw (0.8,0.8) node {$(1,1)$};
\draw [dashed] (1/4,2/3) -- (1/4,1);
\draw [dashed] (2/3,1/4) -- (1,1/4);
\draw [dashed] (1/2,7/12) -- (1/2,1);
\draw [dashed] (7/12,1/2) -- (1,1/2);
\draw (0.5,-0.3) node {\footnotesize (a) Violating FOC};
\end{tikzpicture}%
\begin{tikzpicture}[baseline={(0,0)},x=4.5cm,y=4.5cm]
%\fill [red!10] ({2/3-sqrt(2)/3},2/3) -- (2/3,2/3) -- ({2/3-sqrt(2)/6},{2/3-sqrt(2)/6});
\fill [blue!10] (0,2/3) -- (0.07,2/3) -- (0.07,0.72) -- (0,0.72);
%\draw (0,0) node [below left] {0} -- (1,0) node [below] {1} -- (1,1) -- (0,1) node [left] {1};
%\draw (0,3/4) -- (0,0);
\draw (0,0) node [below left] {0} -- (1,0) node [below] {1} -- (1,1) -- (0,1) node [left] {1} -- (0,0);
\draw [dotted] (0,0) -- (1,1);
\draw [thick, green] (0,2/3) -- (1,2/3);
\draw [thick, cyan] (0,0.72) -- (0.07,0.72) -- (1/2,7/12) -- (7/12,1/2) -- (3/4,0);
%\draw [thick, red] (0,3/4-0.05) -- (0.05,3/4-0.05) -- (0.05,1) -- (0,1);
\draw (2/8,0.8) node {$(q,1)$};
\draw (0.8,2/8) node {$(1,q)$};
\draw (0.8,0.8) node {$(1,1)$};
\draw [dashed] (1/2,7/12) -- (1/2,1);
\draw [dashed] (7/12,1/2) -- (1,1/2);
\draw [dashed] (0.07,2/3) -- (0.07,1);
\draw [->] (0.135,0.6) node [below] {$(\hat{q},1)$} -- (0.035,0.65);
\draw (0.5,-0.3) node {\footnotesize (b) Violating SOC};
\end{tikzpicture}%
\begin{tikzpicture}[baseline={(0,0)},x=4.5cm,y=4.5cm]
\fill [red!10] ({2/3-sqrt(2)/3},2/3) -- (2/3,2/3) -- ({2/3-sqrt(2)/6},{2/3-sqrt(2)/6});
\fill [blue!10] (2/3,2/3) -- (1,1) -- (1,2/3);
\draw (0,0) node [below left] {0} -- (1,0) node [below] {1} -- (1,1) -- (0,1) node [left] {1} -- (0,0);
\draw [dotted] (0,0) -- (1,1);
\draw [thick, green] ({2/3-sqrt(2)/3},2/3) -- (1,2/3);
\draw [thick, red] (0,2/3) node [black, left] {$\frac{2}{3}$} -- ({2/3-sqrt(2)/3},2/3) -- (2/3,{2/3-sqrt(2)/3}) -- (2/3,0) node [black, below] {$\frac{2}{3}$};
\draw (0.1,0.8) node {$(0,1)$};
\draw (0.8,0.1) node {$(1,0)$};
\draw (0.8,0.8) node {$(1,1)$};
\draw [dotted] (0,{2/3-sqrt(2)/3}) node [left] {$\frac{2-\sqrt{2}}{3}$} -- (2/3,{2/3-sqrt(2)/3});
\draw [dashed] (2/3,{2/3-sqrt(2)/3}) -- (1,{2/3-sqrt(2)/3});
\draw [dotted] ({2/3-sqrt(2)/3},0) node [below] {$\frac{2-\sqrt{2}}{3}$} -- ({2/3-sqrt(2)/3},2/3);
\draw [dashed] ({2/3-sqrt(2)/3},2/3) -- ({2/3-sqrt(2)/3},1);
\draw (0.5,-0.3) node {\footnotesize (c) Optimal mechanism};
\end{tikzpicture}
\caption{Two products with uniformly distributed types.}
\label{Fig_Uniform}
\end{figure}

The cross-subsidization of MR is reminiscent of “bunching” or “ironing” in unidimensional screening problems. In the two-dimensional setting, the function $\overline{u}$ simultaneously determines both the menu of bundles and the endogenous, type-dependent participation constraints. The one-dimensional trade-off between $1-\overline{u}$ and $\zeta$ sidesteps the verification of convex dominance between multidimensional measures in DDT, as the threshold function $\zeta$ can be directly computed from the type distribution.\footnote{We may also think of $\zeta$ as a ``benchmark transportation plan'' shipping the positive MR in the interior to the negative MR at the top boundary/right boundary. However, note that ``ironing'' here pertains to the primal variables $\overline{u}$ of the revenue maximization problem, whereas the transportation problem in DDT pertains to the dual variables.} Our analysis for the uniform distribution immediately generalizes to two different classes of distributions. First, when $\zeta$ is nondecreasing, as implied by the assumptions in \citet{manelli2006bundling} and \citet{hart2010revenue}, the same analysis implies that the optimal mechanism is deterministic. Alternatively, if $\zeta$ is nonincreasing and concave, as in Example 3 of DDT or as assumed by \citet{giannakopoulos2018selling}, the optimal one-dimensional indirect utility function $\overline{u}$ is characterized by one cutoff $a^1$ such that $1-\overline{u}$ coincides with $\zeta$ for all $x_1 < a^1$, and $\overline{u}'(x_1)=1$ for all $x_1 \geq a^1$.\footnote{\citet{giannakopoulos2018selling} restrict attention to independent distributions and do not prove that the grand bundle must be active to pin down $a^1$.} When $a^1$ is positive and $\zeta$ is strictly concave, the optimal $\overline{u}$ features probabilistic bundling with a menu of infinitely many bundles. %Both classes of distributions include the uniform distribution and confirm that its optimal mechanism is deterministic and $1-\overline{u}$ coincides with $\zeta$ to the left of $(2-\sqrt{2})/3$.

The same analysis also identifies optimal mechanisms for other distributions that have not been solved previously. First, when $\zeta$ is concave but not monotone, which is the case for a truncated normal distribution, we show that $\overline{u}$ is characterized by at most two cutoffs $b^0 \leq a^1$ such that the bundle $(0,1)$ is sold to types to the left of $b^0$ and the grand bundle is sold to types to the right of $a^1$; moreover, $1-\overline{u}$ coincides with $\zeta$ between the two cutoffs. For a truncated normal distribution, $\zeta$ is strictly concave and peaks at one half of the peak of the density. Our numerical simulations show that the optimal mechanism has an infinite menu when the peak of the density is small and becomes deterministic when it is large, with a cutoff of approximately $0.22$.\footnote{\citet{schmalensee1984gaussian} pioneers the investigation of bundling strategies for bivariate normal distributions (without truncation) but restricted attention to deterministic mechanisms.} Second, when $\zeta$ is convex, which is the case for a truncated Pareto distribution, the mechanism sells at most three bundles: above the 45-degree line, the grand bundle for types to the right of a cutoff $a^1$ and another bundle $(q,1)$ with $q<1$ for types to the left of $a^1$. Again, when $q>0$, the optimal mechanism features what we call ``probabilistic mixed bundling''. In general, as long as $\zeta$ is sufficiently well-behaved, FOC and SOC reduce solving the optimal indirect utility function to solving an optimization problem with few variables, which is reminiscent of solving partitional signals for information design.\footnote{See, for instance, \citet{dworczak2019simple,kleiner2021extreme,chen2023information,bergemann2024bidder}, and \citet{lyu2023coarse}.} For general $\zeta$, we provide an upper bound on the number of pooling bundles even without imposing a global monotonicity, convexity, or concavity condition on $\zeta$. Moreover, we provide a necessary and sufficient condition for the optimality of $\overline{u}$, which is entirely based on our perturbation approach. Together with the MR conditions, it pins down the remaining parameters and serves as a certificate of optimality.

We highlight the unifying power of the MR approach through three economic applications. The first two modify the MR measure, whereas the third modifies the allocation space and hence the geometry along which MR is aggregated.

First, in regulating a monopolist with privately known costs, placing welfare weight $\omega$ on the firm's profit changes only the interior MR density. Under symmetry and modified strict regularity, every symmetric optimal mechanism admits an equivalent implementation using aggregate-output targets, with the firm choosing the cost-minimizing composition. Unlike \citet{armstrong2018multiproduct}, this result does not require independence between aggregate and relative costs. Moreover, under uniform distribution, increasing $\omega$ raises output and renders relative costs eventually irrelevant for allocation.

Second, under an aggregate quantity floor, the multiplier $\lambda$ makes MR mass on the upper edges more negative, creates positive MR mass on the lower edges, and may also change the interior MR density. Equivalently, it acts as a common per-unit subsidy in an auxiliary problem obtained by shifting every type by $\lambda\mathbf{1}$. With uniform types, tightening the floor moves the optimum from mixed bundling to probabilistic bundling, and ultimately to pure bundling. Thus, the constraint changes how products are combined, and whether lotteries are used—not merely total output. The finding complements \citet{loertscher2022monopoly}: their randomization hanges which qualities clear the market, whereas ours changes how multiple products are combined to satisfy a sales requirement.

Third, substitution and complementarity reshape the allocation space and therefore the geometry of demand. We parameterize the value of the grand bundle by $k(x_1+x_2)$ and let $\mathcal Q_k$ be the convex hull of ${\mathbf0,\mathbf e_1,\mathbf e_2,(k,k)}$. The goods are perfect substitutes at $k=1/2$, partial substitutes for $1/2<k<1$, and complements for $k>1$. Changing $k$ changes the slope of the upper boundary of $\mathcal Q_k$. This changes the direction in type space along which buyers demand the same allocation and, consequently, the direction along which MR is accumulated. For SRS distributions, the corresponding cutoff $\zeta_k$ is obtained by integrating MR along this direction, allowing the one-dimensional method of the additive benchmark to continue. For uniform types, the MR approach provides new characterizations for interaction parameters not covered by the familiar additive and perfect-substitute benchmarks. It proves that deterministic mixed bundling is optimal throughout the previously unresolved partial-substitution region, $1/2<k<1$, and reduces every symmetric optimum under complementarity, $k>1$, to at most the grand bundle and a symmetric pair of lotteries. In particular, lotteries are optimal under weak complementarity and disappear in favor of pure bundling at $k\approx 1.21$.

\begin{figure}[ht]
\centering
\begin{tikzpicture}[baseline={(0,0)},x=4.5cm,y=4.5cm]
\fill [blue!10] (0,2/3) -- (1/3,1) -- (0,1);
\fill [red!10]  (0,0.5) -- (0.5,0.5) -- (2/3,2/3) -- (0,2/3);
\draw (0,0) node [below left] {0} -- (1,0) node [below] {1} -- (1,1) -- (0,1) node [left] {1} -- (0,0);
\draw [dotted] (0,0) -- (0.5,0.5);
\draw [dashed] (0.5,0.5) -- (1,1);
\draw [dotted] (0,2/3) -- (1/3,1);
\draw [thick, green] (0,2/3) node [left, black] {$\frac{2}{3}$} -- (2/3,2/3);
\draw [thick, red] (0,0.5) -- (0.5,0.5) -- (0.5,0);
\draw (0.4,0.8) node {$(0,1)$};
\draw (0.8,0.4) node {$(1,0)$};
\draw [->] (0.85,0.85) -- (0.9,0.9) node [midway, left] {$\mathbf{n}_{k}$};
\draw (0.5,-0.2) node {\footnotesize (a) $k=1/2$};
\end{tikzpicture}
\begin{tikzpicture}[baseline={(0,0)},x=4.5cm,y=4.5cm]
\fill [blue!10] (0,2/3) -- (1/6,1) -- (0,1);
\fill [blue!10] (2/3,2/3) -- (5/6,2/3) -- (1,1);
\fill [red!10]  (0,2/3) -- (2/3,2/3) -- (3/8,3/8) -- (0.25,0.5) -- (0,0.5);
\draw (0,0) node [below left] {0} -- (1,0) node [below] {1} -- (1,1) -- (0,1) node [left] {1} -- (0,0);
\draw [dotted] (0,0) -- (1,1);
\draw [thick, green] (0,2/3) node [left, black] {$\frac{2}{3}$} -- (5/6,2/3);
\draw [dotted] (0,2/3) -- (1/6,1);
\draw [dotted] (5/6,2/3) -- (1,1);
\draw [thick, red] (0,0.5) -- (0.25,0.5) -- (0.5,0.25) -- (0.5,0);
\draw [dashed] (0.25,0.5) -- (0.5,1);
\draw [dashed] (0.5,0.25) -- (1,0.5);
%\draw [thick, red] (0,2/3) node [black, left] {$\frac{2}{3}$} -- ({2/3-sqrt(2)/3},2/3) -- (2/3,{2/3-sqrt(2)/3}) -- (2/3,0) node [black, below] {$\frac{2}{3}$};
\draw (0.2,0.8) node {$(0,1)$};
\draw (0.8,0.2) node {$(1,0)$};
\draw (0.8,0.8) node {$(1,1)$};
\draw [->] (0.425,0.85) -- (0.45,0.9) node [midway, left] {$\mathbf{n}_{k}$};
\draw (0.5,-0.2) node {\footnotesize (b) $1/2<k<1$};
\end{tikzpicture}%
\begin{tikzpicture}[baseline={(0,0)},x=4.5cm,y=4.5cm]
\fill [red!10] (0,0.8) -- (2/15,11/15) -- (0,1);
\fill [blue!10] (2/15,11/15) -- (4/15,2/3) -- (1/6,2/3);
\fill [red!10] (4/15,2/3) -- (2/3,2/3) -- (0.5,0.5) -- (0.4,0.6);
\fill [blue!10] (2/3,2/3) -- (1,2/3) -- (1,1);
\draw (0,0) node [below left] {0} -- (1,0) node [below] {1} -- (1,1) -- (0,1) node [left] {1} -- (0,0);
\draw [dotted] (0,0) -- (1,1);
\draw [dotted] (0,2/3) -- (1/6,2/3);
\draw [thick, green] (1/6,2/3) -- (1,2/3);
\draw (0,2/3) node [left] {$\frac{2}{3}$};
\draw [dotted] (1/6,2/3) -- (0,1);
\draw [thick, red] (0,0.8) -- (0.4,0.6) -- (0.6,0.4) -- (0.8,0);
\draw [dashed] (0.4,0.6) -- (0.2,1);
\draw [dashed] (0.6,0.4) -- (1,0.2);
\draw [->] (0.3,0.8) -- (0.275,0.85) node [midway, right] {$\mathbf{n}_{k}$};
\draw (0.8,0.8) node {$(k,k)$};
\draw (0.15,0.85) node {$(q_{1},q_{2})$};
\draw (0.85,0.15) node {$(q_{2},q_{1})$};
\draw (0.5,-0.2) node {\footnotesize (c) $k>1$};
\end{tikzpicture}
\caption{Substitutes or complements with uniformly distributed types.}
\label{Fig_Uniform_k}
\end{figure}
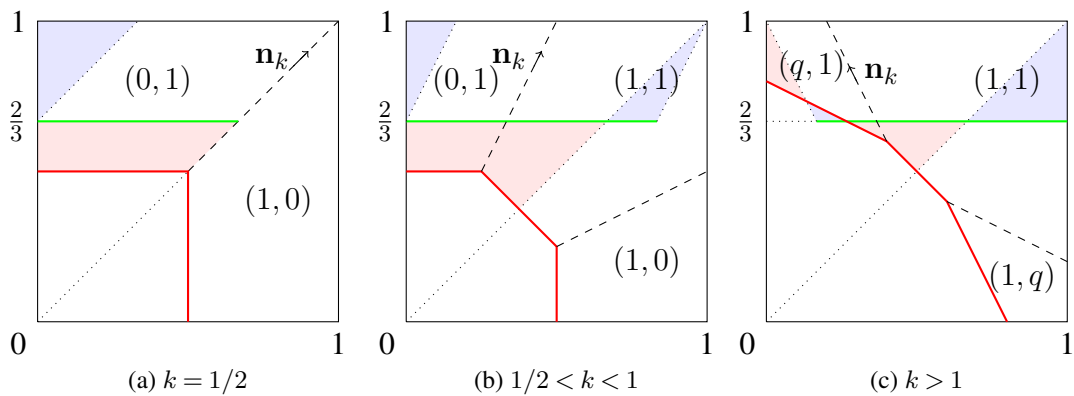

Beyond two items, FOC and SOC continue to reveal structural features of optimal mechanisms. First, under mild assumptions, we apply SOC to show that every bundle that uniquely maximizes some type's gross utility must be pooling in every optimal mechanism, thereby provide an almost distribution-free lower bound on menu size. In particular, the grand bundle must be pooling when all products are weak complements, whereas every separate item must be pooling when they are substitutes.
Second, under strict regularity, FOC restricts an active bundle to exhaust feasibility in at least one outward-normal direction associated with a top facet of the type space. This directional analogue of ``no distortion at the top'' excludes every interior bundle and applies to the identical-products model in \citet{bikhchandani2022selling,bikhchandani2024rank}. Third, we offer sufficient conditions for the existence of an exclusion set, which can be viewed as a counterpart of Proposition 1 of \citet{armstrong1996multiproduct} and forms the basis of the construction of exclusion mechanisms in Section 7 of DDT. In particular, exclusion occurs whenever $\mathcal{X}=\mathcal{Q}=[0,1]^{N}$, generalizing Theorem 16 of \citet{manelli2007multidimensional}. 
Finally, FOC and SOC provide necessary conditions that rule out simple bundling strategies. For additive buyers with any fixed i.i.d.\ marginal density satisfying mild regularity and positivity conditions, pure bundling is suboptimal when the number of products is sufficiently large, extending DDT's result beyond uniform distributions. For unit-demand buyers, symmetric separate selling is suboptimal whenever $x\mapsto f(x,\dotsc,x)$ is nonincreasing and $N\leq\underline{x}/(\overline{x}-\underline{x})$, without requiring independence or symmetry of the type distribution. This complements the optimality result of \citet{hashimoto2025selling} for uniformly distributed values with a sufficiently low lower bound.

The remainder of the paper is organized as follows. \Cref{Sec_MainResults} develops the marginal-revenue approach and introduces our FOC and SOC conditions. \Cref{Sec_Additive} applies these conditions to the two-product additive benchmark and characterizes optimal mechanisms. \Cref{Sec_Applications} applies the marginal-revenue approach to study quantity constraints, regulation, and substitutes and complements. \Cref{Sec_General} derives results on optimal mechanisms beyond two products. \Cref{Sec_Conclusion} concludes. Omitted proofs and additional results are collected in the Appendix.

%\Cref{Sec_MainResults} introduces our setting and the main theorem. \Cref{Sec_Additive} applies the main theorem to our leading case where the seller has two additive-value products. \Cref{Sec_Applications} presents three applications. \Cref{Sec_General} states some general results that are not restricted to two products. \Cref{Sec_Conclusion} concludes. Omitted proofs and discussions are relegated to the Supplemental Appendix.

\section{Marginal Revenue in Multiproduct Monopoly}
\label{Sec_MainResults}

\subsection{Setup}
\label{Sec_Model}

A seller (she) produces and sells $N\geq 1$ products to a unit mass of buyers (he). Production is costless, and the set of feasible quantity vectors $\mathbf{q}$ is denoted by $\mathcal{Q}$.\footnote{Throughout the paper, we use boldface letters such as $\mathbf{x}$ to denote vectors and $x_{n}$ to denote the $n$-th component of $\mathbf{x}$. We use $\mathbf{e}_{n}$ for the $n$-th standard basis vector, $\mathbf{c}$ for a vector whose components all equal $c$, and $\geq$ for the usual component-wise order.} We refer to $\mathbf{q}$ as a bundle or allocation and to $\mathcal{Q}$ as the allocation space. Each buyer has a private type $\mathbf{x}\in\mathcal{X}$ and obtains utility $\mathbf{x}\cdot\mathbf{q}$ from consuming bundle $\mathbf{q}$. Assume that $\mathcal{X}$ and $\mathcal{Q}$ are both convex and compact subsets of $\mathbb{R}_{+}^{N}=\{\mathbf{z}\in\mathbb{R}^{N}:\mathbf{z}\geq\mathbf{0}\}$, with nonempty interior and piecewise smooth boundaries, and that $\mathbf{0}\in\mathcal{Q}$, so buyers can always choose to purchase nothing.
Types are distributed according to a p.d.f.\ $f$ on $\mathcal{X}$. We assume that $f$ is continuous on $\mathcal{X}$, continuously differentiable and strictly positive on $\Int(\mathcal{X})$, and satisfies $\int_{\Int(\mathcal{X})}\lVert\nabla f(\mathbf{x})\rVert\dd\mathbf{x}<\infty$.
We further assume that the MR density $\phi$ defined below admits a continuous extension to $\mathcal{X}$. These conditions allow the derivatives of $f$ to be unbounded near the boundary.

The seller designs a price schedule (or a taxation mechanism) $p:\mathcal{Q}\to \mathbb{R}$ that assigns a take-it-or-leave-it price to each bundle. Buyers can then freely choose any $\mathbf{q}\in \mathcal{Q}$ and pay the corresponding price without reporting their types. A price schedule is \emph{feasible} if:
\begin{equation*}
p \text{ is lower semicontinuous},
\tag{LS}
\label{LS}
\end{equation*}
and
\begin{equation*}
p(\mathbf{0})\leq 0.
\tag{IR$_{p}$}
\label{IR_p}
\end{equation*}
Condition \eqref{LS} ensures that each buyer can find his most preferred bundle: Given any $\mathbf{x}\in \mathcal{X}$, $\mathbf{x}\cdot\mathbf{q}-p(\mathbf{q})$ is upper semicontinuous, so it attains its maximum on $\mathcal{Q}$. Consequently, buyers' indirect utility function is well-defined:
\begin{equation*}
\label{IC_p}
u(\mathbf{x})=\max_{\mathbf{q}\in \mathcal{Q}}\{\mathbf{x}\cdot\mathbf{q}-p(\mathbf{q})\}.
\tag{IC$_{p}$}
\end{equation*}
Condition \eqref{IC_p} is analogous to the incentive-compatibility constraint for direct mechanisms.
Condition \eqref{IR_p} ensures buyers' participation.\footnote{In fact, we only need $p(\mathbf{q})\leq 0$ for some $\mathbf{q}\in \mathcal{Q}$ to guarantee participation, while imposing $p(\mathbf{0})\leq 0$ is without loss of generality. It can also be verified that optimality requires $p(\mathbf{0})=0$. See \Cref{Sec_DirectMech} for an equivalent model setup using direct mechanisms.}

Since $u$ is convex, it is differentiable almost everywhere. By standard arguments, the seller's total revenue is
\begin{equation*}
\TR (p)=\int_{\mathcal{X}}[\mathbf{x}\cdot \nabla u(\mathbf{x})-u(\mathbf{x})]f(\mathbf{x})\dd\mathbf{x}.
\end{equation*}
Applying the divergence theorem, we can reformulate the total revenue as a functional of $u$:
\begin{equation*}
%\label{Eqn_Rev}
\TR (p)=\int_{\mathcal{X}}u\dd \mu,
\end{equation*}
where $\mu$ is a signed Radon measure defined as follows: For any measurable $\mathcal{A}\subseteq \mathcal{X}$,
\begin{equation*}
%\label{Eqn_mu}
\mu(\mathcal{A})=\int_{\mathcal{A}\cap \bd (\mathcal{X})}(\mathbf{x}\cdot\mathbf{n}_{\mathbf{x}})f(\mathbf{x})\dd \sigma (\mathbf{x})-\int_{\mathcal{A}}\phi(\mathbf{x})\dd \mathbf{x},
\end{equation*}
where $\mathbf{n}_{\mathbf{x}}$ is the outward unit normal to the boundary $\bd (\mathcal{X})$,\footnote{Since $\bd (\mathcal{X})$ is piecewise smooth, the set of points where $\mathbf{n}_{\mathbf{x}}$ is not uniquely defined is $\mu$-negligible.} $\sigma (\mathbf{x})$ is the surface integral on $\bd (\mathcal{X})$, and
\begin{equation*}
\phi (\mathbf{x})=\mathbf{x}\cdot \nabla f(\mathbf{x})+(N+1)f(\mathbf{x}),\quad \forall\, \mathbf{x}\in\Int(\mathcal{X}).
\end{equation*}
When $\mathcal{X}$ is a hyperrectangle, $\mu$ coincides with the transformed measure defined by DDT, except that there is no mass point at the bottom-left corner of $\mathcal{X}$. As a result, $\mu (\mathcal{X})$ equals $-1$ in our model but zero in DDT.

A mechanism is \emph{optimal} if it solves the seller's revenue-maximization problem:
\begin{equation*}
\begin{gathered}
\max_{p}\int_{\mathcal{X}}u\dd \mu,\\
\text{subject to \eqref{LS}, \eqref{IR_p}, and \eqref{IC_p}.}
\end{gathered}
\tag{P}
\label{Eqn_Problem_p}
\end{equation*}
%\zq{In \Cref{Sec_DirectMech}, we provide an equivalent formulation of the model in terms of direct mechanisms to show that our indirect-mechanism description is without loss of generality.}

%Define the subdifferential of $u$ as
%\begin{equation*}
%\partial_{\mathcal{Q}} u(\mathbf{x})=\{\mathbf{q}\in \mathcal{Q}: \mathbf{x}\cdot\mathbf{q}-u(\mathbf{x})\geq \hat{\mathbf{x}}\cdot\mathbf{q}-u(\hat{\mathbf{x}}),\ \forall\, \hat{\mathbf{x}}\in \mathcal{X}\}.
%\end{equation*}
%Then a type-$\mathbf{x}$ buyer purchases $\mathbf{q}$ if $\mathbf{q}\in \partial_{\mathcal{Q}} u(\mathbf{x})$.  Hence, 

%\section{Main Theorem}
%\label{Sec_MainResult}

\subsection{Demand Set and Marginal Revenue}

By \eqref{IC_p}, $u$ is the ($\mathcal{Q}$-restricted) convex conjugate of $p$, so the ($\mathcal{X}$-restricted) convex conjugate of $u$ is the ($\mathcal{X}$-restricted) closed convex envelope of $p$, denoted by\footnote{Throughout the paper, convex conjugates and subdifferentials are computed with respect to the actual domain, such as $\mathcal{X}$ or $\mathcal{Q}$, rather than the entire $\mathbb{R}^{N}$.}
\begin{equation*}
\overline{p}(\mathbf{q})=\max_{\mathbf{x}\in \mathcal{X}}\{\mathbf{x}\cdot\mathbf{q}-u(\mathbf{x})\}.
\end{equation*}
Since $p$ and $\overline{p}$ induce the same $u$, we always use $\overline{p}$ in the following analysis.

The \emph{demand set} of a bundle $\mathbf{q}$, denoted by $\mathcal{D}(\mathbf{q})$, is the set of buyers who are willing to purchase $\mathbf{q}$. Formally,
\begin{equation*}
\mathcal{D}(\mathbf{q})=\{\mathbf{x}\in \mathcal{X}: \mathbf{x}\cdot\mathbf{q}-\overline{p}(\mathbf{q})\geq \mathbf{x}\cdot\hat{\mathbf{q}}-\overline{p}(\hat{\mathbf{q}}),\ \forall\, \hat{\mathbf{q}}\in \mathcal{Q}\}.
\end{equation*}
By convex conjugacy, $\mathbf{x}\in \mathcal{D}(\mathbf{q})$ if and only if $\mathbf{q}$ belongs to the $\mathcal{Q}$-subdifferential of $u$:
\begin{equation*}
\mathbf{q}\in \partial_{\mathcal{Q}}u(\mathbf{x})=\{\mathbf{q}\in \mathcal{Q}:u(\mathbf{y})\geq u(\mathbf{x})+\mathbf{q}\cdot (\mathbf{y}-\mathbf{x}),\ \forall\, \mathbf{y}\in \mathcal{X}\},
\end{equation*}
and $\mathcal{D}(\mathbf{q})$ is just the $\mathcal{X}$-subdifferential of $\overline{p}$ at $\mathbf{q}$: $\mathcal{D}(\mathbf{q})=\partial_{\mathcal{X}} \overline{p}(\mathbf{q})$.
For a closed (and hence compact) set of bundles $Q\subseteq \mathcal{Q}$, define
\begin{equation*}
\begin{aligned}
\mathcal{D}(Q) & =\{\mathbf{x}\in\mathcal{X}:\partial_{\mathcal{Q}} u(\mathbf{x})\cap Q \neq \varnothing\}, \\
\mathcal{D}^{*}(Q) & =\{\mathbf{x}\in\mathcal{X}:\partial_{\mathcal{Q}} u(\mathbf{x})\subseteq Q\}.
\end{aligned}
\end{equation*}
Intuitively, $\mathcal{D}(Q)$ is the set of buyers that demand \emph{at least one} bundle in $Q$, whereas $\mathcal{D}^{*}(Q)$ is the set of buyers that demand \emph{only} bundles in $Q$.\footnote{By extending the domain of $\overline{p}$ to the entire $\mathbb{R}^{N}$, Theorem 24.7 of \citet{rockafellar1970} implies that $\mathcal{D}(Q)$ is compact and hence Borel measurable.}

The \emph{marginal revenue} of $Q$ is the marginal change in total revenue resulting from a uniform price perturbation across all bundles in $Q$.
For the perturbed price to be feasible, some restrictions are necessary.

Throughout, $B_{\delta}(A)$ denotes the open $\delta$-neighborhood of a set $A$ relative to its ambient space. In particular,
\begin{equation*}
B_{\delta}(Q)=\{\mathbf{q}\in\mathcal{Q}:\exists\,\hat{\mathbf{q}}\in Q\text{ such that }\lVert\mathbf{q}-\hat{\mathbf{q}}\rVert<\delta\}.
\end{equation*}
Let $Q\subseteq\mathcal{Q}$ be closed with $\mathbf{0}\notin Q$. For any $\varepsilon>0$ and all sufficiently small $\delta>0$, consider the perturbed price schedules
\begin{equation*}
\overline{p}-\varepsilon\mathbb{I}_{Q}
\quad\text{and}\quad
\overline{p}+\varepsilon\mathbb{I}_{B_{\delta}(Q)}.
\end{equation*}
Both price schedules are lower semicontinuous and feasible: the first because $Q$ is closed, and the second because $B_{\delta}(Q)$ is open in $\mathcal{Q}$ and does not contain $\mathbf{0}$.

Denote by $u_{-\varepsilon,Q}$ and $u_{+\varepsilon, B_{\delta}(Q)}$ the indirect utility functions induced by these two price schedules. Whenever the limits exist, define the one-sided marginal revenues by\footnote{The two perturbations are necessarily asymmetric when the allocation space is continuous. Intuitively, a price reduction affects a buyer whenever some optimal bundle lies in $Q$, whereas a price increase affects him only when \emph{every} optimal bundle lies in $Q$. The asymmetry disappears for finite or otherwise isolated menus.%A price reduction makes bundles in a closed set $Q$ more attractive, and the resulting maximization remains well behaved. By contrast, increasing prices only on $Q$ may allow buyers to evade the increase by choosing bundles arbitrarily close to $Q$. This topological asymmetry has an economic counterpart: 
}
\begin{equation*}
\begin{aligned}
\MR_{-}(Q) & =\lim_{\varepsilon\downarrow0}\frac{\TR(\overline{p}-\varepsilon\mathbb{I}_{Q})-\TR(\overline{p})}{-\varepsilon}
=\lim_{\varepsilon\downarrow0}\int_{\mathcal{X}}\frac{u_{-\varepsilon,Q}-u}{-\varepsilon}\dd\mu,\\
\MR_{+}(Q) & =\lim_{\delta\downarrow0}\lim_{\varepsilon\downarrow0}\frac{\TR(\overline{p}+\varepsilon\mathbb{I}_{B_{\delta}(Q)})-\TR(\overline{p})}{\varepsilon}
=\lim_{\delta\downarrow0}\lim_{\varepsilon\downarrow0}\int_{\mathcal{X}}\frac{u_{+\varepsilon, B_{\delta}(Q)}-u}{\varepsilon}\dd\mu.
\end{aligned}
\end{equation*}
If $\MR_{-}(Q)=\MR_{+}(Q)$, we say that $\MR(Q)$ exists and define its value as their common value. For a single bundle, we write $\MR(\mathbf{q})=\MR(\{\mathbf{q}\})$. We also use $\MR_{-}(Q;\overline{p}-\varepsilon\mathbb{I}_{Q})$ to denote the left marginal revenue under the price schedule $\overline{p}-\varepsilon\mathbb{I}_{Q}$. It measures the marginal revenue from a further uniform price reduction on $Q$ after the price has already been reduced by $\varepsilon$, thereby capturing the second-order effect of the original price perturbation.

\Cref{Lem_MR} characterizes the one-sided marginal revenues and gives an exact condition for the existence of $\MR(Q)$.

\begin{lemma}
\label{Lem_MR}
For every feasible mechanism and every closed set $Q$ with $\mathbf{0}\notin Q$,
\begin{equation*}
\begin{aligned}
\MR_{-}(Q) & =-\mu(\mathcal{D}(Q)), \\
\MR_{+}(Q) & =-\mu(\mathcal{D}^{*}(Q)).
\end{aligned}
\end{equation*}
Consequently, $\MR(Q)$ exists if and only if
\begin{equation}
\label{Eqn_Qualification}
\mu (\mathcal{D}(Q)\setminus \mathcal{D}^{*}(Q))=0,
\end{equation}
in which case $\MR(Q) =-\mu(\mathcal{D}(Q))$.
\end{lemma}

%\paragraph{When is MR well-defined?}

To see the intuition for \Cref{Lem_MR}, suppose that the set of feasible allocations $\mathcal{Q}$ is finite, and that each allocation is sold to a positive mass of buyers in a mechanism. Then
\begin{equation*}
\TR (p)=\int_{\mathcal{X}}\max_{\mathbf{q}\in \mathcal{Q}}\{\mathbf{x}\cdot\mathbf{q}-p(\mathbf{q})\} \dd \mu,
\end{equation*}
and by the envelope theorem,
\begin{equation*}
\MR (\mathbf{q})=\frac{\partial \TR (p)}{\partial p(\mathbf{q})}=\int_{\mathcal{D}(\mathbf{q})}-1\dd \mu =-\mu (\mathcal{D}(\mathbf{q})).
\end{equation*}
However, when a demand set has zero measure with respect to $f$, this shortcut may fail. Therefore, the qualification $\mu (\mathcal{D}(Q)\setminus \mathcal{D}^{*}(Q))=0$ in \Cref{Lem_MR} cannot be omitted.\footnote{In \Cref{Sec_MR_Example}, we provide an example in which the qualification fails and MR is not well-defined.}

\subsection{The Main Theorem}

If $p$ is optimal, then $\TR(\overline{p}-\varepsilon\mathbb{I}_{Q})$ and, for every sufficiently small $\delta>0$, $\TR(\overline{p}+\varepsilon\mathbb{I}_{B_{\delta}(Q)})$ should attain their global maxima at $\varepsilon=0$.
This implies the following two sets of necessary conditions.

First-order conditions (FOCs), which consist of three statements:
\begin{gather*}
\MR_{-}(Q) \geq 0, \tag{FOC$_{-}$}\label{Eqn_FOC-} \\
\MR_{+}(Q) \leq 0, \tag{FOC$_{+}$}\label{Eqn_FOC+} \\
\MR (Q) =0, \text{ provided that }\eqref{Eqn_Qualification}\text{ holds}. \tag{FOC}\label{Eqn_FOC}
\end{gather*}

Second-order condition (SOC):
\begin{equation*}
\forall\,\overline{\varepsilon}>0,\ \exists\, \varepsilon \in (0,\overline{\varepsilon}), \text{ such that } \MR_{-}(Q;\overline{p}-\varepsilon\mathbb{I}_{Q})\geq 0.
\tag{SOC}
\label{Eqn_SOC}
\end{equation*}

Our main result, \Cref{Thm_MR}, is simply a collection of these necessary conditions.

\begin{theorem}
\label{Thm_MR}
For every optimal mechanism and every closed set $Q\subseteq \mathcal{Q}$ with $\mathbf{0} \notin Q$, both \emph{(\hyperref[Eqn_FOC]{FOCs})} and \eqref{Eqn_SOC} hold.
\end{theorem}

We provide a simple sufficient condition for the qualification \eqref{Eqn_Qualification}.
For any nonzero vector $\mathbf{n}\in \mathbb{R}^{N}$, denote by $\mathcal{S}(\mathbf{n})$ the \emph{exposed face} of $\mathcal{X}$ with $\mathbf{n}$ as its outward normal:
\begin{equation*}
\mathcal{S}(\mathbf{n})=\{\mathbf{x}\in \mathcal{X}:(\mathbf{x}-\hat{\mathbf{x}})\cdot \mathbf{n}\geq 0 \quad \forall\hat{\mathbf{x}}\in \mathcal{X}\}.
\end{equation*}
The (signed) distance between $\mathcal{S}(\mathbf{n})$ and $\mathbf{0}$ equals $(\mathbf{x}\cdot \mathbf{n})/\lVert \mathbf{n}\rVert$ for any $\mathbf{x}\in \mathcal{S}(\mathbf{n})$ and is denoted by $s(\mathbf{n})$. Denote by $\overline{\bd}(\mathcal{X})$ and $\underline{\bd}(\mathcal{X})$ the \emph{top boundary} and the \emph{bottom boundary} of $\mathcal{X}$, respectively:
\begin{equation*}
\begin{aligned}
\overline{\bd}(\mathcal{X}) & =\bigcup_{s(\mathbf{n})>0}\mathcal{S}(\mathbf{n}), \\
\underline{\bd}(\mathcal{X}) & =\bigcup_{s(\mathbf{n})< 0}\mathcal{S}(\mathbf{n}).
\end{aligned}
\end{equation*}
Note that $\overline{\bd}(\mathcal{X}) \cap \underline{\bd}(\mathcal{X})$ may be nonempty.

\begin{corollary}
\label{Cor_MR}
Suppose that $\mu (\underline{\bd}(\mathcal{X}))=0$, which holds, in particular, when $\mathbf{0}\in \mathcal{X}$. Then for every optimal mechanism and every closed set $Q \subseteq \mathcal{Q}$ with $\mathbf{0}\notin Q$, $\MR (Q)=0$.
\end{corollary}

Recall that $\mathcal{X}$ is convex. If $\mathbf{0}\in \mathcal{X}$, then there must be $\mathbf{x}\cdot \mathbf{n}_{\mathbf{x}}\geq 0$ for any boundary point $\mathbf{x}$, which implies $\underline{\bd}(\mathcal{X})=\varnothing$. In this case, we can conveniently use $\MR (\mathbf{q})=0$ for any nonzero bundle $\mathbf{q}$. Section \ref{Sec_Additive} and \ref{Sec_Applications} will mostly focus on this special case.

\section{Two-Product Monopoly}
\label{Sec_Additive}

This section applies \Cref{Thm_MR} to the leading case where $\mathcal{X}=\mathcal{Q}=[0,1]^{2}$. That is, the seller has two products and a buyer needs at most one unit of each product. Products have \emph{additive} values, meaning that a type-$\mathbf{x}$ buyer's utility from consuming the grand bundle $\mathbf{1}$ is $\mathbf{x}\cdot \mathbf{1}=x_{1}+x_{2}$.
In this setting, the seller's total revenue simplifies to
\begin{equation*}
\int_{\mathcal{X}}u\dd \mu =\int^{1}_{0}u(x,1)f(x,1)\dd x+\int^{1}_{0}u(1,x)f(1,x)\dd x-\int_{\mathcal{X}}u(\mathbf{x})\phi (\mathbf{x})\dd \mathbf{x}.
\end{equation*}

We introduce some definitions on distributions and bundles.

\begin{definition}
The density function $f$ is \emph{strictly regular and symmetric (SRS)} if: (1) $f$ is strictly regular, that is, $\phi(\mathbf{x})>0$ for all $\mathbf{x}\in \Int (\mathcal{X})$; and (2) $f$ is symmetric, that is, $f(x_{1},x_{2})=f(x_{2},x_{1})$ for all $\mathbf{x}$.\footnote{
The condition $\phi (\mathbf{x})\geq 0$ is also called the hazard rate condition in the literature; see, for instance, \citet{mcafee1988multidimensional,manelli2006bundling,hart2010revenue,pavlov2011optimal,pavlov2011property,menicucci2015optimality,tang2017optimal,giannakopoulos2018selling}, and \citet{bikhchandani2022selling}. Here we require a strict inequality to ensure that our results hold for all optimal mechanisms and that the $\zeta$ function introduced later is well-defined.\label{Ftn_Reg}
}
\end{definition}

Throughout \Cref{Sec_Additive,Sec_SubsComp} and their associated Appendices, we assume that $f$ is SRS and, without loss of generality, focus on symmetric mechanisms. To avoid repetition, we use ``every optimal mechanism'' to represent ``every symmetric optimal mechanism'' in these sections.

%We also classify bundles based on their demand sets.

\begin{definition}
A bundle $\mathbf{q}\neq \mathbf{0}$ is called \emph{active} if, for every $\delta > 0$, $\mathcal{D}(\overline{B}_{\delta}(\mathbf{q}))$ has positive measure with respect to $f$, where $\overline{B}_{\delta}(\mathbf{q})$ denotes the closure of $B_{\delta}(\mathbf{q})$; otherwise, $\mathbf{q}$ is \emph{inactive}. 
An active bundle $\mathbf{q}$ is \emph{pooling} if $\mathcal{D}(\mathbf{q})$ has positive measure with respect to $f$; otherwise, it is \emph{separating}.
\end{definition}
Intuitively, a bundle $\mathbf{q}$ is active if, no matter how small a neighborhood we take around it, a nonnegligible set of buyers chooses bundles within that neighborhood. Moreover, if $\mathbf{q}$ itself is purchased by a positive mass of buyers, we call it pooling; if buyers are continuously distributed in a neighborhood of $\mathbf{q}$, we call it separating.
A mechanism involving separating bundles can be implemented only with an \emph{infinite} menu of prices. Accordingly, the indirect utility function $u$ cannot be affine in a neighborhood of the demand set of a separating bundle.

%\zq{The remainder of this section is organized as follows.}
%\Cref{Sec_Additive_Menu} takes the type distribution as given and derives bounds on the number of active bundles, or menu size, in optimal mechanisms. \Cref{Sec_Sufficiency} takes the mechanism as given and derives necessary and sufficient conditions for its optimality using the same variational approach as that used to establish \Cref{Thm_MR}. Finally, \Cref{Sec_Additive_Bundle} considers commonly used bundling strategies and shows that our results yield simple optimality conditions for them. \Cref{Sec_Additive_Summary} summarizes with some remarks on the literature.

%The only distributional assumption we need is the following.

%\begin{definition}
%A density function $f$ is \emph{strictly regular} if
%\begin{equation*}
%\phi(\mathbf{x})=\mathbf{x}\cdot \nabla f(\mathbf{x})+(N+1)f(\mathbf{x})>0 \text{ for almost all } \mathbf{x}\in \mathcal{X},
%\end{equation*}
%and $f$ is \emph{symmetric} if
%\begin{equation*}
%f(\mathbf{x})=f(\hat{\mathbf{x}})\text{ for any }\hat{\mathbf{x}} \text{ that is a rearrangement of }\mathbf{x}.
%\end{equation*}
%\end{definition}

%By \Cref{Lem_MR},
%\begin{equation}
%\label{Eqn_2D_MR}
%\MR (Q)=\int_{\mathcal{D}(Q)}\phi(\mathbf{x})\dd \mathbf{x}-\int_{\mathcal{D}(Q)\cap \bd(\mathcal{X})}f(\mathbf{x})(\mathbf{x}\cdot\mathbf{n}_{\mathbf{x}})\dd \sigma (\mathbf{x}).
%\end{equation}

Under SRS, \Cref{Thm_MR} yields the following characterization of the buyer's indirect utility function, originally established in \citet[Lemma 2]{mcafee1988multidimensional}.

\begin{lemma}
\label{Lem_MM}
Suppose that $f$ is SRS. In every optimal mechanism, if $(q_{1},q_{2})$ is active, then $q_{1}=1$ or $q_{2}=1$, and $(q_{2},q_{1})$ is also active. Moreover,
\begin{equation}
\label{Eqn_MM}
u(\mathbf{x})=\max\{\overline{u}(x_{1})-(1-x_{2}),\overline{u}(x_{2})-(1-x_{1}),0\},
\end{equation}
where $\overline{u}(\cdot)=u(\cdot,1)=u(1,\cdot)$.
\end{lemma}

Equation \eqref{Eqn_MM} reveals substantial information on the structure of an optimal mechanism. Denote by $[a^{q},b^{q}]$ the interval of types such that all buyers in $[a^{q},b^{q}]\times\{1\}$ choose bundle $(q,1)$. Formally,
\begin{equation}
\label{Eqn_Endpoint}
[a^{q},b^{q}]=\{x\in [0,1]:\overline{u}_{-}'(x)\leq q\leq \overline{u}_{+}'(x)\}.
\end{equation}
When $q<1$, the demand set of $(q,1)$ is fully determined by $a^{q}$, $b^{q}$, and $\overline{u}$:
\begin{equation*}
\mathcal{D}(q,1)=\{\mathbf{x}\in \mathcal{X}: a^{q}\leq x_{1}\leq b^{q},\ x_{2}\geq \max (1-\overline{u}(x_{1}),x_{1})\}.
\end{equation*}
One can interpret $[a^{q},b^{q}]$ as the ``projection'' of $\mathcal{D}(q,1)$ to the boundary $[0,1]\times \{1\}$. Motivated by this, we can also project $\MR (q,1)$ to the same boundary:
\begin{equation*}
\MR (q,1)=-\mu (\mathcal{D}(q,1)) =\int_{\mathcal{D}(q,1)}\phi(\mathbf{x})\dd \mathbf{x}-\int_{\mathcal{D}(q,1)\cap \bd(\mathcal{X})}f(\mathbf{x})(\mathbf{x}\cdot\mathbf{n}_{\mathbf{x}})\dd \sigma (\mathbf{x})=\int^{b^{q}}_{a^{q}}\Phi (x_{1})\dd x_{1},
\end{equation*}
where $\Phi (x_{1})$ is the \emph{cumulative MR} on $\{x_{1}\}\times [x_{1},1]$:
\begin{equation}
\label{Eqn_Phi}
\Phi (x_{1})=\int^{1}_{\max(1-\overline{u}(x_{1}),x_{1})}\phi(x_{1},x_{2})\dd x_{2}-f(x_{1},1).
\end{equation}

\Cref{Prop_Additive} summarizes the implications of \Cref{Thm_MR} on the properties of $\Phi$.

\begin{proposition}
\label{Prop_Additive}
Suppose that $f$ is SRS. Then in every optimal mechanism:
\begin{enumerate}
\item \label{Prop_Additive_2}
If $(q,1)$ is pooling, then
\begin{equation}
\label{Eqn_IntPhi}
\int_{a^{q}}^{b^{q}}\Phi(x_{1})\dd x_{1}=0.
\end{equation}
Moreover, $\Phi(a^{q})\geq 0$ unless $q=0$, and $\Phi(b^{q})\geq 0$ unless $q=1$.
\item \label{Prop_Additive_3}
If $(q,1)$ is separating, then $a^{q}=b^{q}$ and $\Phi(a^{q})=0$.
\end{enumerate}
\end{proposition}

An immediate consequence of \Cref{Prop_Additive} is that the grand bundle $\mathbf{1}$ is pooling. Otherwise, either $\mathbf{1}$ is separating, or $\mathbf{1}$ is inactive and there exists an active bundle $(q,1)$ with $q<1$ and $b^{q}=1$. Both cases require $\Phi(1)=-f(1,1)=0$, a contradiction.\footnote{Note that SRS implies $f(x,1)>0$ and $f(1,x)>0$ for every $x\in (0,1]$: the function $\tau \mapsto \tau^{3}f(\tau x,\tau)$ is nondecreasing because its derivative is $\tau^{2}\phi (\tau x,\tau)>0$, while interior positivity gives $f(\tau x,\tau)>0$ for $\tau \in(0,1)$, and hence $f(x,1)>0$. \label{Ftn_Top_density}}

\begin{corollary}
\label{Cor_GrandBundle}
In every optimal mechanism, the grand bundle $\mathbf{1}$ is pooling.
\end{corollary}

\subsection{Type-dependent Posted Price}

\Cref{Prop_Additive}, together with the form of $\Phi$, motivates the definition of a function $\zeta:[0,1]\to \mathbb{R}$, such that
\begin{equation}
\label{Eqn_Zeta}
\zeta(x_{1})=\sup\left\{\zeta\in[0,1]:\int^{1}_{\zeta}\phi(\mathbf{x})\dd x_{2}-f(x_{1},1)\geq 0\right\},\quad \forall\,x_{1}\in(0,1),
\end{equation}
with the convention $\sup\varnothing=0$. Under SRS, $\zeta$ is continuous on $(0,1)$. Whenever the relevant one-sided limits exist, we extend $\zeta$ to $x_{1}=0$ and $1$ by these limits. By construction, for almost all $x\in [0,1]$, the sign of $\Phi (x)$ is the same as the sign of $\zeta (x)-\max(1-\overline{u}(x),x)$, and $\zeta<1$ since $f(x_{1},1)>0$.\footnote{
\label{Ftn_Zeta_Endpoints}
Moreover, whenever $\zeta(0)$ exists, an optimal mechanism with a pooling bundle $(q,1)$ satisfying $q>0$ and $a^{q}=0$ must have $1-\overline{u}(0)\leq \zeta(0)$. Otherwise, by continuity, discounting $(\hat{q},1)$ and $(1,\hat{q})$ for a fixed $0\leq \hat{q}<q$ yields strictly negative marginal revenue for every sufficiently small discount, contradicting \eqref{Eqn_SOC}.
}

We shall refer to $\zeta$ as a \emph{type-dependent posted price}, although $\zeta(x)$ is not itself a price. To see the connection, consider a single-product monopoly with $\mathcal{X}=[0,1]$. Under a regular distribution, the optimal posted price $p$ satisfies
\begin{equation*}
\text{Virtual value} =p-\frac{1-F(p)}{f(p)}=\frac{1}{f(p)}\left[f(1)-\int^{1}_{p}\phi(x)\dd x\right]=0.
\end{equation*}
Thus, up to multiplication by the density, $\zeta$ is analogous to the optimal posted price in a single-product monopoly. The only difference is that $\zeta$ is computed for each $x_{1}$ and is therefore type-dependent. Consequently, we may need to ``iron'' $\zeta$ across $x_{1}$ subject to the convexity constraint.

As an alternative interpretation, when $\zeta>0$, \eqref{Eqn_IntPhi} can be rewritten as
\begin{equation*}
\int^{b^{q}}_{a^{q}}\int^{1}_{\max(1-\overline{u}(x_{1}),x_{1})}\phi(\mathbf{x})\dd x_{2}\dd x_{1}=\int^{b^{q}}_{a^{q}}\int^{1}_{\zeta(x_{1})}\phi(\mathbf{x})\dd x_{2}\dd x_{1}.
\end{equation*}
Since $\phi$ can be interpreted as the ``density'' of $\MR$, \eqref{Eqn_FOC} amounts to a cancellation of ``mass'' between $\zeta(x)$ and $\max(1-\overline{u}(x),x)$ over $x\in[a^{q},b^{q}]$. In this sense, $\zeta$ can be viewed as a \emph{pointwise solution to cumulative $\MR=0$}, or equivalently, as an \emph{unconstrained maximizer of total revenue}. If we let $g(x_{1})=\int^{1}_{0}f(x_{1},x_{2})\dd x_{2}$ be the marginal density of $f$, then $\zeta >0$ if and only if
\begin{equation*}
\int^{1}_{0}\phi (\mathbf{x})\dd x_{2}-f(x_{1},1)=x_{1}g'(x_{1})+2g(x_{1})>0,
\end{equation*}
which is just the condition for one-dimensional strict regularity.

In \Cref{Fig_Prop_2D}, the red curve represents $1-\overline{u}(x_{1})$ and its symmetric counterpart $1-\overline{u}(x_{2})$, and the green curve represents $\zeta$. Demand sets are separated by dashed lines. The areas colored in blue or red capture the extent to which $\max (1-\overline{u}(x_{1}),x_{1})$ deviates from $\zeta (x_{1})$. The integral of $\phi$ on the blue area is the ``surplus'' of MR, while the integral of $\phi$ on the red area is the ``deficit'' of MR. Panel (a) depicts a case in which a pooling bundle $(q,1)$ violates \eqref{Eqn_IntPhi} in part \ref{Prop_Additive_2} of \Cref{Prop_Additive}. Panel (b) depicts a case in which the left endpoint of a pooling bundle $a^{q}$ violates $\Phi (a^{q})\geq 0$ in part \ref{Prop_Additive_2} of \Cref{Prop_Additive}. Panel (c) illustrates part \ref{Prop_Additive_3} of \Cref{Prop_Additive}, namely, if all bundles in $[q,\hat{q}]\times \{1\}$ are separating, then $\max (1-\overline{u}(x_{1}),x_{1})$ should coincide with $\zeta$ on $[a^{q},a^{\hat{q}}]$.

\begin{figure}[ht]
\centering
\begin{tikzpicture}[baseline={(0,0)},x=4.5cm,y=4.5cm]
%\fill [blue!10,domain=0:0.1] plot ({\x}, {-1*(\x-0.2)^2+0.91}) -- (0,0.9);
\fill [red!10,domain=0.35:0.5]  plot ({\x}, {-1*(\x-0.2)^2+0.91}) -- (0.65,0.73) -- (0.65,0.65) -- (0.6,0.6) -- (0.35,0.8);
\draw (0,0) node [below left] {0} -- (1,0) node [below] {1} -- (1,1) -- (0,1) node [left] {1} -- (0,0);
\draw [dotted] (0,0) -- (1,1);
\draw (0,0.87) node [left] {$\zeta$};
\draw [domain=0:0.2, thick, green] plot ({\x}, {-1*(\x-0.2)^2+0.91});
\draw [domain=0.2:0.5, thick, green] plot ({\x}, {-1*(\x-0.2)^2+0.91});
\draw [thick, green] (0.5,0.82) -- (1,0.52);
\draw [thick, red] (0,0.8) -- (0.35,0.8) -- (0.6,0.6) -- (0.8,0.35) -- (0.8,0);
\draw [dashed] (0.35,0.8) -- (0.35,1);
\draw [dotted] (0.35,0.8) -- (0.35,0) node [below] {$a^{q}$};
\draw [dashed] (0.65,0.65) -- (0.65,1);
\draw [dotted] (0.65,0.65) -- (0.65,0) node [below] {$b^{q}$};
\draw (0.5,0.85) node {$(q,1)$};
\draw (0.5,-0.2) node {\footnotesize (a) $(q,1)$ violates \eqref{Eqn_FOC}};
\end{tikzpicture}%
\begin{tikzpicture}[baseline={(0,0)},x=4.5cm,y=4.5cm]
\fill [blue!10,domain=0.3:0.4] plot ({\x}, {0.5*(\x-0.3)^2+0.6}) -- (0.4,0.76) -- (0.3,0.8);
%\fill [red!10,domain=0.35:0.5]  plot ({\x}, {-1*(\x-0.2)^2+0.91}) -- (0.65,0.73) -- (0.65,0.65) -- (0.6,0.6) -- (0.35,0.8);
\draw (0,0) node [below left] {0} -- (1,0) node [below] {1} -- (1,1) -- (0,1) node [left] {1} -- (0,0);
\draw [dotted] (0,0) -- (1,1);
\draw (0,0.645) node [left] {$\zeta$};
\draw [domain=0:1, thick, green] plot ({\x}, {0.5*(\x-0.3)^2+0.6});
%\draw [domain=0.2:0.5, thick, green] plot ({\x}, {-1*(\x-0.2)^2+0.91});
%\draw [thick, green] (0.5,0.82) -- (1,0.52);
\draw [thick, red] (0,0.8) -- (0.35,0.8) -- (0.6,0.6) -- (0.8,0.35) -- (0.8,0);
\draw [thick, red] (0.3,0.8) -- (0.4,0.76);
\draw [dashed] (0.35,0.8) -- (0.35,1);
\draw [dotted] (0.35,0.8) -- (0.35,0) node [below] {$a^{q}$};
\draw [dashed] (0.65,0.65) -- (0.65,1);
\draw [dotted] (0.65,0.65) -- (0.65,0) node [below] {$b^{q}$};
\draw [dotted] (0.3,0.8) -- (0.3,1);
\draw [dotted] (0.4,0.76) -- (0.4,1);
\draw (0.5,0.85) node {$(q,1)$};
\draw (0.5,-0.2) node {\footnotesize (b) $a^{q}$ violates \eqref{Eqn_SOC}};
\end{tikzpicture}%
\begin{tikzpicture}[baseline={(0,0)},x=4.5cm,y=4.5cm]
%\fill [blue!10,domain=0:0.1] plot ({\x}, {-1*(\x-0.2)^2+0.91}) -- (0,0.9);
%\fill [red!10,domain=0.35:0.5]  plot ({\x}, {-1*(\x-0.2)^2+0.91}) -- (0.65,0.73) -- (0.65,0.65) -- (0.6,0.6) -- (0.35,0.8);
\draw (0,0) node [below left] {0} -- (1,0) node [below] {1} -- (1,1) -- (0,1) node [left] {1} -- (0,0);
\draw [dotted] (0,0) -- (1,1);
\draw (0,0.87) node [left] {$\zeta$};
\draw [domain=0:0.35, thick, green] plot ({\x}, {-1*(\x-0.2)^2+0.91});
\draw [domain=0.35:0.5, thick, red] plot ({\x}, {-1*(\x-0.2)^2+0.91});
\draw [domain=0.35:0.5, thick, red] plot ({-1*(\x-0.2)^2+0.91}, {\x});
\draw [thick, red] (0.5,0.82) -- (0.65, 0.73) -- (0.73,0.65) -- (0.82,0.5);
\draw [thick, green] (0.65, 0.73) -- (1,0.52);
\draw [thick, red] (0,0.8875) -- (0.35,0.8875);
\draw [thick, red] (0.8875,0) -- (0.8875,0.35);
\draw [dashed] (0.35,0.8875) -- (0.35,1);
\draw [dotted] (0.35,0.8875) -- (0.35,0);
\draw (0.35,-0.12) node [anchor=south] {$a^{q}$};
\draw [dashed] (0.65,0.73) -- (0.65,1);
\draw [dotted] (0.65,0.73) -- (0.65,0);
\draw (0.65,-0.12) node [anchor=south] {$a^{\hat{q}}$};
\draw (0.5,0.85) node {$(q,1)$};
\draw (0.5,-0.2) node {\footnotesize (c) \eqref{Eqn_FOC} implies $\overline{u}=1-\zeta $};
\end{tikzpicture}%
\caption{Using $\zeta$ to verify \eqref{Eqn_FOC} and \eqref{Eqn_SOC}.}
\label{Fig_Prop_2D}
\end{figure}

\paragraph{An elasticity rule.}

When $x_{1}$ and $x_{2}$ are i.i.d. with marginal p.d.f. $g$ and c.d.f. $G$, \eqref{Eqn_Zeta} reduces to
\begin{equation}
\label{Eqn_Zeta_iid}
\frac{\zeta(x)g(\zeta(x))}{1-G(\zeta(x))}=2+\frac{xg'(x)}{g(x)}.
\end{equation}
Let $h(x)=xg'(x)/g(x)$. Then $f$ (or equivalently $g$) is strictly regular if $h>-3/2$. Moreover, when both $h$ and $\zeta$ are differentiable at $x$, the sign of $\zeta'(x)$ is the same as the sign of $h'(x)$. 

Equation \eqref{Eqn_Zeta_iid} gives $\zeta$ a natural interpretation in terms of elasticity. Since $1-G(x)$ is the demand curve associated with the marginal distribution, its (absolute) price elasticity is
\begin{equation*}
\Elas (x)=\frac{xg(x)}{1-G(x)}.
\end{equation*}
Then \eqref{Eqn_Zeta_iid} can be written as
\begin{equation}
\label{Eqn_Zeta_elasticity}
\Elas (\zeta(x))=2+h(x).
\end{equation}
For comparison, in a single-product monopoly, the standard elasticity rule is
\[
\Elas (p)=1.
\]
Equation \eqref{Eqn_Zeta_elasticity} can therefore be viewed as a multiproduct analogue: the type-dependent posted price $\zeta$ is characterized by an elasticity condition, with the target elasticity adjusted by $h(x)$, which captures the local shape of the valuation density.

\subsection{Optimal Mechanisms}
\label{Sec_Additive_Menu}

This section presents several corollaries and examples illustrating how $\zeta$ can be used to characterize optimal mechanisms and menu size. Our recipe is simple:
\begin{itemize}
    \item Compute the type-dependent posted price $\zeta$.
    \item If $\zeta$ satisfies the properties required of a feasible $1-\overline{u}$, let $1-\overline{u}$ ``follow'' $\zeta$.
    \item Otherwise, ``iron'' $\zeta$ subject to \eqref{Eqn_FOC} and \eqref{Eqn_SOC}.
\end{itemize}
In Examples \ref{Exp_PowerPositive}, \ref{Exp_Pareto}, and \ref{Exp_Gamma}, we assume that $x_{1}$ and $x_{2}$ are i.i.d., while we allow $x_{1}$ and $x_{2}$ to be correlated in the other examples. We also assume that $\zeta (0)$ is always well-defined. Detailed computations are provided in \Cref{Sec_Examples}.

\paragraph{Nondecreasing or convex $\zeta$.}

We start with two environments where any optimal mechanism has only finitely many active bundles, and all of them are pooling.

When $\zeta$ is nondecreasing and crosses the 45-degree line exactly once, any optimal mechanism is deterministic, that is, the set of active bundles is a subset of $\{\mathbf{e}_{1},\mathbf{e}_{2},\mathbf{1}\}$.

\begin{corollary}
\label{Cor_Inc}
Suppose that $\zeta$ is nondecreasing and crosses the 45-degree line exactly once. Then in every optimal mechanism, there exists $a^{1}<1$ such that:
\begin{enumerate}
\item $\overline{u}'(x)=0$ for $x\in (0,a^{1})$,
\item $\overline{u}'(x)=1$ for $x\in (a^{1},1)$.
\end{enumerate}
That is, any optimal mechanism is deterministic.
\end{corollary}

When $\zeta$ is convex, there are at most three active bundles in any optimal mechanism.

\begin{corollary}
\label{Cor_Cvx}
Suppose that $\zeta$ is convex on $[0,1]$. Then in every optimal mechanism, there exist $a^{1}<1$ and $q<1$ such that:
\begin{enumerate}
\item $\overline{u}'(x)=q$ for $x\in (0,a^{1})$,
\item $\overline{u}'(x)=1$ for $x\in (a^{1},1)$.
\end{enumerate}
That is, any optimal mechanism contains at most three nonzero menus.
\end{corollary}

\Cref{Fig_CvxInc} illustrates Corollaries \ref{Cor_Inc} and \ref{Cor_Cvx}. In each panel of the figure, demand sets are separated by dashed lines. The top-right region of the red curve consists of demand sets of all nonzero bundles, while the bottom-left region of the red curve is the demand set of the zero bundle. Panels (a) and (b) correspond to Corollaries \ref{Cor_Inc} and \ref{Cor_Cvx} with $a^{1}>0$. Panel (c) depicts the special case of pure bundling, which may occur regardless of the shape of $\zeta$.

\begin{figure}[ht]
\centering
\begin{tikzpicture}[baseline={(0,0)},x=4.5cm,y=4.5cm]
\fill [blue!10,domain=0:0.32] plot ({\x}, {0.5*(\x)^2+0.45}) -- (0,0.5);
\fill [red!10,domain=0.32:0.4]  plot ({\x}, {0.5*(\x)^2+0.45}) -- (0.4,0.5) -- (0.32,0.5);
\fill [red!10,domain=0.4:0.68] plot ({\x}, {0.5*(\x)^2+0.45}) -- (0.45,0.45) -- (0.4,0.5) -- (0.4,0.53);
\fill [blue!10,domain=0.68:1]  plot ({\x}, {0.5*(\x)^2+0.45}) -- (1,1) -- (0.68,0.68);
\draw (0,0) node [below left] {0} -- (1,0) node [below] {1} -- (1,1) -- (0,1) node [left] {1} -- (0,0);
\draw [dotted] (0,0) -- (1,1);
\draw [domain=0:1, thick, green, smooth] plot ({\x}, {0.5*(\x)^2+0.45});
\draw (0,0.45) node [left] {$\zeta$};
\draw [thick, red] (0,0.5) -- (0.4,0.5) -- (0.5,0.4) -- (0.5,0);
\draw (0.8,0.8) node {$\mathbf{1}$};
\draw (0.2,0.8) node {$\mathbf{e}_{2}$};
\draw (0.8,0.2) node {$\mathbf{e}_{1}$};
\draw [dotted] (0.4,0) node [below] {$a^{1}$} -- (0.4,0.5);
\draw [dashed] (0.4,0.5) -- (0.4,1);
\draw [dashed] (0.5,0.4) -- (1,0.4);
\draw (0.5,-0.2) node {\footnotesize (a) Nondecreasing $\zeta$};
\end{tikzpicture}
\begin{tikzpicture}[baseline={(0,0)},x=4.5cm,y=4.5cm]
\fill [red!10,domain=0:0.077] plot ({\x}, {0.25*(\x-1)^8+0.63}) -- (0,0.8);
\fill [blue!10,domain=0.077:0.316]  plot ({\x}, {0.25*(\x-1)^8+0.63}) -- (0.077,0.7615);
\fill [red!10,domain=0.316:0.63] plot ({\x}, {0.25*(\x-1)^8+0.63}) -- (0.5,0.5) -- (0.4,0.6) -- (0.316,0.642);
\fill [blue!10,domain=0.63:1]  plot ({\x}, {0.25*(\x-1)^8+0.63}) -- (1,1);
\draw (0,0) node [below left] {0} -- (1,0) node [below] {1} -- (1,1) -- (0,1) node [left] {1} -- (0,0);
\draw [dotted] (0,0) -- (1,1);
\draw [domain=0:1, thick, green, smooth] plot ({\x}, {0.25*(\x-1)^8+0.63});
\draw (0,0.88) node [left] {$\zeta$};
\draw [thick, red] (0,0.8) -- (0.4,0.6) -- (0.6,0.4) -- (0.8,0);
\draw (0.8,0.8) node {$\mathbf{1}$};
\draw (0.2,0.85) node {$(q,1)$};
\draw (0.85,0.2) node {$(1,q)$};
\draw [dotted] (0.4,0) node [below] {$a^{1}$} -- (0.4,0.6);
\draw [dashed] (0.4,0.6) -- (0.4,1);
\draw [dashed] (0.6,0.4) -- (1,0.4);
\draw (0.5,-0.2) node {\footnotesize (b) Convex $\zeta$};
\end{tikzpicture}
\begin{tikzpicture}[baseline={(0,0)},x=4.5cm,y=4.5cm]
\fill [red!10,domain=0:0.5] plot ({\x}, {2*(\x-0.3)^2+0.54}) -- (0.5,0.5) -- (0.3,0.3) -- (0,0.6);
\fill [red!10,domain=0.5:0.7] plot ({\x}, {-2*(\x-0.7)^2+0.7}) -- (0.5,0.5);
\fill [blue!10,domain=0.7:1]  plot ({\x}, {-2*(\x-0.7)^2+0.7}) -- (1,1);
%\fill [red!10,domain=0:0.2] plot ({\x}, {2*(\x-0.3)^2+0.54}) -- (0.2,0.45) -- (0.325,0.325) -- (0.3,0.3) -- (0,0.6);
\draw (0,0) node [below left] {0} -- (1,0) node [below] {1} -- (1,1) -- (0,1) node [left] {1} -- (0,0);
\draw [domain=0:0.5, thick, green] plot ({\x}, {2*(\x-0.3)^2+0.54});
\draw [domain=0.5:1, thick, green] plot ({\x}, {-2*(\x-0.7)^2+0.7});
\draw (0,0.72) node [left] {$\zeta$};
\draw [thick, red] (0,0.6) -- (0.6,0);
\draw [dotted] (0,0) -- (1,1);
\draw (0.8,0.8) node {$\mathbf{1}$};
\draw (0.5,-0.2) node {\footnotesize (c) Pure bundling};
\end{tikzpicture}
\caption{Nondecreasing or convex $\zeta$.}
\label{Fig_CvxInc}
\end{figure}
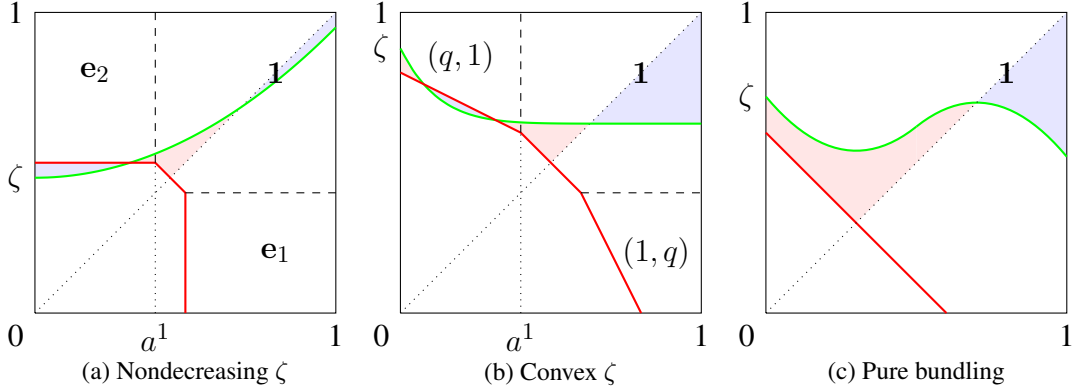

\begin{example}[Power-law]
\label{Exp_PowerPositive}
Let $g(x)=\kappa (x+\theta)^{\eta-1}$, where $\theta\geq 0$, $\eta\geq 1$, and $\kappa $ is a normalization constant. Then $h(x)=(\eta-1)x/(x+\theta)\geq 0$ for all $x\in [0,1]$, so $g$ is strictly regular. It can be proved that $\zeta$ is nondecreasing and crosses the 45-degree line exactly once. In particular, when $\theta=0$, $\zeta$ equals a constant: $\zeta (x)=[(\eta +1)/(2\eta +1)]^{1/\eta}$.
\end{example}

\begin{example}[Truncated Pareto]
\label{Exp_Pareto}
Let $g(x)=\kappa (1+x)^{-\eta}$, where $\eta\in (0,3)$, and $\kappa $ is a normalization constant. Then $h(x)=-\eta x/(x+1)>-3/2$ for all $x\in [0,1]$, so $g$ is strictly regular. It can be proved that $\zeta$ is strictly decreasing and strictly convex.

We perform numerical simulations for this example. The results show that the optimal mechanism requires randomization when $\eta$ is small and becomes pure bundling when $\eta$ is large. The cutoff is approximately 2.80.
\end{example}

\begin{example}[Uniform marginals with correlation]
\label{Exp_UniformCorrelated}
Consider a joint p.d.f. generated from the Farlie-Gumbel-Morgenstern copula \citep{nelsen2006introduction} with uniform marginals: $f(\mathbf{x})=1+3\rho (1-2x_{1})(1-2x_{2})$, where $\rho \in [-1/3,1/3]$ is the correlation coefficient between two dimensions.\footnote{Let $y_{1}$ and $y_{2}$ be the quantiles of two marginal distributions. Then the Farlie-Gumbel-Morgenstern copula can be expressed as $y_{1}y_{2}[1+3\rho (1-y_{1})(1-y_{2})]$, where $\rho$ is Spearman's rank correlation coefficient. In \Cref{Exp_UniformCorrelated}, $\rho$ is identical to the linear correlation coefficient since we consider uniform marginals.} It can be verified that $f$ is strictly regular if and only if $\rho \in [-1/7,1/5]$. Conditional on strict regularity, when $\rho\in [-1/7,0)$, $\zeta$ is strictly decreasing and concave, hence optimal mechanisms are characterized by \Cref{Cor_CcvDec}; when $\rho \in (0,1/5]$, $\zeta$ is strictly increasing and concave, hence optimal mechanisms are characterized by \Cref{Cor_Inc}.
\end{example}

Under the i.i.d. assumption, \citet{manelli2006bundling} and \citet{hart2010revenue} both prove that deterministic mechanisms are optimal when $h$ (or equivalently $\zeta$) is nondecreasing.\footnote{Their result does not depend on the number of crossings between $\zeta$ and the 45-degree line, because they pool the pointwise MR within the demand set of $\mathbf{1}$. Our approach can only rule out suboptimal mechanisms, so it may ``overkill'' in some situations. Their tighter result can be obtained, for instance, by applying the sufficient and necessary conditions in \Cref{Prop_Sufficiency}.} \citet{tang2017optimal} extend the analysis to general rectangular type spaces while maintaining the assumptions of independent distribution and nondecreasing $\zeta$. They find that randomization may arise, but the optimal menu size (excluding the zero bundle) is bounded by five. Note that \Cref{Exp_PowerPositive} with $\theta=0$ is one of the leading examples in \citet{manelli2006bundling} and \citet{tang2017optimal}.

\paragraph{Concave $\zeta$.}

When $\zeta$ is concave, the optimal $\overline{u}$ is characterized by two cutoffs. When $x$ is below the smaller cutoff, $\overline{u}'(x)=0$. When $x$ is above the larger cutoff, $\overline{u}'(x)=1$. When $x$ lies between the two cutoffs, $\overline{u}$ coincides with $1-\zeta$.

\begin{corollary}
\label{Cor_Ccv}
Suppose that $\zeta$ is concave on $[0,1]$. Then in every optimal mechanism, there exist $0\leq b^{0}\leq a^{1}<1$ such that:
\begin{enumerate}
\item $\overline{u}'(x)=0$ for $x\in (0,b^{0})$,
\item $\overline{u}(x)=1-\zeta (x)$ for $x\in (b^{0},a^{1})$,
\item $\overline{u}'(x)=1$ for $x\in (a^{1},1)$.
\end{enumerate}
That is, an optimal mechanism has infinite menu size if $b^{0}<a^{1}$ and $\zeta$ is strictly concave on $(b^{0},a^{1})$, and degenerates to a deterministic mechanism if $b^{0}=a^{1}$.
\end{corollary}

As a special case of \Cref{Cor_Ccv}, when $\zeta$ is nonincreasing and concave, the smaller cutoff $b^{0}$ equals zero, meaning that $\overline{u}(x)=1-\zeta (x)$ for all $x\in [0,a^{1}]$. The optimal $\overline{u}$ ``follows'' $1-\zeta$ until it reaches the smallest point where $\overline{u}_{+}'(x)=1$. This result illustrates the intuition that $\zeta$ is an unconstrained pointwise solution to the condition $\MR=0$. If $1-\zeta$ happens to share the same properties as a feasible $\overline{u}$, then it becomes a natural candidate for the optimal $\overline{u}$. Due to its importance, we say that $f$ (or $\zeta$) is \emph{double-regular} if $\zeta$ is nonincreasing and concave on $[0,1]$.

\begin{corollary}
\label{Cor_CcvDec}
Suppose that $\zeta$ is double-regular. Then in every optimal mechanism, there exists $a^{1}\in [0,1)$ such that:
\begin{enumerate}
\item $\overline{u}(x)=1-\zeta (x)$ for $x\in (0,a^{1})$,
\item $\overline{u}'(x)=1$ for $x\in (a^{1},1)$.
\end{enumerate}
That is, an optimal mechanism has infinite menu size if $a^{1}>0$ and $\zeta$ is strictly concave, and degenerates to pure bundling if $a^{1}=0$.
\end{corollary}

\Cref{Fig_Ccv} graphically illustrates Corollaries \ref{Cor_Ccv} and \ref{Cor_CcvDec}. In panel (a), since $\zeta$ may be strictly increasing when $x$ is small, it is possible that $1-\overline{u}$ cannot follow $\zeta$ when $x<b^{0}$. In this case, we may need to ``iron'' MR, so pooling arises naturally in optimal mechanisms. In panel (b), $b^{0}=a^{1}$, $1-\overline{u}$ does not follow $\zeta$ on the entire $[0,1]$, meaning that all active bundles are pooling. In panel (c), $\zeta$ is double-regular, so $1-\overline{u}=\zeta$ on $[0,a^{1}]$.

\begin{figure}[ht]
\centering
\begin{tikzpicture}[baseline={(0,0)},x=4.5cm,y=4.5cm]
\fill [blue!10,domain=0:0.1] plot ({\x}, {-1*(\x-0.2)^2+0.91}) -- (0,0.9);
\fill [red!10,domain=0.1:0.3]  plot ({\x}, {-1*(\x-0.2)^2+0.91}) -- (0.1,0.9);
\fill [red!10] plot (0.5,0.82) -- (0.7,0.7) -- (0.66,0.66);
\fill [blue!10] (0.7,0.7) -- (1,0.52) -- (1,1);
\draw (0,0) node [below left] {0} -- (1,0) node [below] {1} -- (1,1) -- (0,1) node [left] {1} -- (0,0);
\draw [dotted] (0,0) -- (1,1);
\draw (0,0.87) node [left] {$\zeta$};
\draw [domain=0:0.2, thick, green] plot ({\x}, {-1*(\x-0.2)^2+0.91});
\draw [domain=0.2:0.3, thick, green] plot ({\x}, {-1*(\x-0.2)^2+0.91});
\draw [thick, green] (0.5,0.82) -- (1,0.52);
\draw [thick, red] (0,0.9) -- (0.3,0.9);
\draw [domain=0.3:0.5, thick, red] plot ({\x}, {-1*(\x-0.2)^2+0.91}) -- (0.82,0.5);
\draw [domain=0.3:0.5, thick, red] plot ({-1*(\x-0.2)^2+0.91}, {\x});
\draw [thick, red] (0.9,0) -- (0.9,0.3);
\draw (0.8,0.8) node {$\mathbf{1}$};
\draw (0.15,0.95) node {$\mathbf{e}_{2}$};
%\draw (0.4,0.95) node {$(q,1)$};
\draw (0.92,0.15) node {$\mathbf{e}_{1}$};
%\draw (0.92,0.4) node {$(1,q)$};
\draw (0.5,-0.2) node {\footnotesize (a) $b^{0}<a^{1}$};
\draw [dotted] (0.3,0) node [below] {$b^{0}$} -- (0.3,0.9);
\draw [dashed] (0.3,0.9) -- (0.3,1);
\draw [dashed] (0.9,0.3) -- (1,0.3);
\draw [dotted] (0.5,0) node [below] {$a^{1}$} -- (0.5,0.82);
\draw [dashed] (0.5,0.82) -- (0.5,1);
\draw [dashed] (0.82,0.5) -- (1,0.5);
\end{tikzpicture}
\begin{tikzpicture}[baseline={(0,0)},x=4.5cm,y=4.5cm]
\fill [blue!10,domain=0:0.1] plot ({\x}, {-1*(\x-0.2)^2+0.91}) -- (0,0.9);
\fill [red!10,domain=0.1:0.5]  plot ({\x}, {-1*(\x-0.2)^2+0.91}) -- (0.7,0.7) -- (0.55,0.55) -- (0.2,0.9) -- (0.1,0.9);
\fill [blue!10] (0.7,0.7) -- (1,0.52) -- (1,1);
\draw (0,0) node [below left] {0} -- (1,0) node [below] {1} -- (1,1) -- (0,1) node [left] {1} -- (0,0);
\draw [dotted] (0,0) -- (1,1);
\draw (0,0.87) node [left] {$\zeta$};
\draw [domain=0:0.2, thick, green] plot ({\x}, {-1*(\x-0.2)^2+0.91});
\draw [domain=0.2:0.5, thick, green] plot ({\x}, {-1*(\x-0.2)^2+0.91});
\draw [thick, green] (0.5,0.82) -- (1,0.52);
\draw [thick, red] (0,0.9) -- (0.2,0.9) -- (0.9,0.2) -- (0.9,0);
\draw (0.8,0.8) node {$\mathbf{1}$};
\draw (0.1,0.95) node {$\mathbf{e}_{2}$};
\draw (0.92,0.1) node {$\mathbf{e}_{1}$};
\draw (0.5,-0.2) node {\footnotesize (b) $b^{0}=a^{1}$};
\draw [dotted] (0.2,0) node [below] {$b^{0}$} -- (0.2,0.9);
\draw [dashed] (0.2,0.9) -- (0.2,1);
\draw [dashed] (0.9,0.2) -- (1,0.2);
\end{tikzpicture}
\begin{tikzpicture}[baseline={(0,0)},x=4.5cm,y=4.5cm]
\fill [red!10,domain=0.3:0.64] plot ({\x}, {-0.5*(\x)^2+0.845}) -- (0.55,0.55);
\fill [blue!10,domain=0.64:1]  plot ({\x}, {-0.5*(\x)^2+0.845}) -- (1,1) -- (0.64,0.64);
\draw (0,0) node [below left] {0} -- (1,0) node [below] {1} -- (1,1) -- (0,1) node [left] {1} -- (0,0);
\draw [dotted] (0,0) -- (1,1);
\draw [domain=0:0.3, thick, red, smooth] plot ({\x}, {-0.5*(\x)^2+0.845}) -- (0.55,0.55);
\draw [domain=0:0.3, thick, red, smooth] plot ({-0.5*(\x)^2+0.845},{\x}) -- (0.55,0.55);
\draw [domain=0.3:1, thick, green, smooth] plot ({\x}, {-0.5*(\x)^2+0.845});
\draw (0,0.845) node [left] {$\zeta$};
\draw (0.8,0.8) node {$\mathbf{1}$};
%\draw (0.15,0.92) node {$(q,1)$};
%\draw (0.92,0.15) node {$(1,q)$};
\draw [dotted] (0.3,0) node [below] {$a^{1}$} -- (0.3,0.8);
\draw [dashed] (0.3,0.8) -- (0.3,1);
\draw [dashed] (0.8,0.3) -- (1,0.3);
\draw (0.5,-0.2) node {\footnotesize (c) $b^{0}=0$};
\end{tikzpicture}
\caption{Concave $\zeta$.}
\label{Fig_Ccv}
\end{figure}
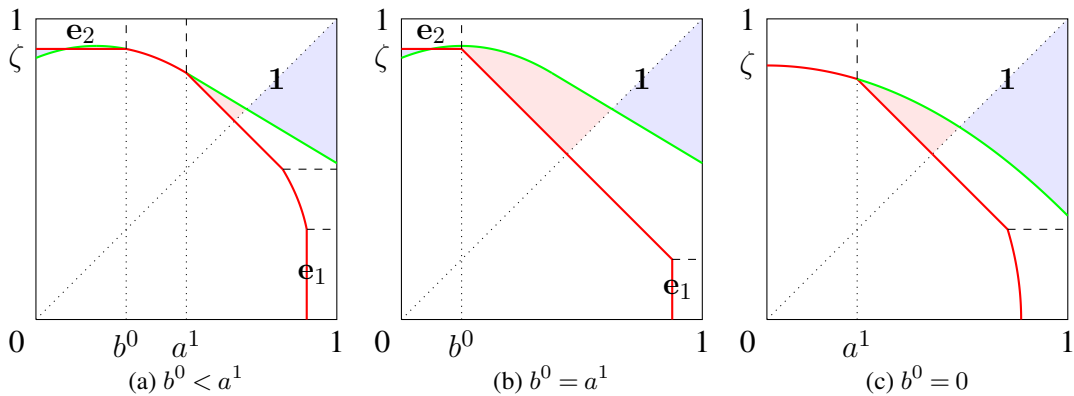

\begin{example}[Truncated normal]
\label{Exp_NormalPositive}
Let $g(x)=\kappa e^{-(x-\theta)^{2}/2}$, where $\theta\in [0,2]$, and $\kappa $ is a normalization constant. Then $h(x)=x(\theta-x)\geq -1$ for all $x\in [0,1]$, so $g$ is strictly regular. It can be proved that $\zeta$ is strictly concave with a peak at $\theta/2$. When $\theta=0$, $\zeta$ is double-regular.

We perform numerical simulations for this example. The results show that the optimal mechanism has infinite menu size (panel (a) of \Cref{Fig_Ccv}) when $\theta$ is small and becomes deterministic (panel (b) of \Cref{Fig_Ccv}) when $\theta$ is large. The cutoff is approximately 0.22. At this point, the optimal mechanism exhibits a discontinuous change in menu complexity, collapsing from an infinite menu to a deterministic one.

Moreover, we study the truncated bivariate normal distribution with $\theta=0$ and Gaussian correlation parameter $\rho$. It can be proved that $f$ is strictly regular if $\rho\in [-1/3,\sqrt{2/3}]$.
According to numerical simulations, $\zeta$ is strictly concave throughout the strictly regular range. It is double-regular when $\rho\leq 0$, whereas it is single-peaked for every $\rho>0$. The optimal mechanism shifts from an infinite menu to deterministic as $\rho$ increases, and the cutoff is approximately 0.18.
\end{example}

\begin{example}[Truncated Gamma]
\label{Exp_Gamma}
Let $g(x)=\kappa x^{\eta-1}e^{-\theta x}$, where $\theta \in (0,1]$, $\eta \geq 1$, and $\kappa $ is a normalization constant. Then $h(x)=\eta-1-\theta x\geq -1$ for all $x\in [0,1]$, so $g$ is strictly regular. It can be proved that $\zeta$ is double-regular.
\end{example}

\citet{giannakopoulos2018selling} study nonincreasing and concave $\zeta$ under a similar regularity assumption and use it to solve an alternative dual problem. However, they only establish weak duality, and it is unclear how to extend their method to cover $\zeta$'s of other shapes. %They do not investigate the connection between $\zeta$ and $\MR$ and consequently do not establish that the grand bundle $\mathbf{1}$ is always active.

\bigskip

In a nutshell, this section develops a tractable recipe to identify optimal mechanisms for a two-product monopoly from \emph{primitives}. Although some of the examples in this section have been studied in the literature, we are not aware of any existing work that derives all these results within a unified framework. In fact, finding optimal mechanisms when types follow (i.i.d. or correlated) normal or Gamma distributions has attracted significant attention, but to our knowledge, remains unsolved in the literature.\footnote{See, e.g., \citet{long1984comments,schmalensee1984gaussian}, and \citet{manelli2006bundling}.} Examples \ref{Exp_NormalPositive} and \ref{Exp_Gamma} can be regarded as a first step toward a complete solution for these cases.\footnote{In \Cref{Sec_Upperbound}, we offer an upper bound on the number of pooling bundles in any optimal mechanism when $\zeta$ does not exhibit a desirable global property as those in Corollaries \ref{Cor_Inc}--\ref{Cor_CcvDec}.}

%In \Cref{Sec_Quasiregular}, we show how to extend our approach to a specific type of irregular distribution.

\subsection{Sufficiency}
\label{Sec_Sufficiency}

While Theorem \ref{Thm_MR} provides only necessary conditions for optimal mechanisms, Corollaries \ref{Cor_Inc}--\ref{Cor_CcvDec} reduce the problem to optimizing over several parameters. In this sense, they offer an almost complete characterization of optimal mechanisms.

This section takes one step further. In \Cref{Prop_Sufficiency}, we give necessary and sufficient conditions for the optimality of $\overline{u}$, which, together with \Cref{Prop_Additive}, help characterize these parameters. The proposition follows from a variational approach similar to that used to establish \Cref{Thm_MR}.

\begin{proposition}
\label{Prop_Sufficiency}
Suppose that $f$ is SRS. Let
\begin{equation}
\label{Eqn_R}
\mathcal{R}_{\overline{u}} (x)=-\int^{1}_{x}\int^{1}_{y}\Phi (z) \dd z \dd y.
\end{equation}
Then a symmetric mechanism given by \Cref{Lem_MM} is optimal if and only if
\begin{align}
\label{Eqn_R0}
\mathcal{R}_{\overline{u}}' (0) & =0, \\
\label{Eqn_Supp}
\supp \overline{u}' & \subseteq \arg\max_{x\in [0,1]} \mathcal{R}_{\overline{u}} (x),
\end{align}
where $\supp$ represents the support of the Lebesgue--Stieltjes measure induced by $\overline{u}'$, with the endpoint convention $\overline{u}_{-}'(0)=0$ and $\overline{u}_{+}'(1)=1$.
\end{proposition}

Propositions \ref{Prop_Additive} and \ref{Prop_Sufficiency} are closely connected. If $(q,1)$ is pooling with $q\in (0,1)$, then $\{a^{q},b^{q}\}\subseteq \supp \overline{u}'$. By \eqref{Eqn_Supp}, $\mathcal{R}_{\overline{u}}$ attains its maximum at $a^{q}$ and $b^{q}$. When $0<a^{q}<b^{q}<1$, first-order conditions imply $\mathcal{R}_{\overline{u}}'(a^{q})=\mathcal{R}_{\overline{u}}'(b^{q})=0$. Thus,
\begin{equation*}
\mathcal{R}_{\overline{u}}'(a^{q})-\mathcal{R}_{\overline{u}}'(b^{q})=\MR (q,1)=0,
\end{equation*}
which is exactly condition \eqref{Eqn_IntPhi} of \Cref{Prop_Additive}.
Second-order conditions imply
\begin{equation*}
\mathcal{R}_{\overline{u}}''(a^{q})=-\Phi (a^{q})\leq 0, \text{ and }\mathcal{R}_{\overline{u}}''(b^{q})=-\Phi (b^{q})\leq 0,
\end{equation*}
which are the same as the two implications of \eqref{Eqn_SOC} in part \ref{Prop_Additive_2} of \Cref{Prop_Additive}.
If $(q,1)$ is separating, then $a^{q}=b^{q}$ and $a^{q}\in \supp\overline{u}'$. $\mathcal{R}_{\overline{u}}''(a^{q})=-\Phi (a^{q})=0$ follows from a shrinking-neighborhood argument.

Intuitively, $\mathcal{R}_{\overline{u}}(x)$ is the first-order linear approximation to total revenue induced by a stepwise allocation on the top boundary $[0,1]\times \{1\}$ (or, equivalently, by a deterministic mechanism). If an allocation rule $\overline{u}'$ is optimal, it must therefore be expressible as a weighted average of revenue-maximizing deterministic mechanisms. This observation yields the necessity of \eqref{Eqn_Supp}. The joint sufficiency of \eqref{Eqn_R0} and  \eqref{Eqn_Supp} follows from that the total revenue is concave in $\overline{u}$ under SRS. We demonstrate this formally in \Cref{Prf_Prop_Sufficiency}.

\Cref{Prop_Sufficiency} provides a straightforward interpretation of \eqref{Eqn_FOC} and \eqref{Eqn_SOC} in the two-product case: they are simply the FOC and SOC of the linearized total revenue function $\mathcal{R}_{\overline{u}}$. 
It also parallels similar support conditions in the information design literature, such as Theorem 1 of \citet{dworczak2019simple}.
Nevertheless, because $\mathcal{R}_{\overline{u}}$ depends on the endogenous mechanism $\overline{u}$, applying \Cref{Prop_Sufficiency} directly remains challenging. Instead, we propose to use the necessary conditions in \Cref{Prop_Additive} to identify key features of $\overline{u}$, and then invoke \Cref{Prop_Sufficiency} as a certificate of optimality.

For example, when $\zeta$ is convex, \Cref{Cor_Cvx} shows that the optimal mechanism has only three active bundles: $(q,1)$, $(1,q)$, and $(1,1)$. We therefore have three parameters to be determined, namely $\overline{u}(0)$, $q$, and $a^{1}$. \Cref{Prop_Additive} already gives us $\MR (\mathbf{1})=0$ and $\MR (q,1)=0$, both of which are due to \eqref{Eqn_Supp}. By \eqref{Eqn_Supp}, $\supp \overline{u}'=\{0,a^{1}\}$, which requires $\mathcal{R}_{\overline{u}}(0)=\mathcal{R}_{\overline{u}}(a^{1})$. Hence, an interior $q\in (0,1)$ must satisfy
\begin{equation*}
\mathcal{R}_{\overline{u}}(0)-\mathcal{R}_{\overline{u}}(a^{1})=-\int^{a^{1}}_{0}\left(\int^{1}_{x}\Phi(z)\dd z\right)\dd x=-a^{1}\int^{1}_{a^{1}}\Phi(x)\dd x-\int^{a^{1}}_{0}x\Phi(x)\dd x=0.
\end{equation*}
Note that $2\int^{1}_{a^{1}}\Phi(x_{1})\dd x_{1}=\MR (1,1)=0$. Thus, on top of the two MR conditions, we have
\begin{equation*}
\int^{a^{1}}_{0}x\Phi(x)\dd x=0,
\end{equation*}
and they jointly pin down the three parameters in \Cref{Cor_Cvx}.

As a comparison, DDT's Theorem 7 (also for two items) assumes neither symmetry nor regularity of the type distribution. Instead, they invoke two one-dimensional functions to describe the exclusion boundary below which types do not participate. While DDT do not offer construction of the two functions from primitives, our approach relies entirely on the features of $\zeta$ which is computed from the density function.

\subsection{Optimality of Bundling Strategies}
\label{Sec_Additive_Bundle}

Since \citet{adams1976commodity}, a large body of literature has examined the optimality of commonly used bundling strategies. Much of the early work, such as \citet{mcafee1989multiproduct}, focuses on comparing specific bundling strategies, primarily due to the difficulty of solving optimal mechanisms. With \Cref{Prop_Additive,Prop_Sufficiency}, we can address this question from a broader perspective, that is, identifying the conditions under which a given mechanism is optimal among all feasible mechanisms.

We begin with \emph{separate selling}, that is, a mechanism in which only $\mathbf{e}_{1}$ and $\mathbf{e}_{2}$ are active.\footnote{In the industrial organization literature, separate selling typically requires that $\mathbf{1}$ is active and sold at a price $p(\mathbf{1})=p(\mathbf{e}_{1})+p(\mathbf{e}_{2})$. Our definition differs slightly, as we want to offer a unified treatment to all deterministic mechanisms in \Cref{Cor_Mixed}.} By \Cref{Cor_GrandBundle}, the grand bundle $\mathbf{1}$ is always pooling. Hence, under our definition, separate selling is never optimal.

The next result provides a necessary and sufficient condition for the optimality of pure bundling. Let $p$ be a pure bundling price schedule. Then by \Cref{Prop_Additive}, two necessary conditions for its optimality are $a^{1}=0$ and $p(\mathbf{1})\leq \zeta (0)$. Hence, a type-$(x,1)$ buyer's indirect utility is
\begin{equation*}
\overline{u}(x)=u(x,1)=x+1-p(\mathbf{1}).
\end{equation*}

%Its key ingredient is the \emph{sold-alone price} for the grand bundle, or equivalently, the revenue-maximizing price when the seller is restricted to pure bundling.

\begin{corollary}
\label{Cor_Pure}
Let $p$ be a revenue-maximizing pure bundling mechanism. Then
\begin{enumerate}
\item $p$ is optimal if and only if \label{Cor_Pure_1}
\begin{equation}
\label{Eqn_Pure_Nec_2}
p(\mathbf{1})\leq 1,\text{ and }\mathcal{R}_{\overline{u}}(x) \leq \mathcal{R}_{\overline{u}}(0),\quad \forall\, x\in [0,1].
\end{equation}
\item Suppose that $\zeta$ crosses the 45-degree line exactly once. Then $p$ is optimal if \label{Cor_Pure_2}
\begin{equation}
\label{Eqn_Pure_Suf}
p(\mathbf{1})\leq x+\zeta (x),\quad \forall\, x\in [0,1].
\end{equation}
\item In addition to part \ref{Cor_Pure_2}, suppose that $\zeta$ is nondecreasing or concave. Then $p$ is optimal if and only if  \label{Cor_Pure_3}
\begin{equation}
\label{Eqn_Pure_Nec}
p(\mathbf{1}) \leq \zeta (0).
\end{equation}
\end{enumerate}
\end{corollary}

For two additive products, \citet[Theorem 6]{hart2014good} establish the optimality of pure bundling on an unbounded type space, while \citet{menicucci2015optimality} provide sufficient conditions when the type space is a rectangle bounded away from the origin. In this sense, \Cref{Cor_Pure} complements these results.

We can apply a similar approach to verify the optimality of deterministic mechanisms. In such mechanisms, $a^{1}=p(\mathbf{1})-p(\mathbf{e}_{2})$, and a type-$(x,1)$ buyer's indirect utility is
\begin{equation*}
\overline{u}(x)=u(x,1)=\max(1-p(\mathbf{e}_{2}),x+1-p(\mathbf{1})).
\end{equation*}

\begin{corollary}
\label{Cor_Mixed}
Let $p$ be a revenue-maximizing deterministic mechanism. Then
\begin{enumerate}
\item $p$ is optimal if and only if \label{Cor_Mixed_1}
\begin{equation}
\label{Eqn_Mixed_2}
\mathcal{R}_{\overline{u}}(x) \leq \mathcal{R}_{\overline{u}}(a^{1}),\quad \forall\, x\in [0,1].
\end{equation}
\item Suppose that $\zeta$ is concave. Then $p$ is optimal if and only if \label{Cor_Mixed_2}
\begin{equation}
\label{Eqn_Mixed}
\max (p(\mathbf{e}_{2}),a^{1}) \leq \zeta (a^{1}).
\end{equation}
\end{enumerate}
\end{corollary}

\begin{figure}[ht]
\centering
\begin{tikzpicture}[baseline={(0,0)},x=4.5cm,y=4.5cm]
\fill [red!10,domain=0:0.5] plot ({\x}, {2*(\x-0.3)^2+0.54}) -- (0.5,0.5) -- (0.3,0.3) -- (0,0.6);
\fill [red!10,domain=0.5:0.7] plot ({\x}, {-2*(\x-0.7)^2+0.7}) -- (0.5,0.5);
\fill [blue!10,domain=0.7:1]  plot ({\x}, {-2*(\x-0.7)^2+0.7}) -- (1,1);
%\fill [red!10,domain=0:0.2] plot ({\x}, {2*(\x-0.3)^2+0.54}) -- (0.2,0.45) -- (0.325,0.325) -- (0.3,0.3) -- (0,0.6);
\draw (0,0) node [below left] {0} -- (1,0) node [below] {1} -- (1,1) -- (0,1) node [left] {1} -- (0,0);
\draw [domain=0:0.5, thick, green] plot ({\x}, {2*(\x-0.3)^2+0.54});
\draw [domain=0.5:1, thick, green] plot ({\x}, {-2*(\x-0.7)^2+0.7});
\draw (0,0.72) node [left] {$\zeta$};
\draw [thick, red] (0,0.6) -- (0.6,0);
\draw (0,0.6) node [left] {$p(\mathbf{1})$};
\draw [dotted] (0,0) -- (1,1);
%\draw [dotted] (0.7,0) node [below] {$x_{\zeta}$} -- (0.7,0.7);
\draw (0.8,0.8) node {$\mathbf{1}$};
\draw (0.5,-0.3) node {(a) Pure bundling};
\end{tikzpicture}%
\begin{tikzpicture}[baseline={(0,0)},x=4.5cm,y=4.5cm]
\fill [blue!10,domain=0:0.1] plot ({\x}, {-10*(\x-0.2)^2+0.91}) -- (0,0.81);
\fill [red!10,domain=0.1:0.2] plot ({\x}, {-10*(\x-0.2)^2+0.91}) -- (0.2,0.81);
\fill [red!10,domain=0.2:0.5]  plot ({\x}, {-1*(\x-0.2)^2+0.91}) -- (0.7,0.7) -- (0.55,0.55) -- (0.3,0.81) -- (0.2,0.81);
\fill [blue!10] (0.7,0.7) -- (1,0.52) -- (1,1);
\draw (0,0) node [below left] {0} -- (1,0) node [below] {1} -- (1,1) -- (0,1) node [left] {1} -- (0,0);
\draw [dotted] (0,0) -- (1,1);
\draw (0,0.81) node [left] {$p(\mathbf{e}_{2})$};
\draw [domain=0:0.2, thick, green] plot ({\x}, {-10*(\x-0.2)^2+0.91});
\draw [domain=0.2:0.5, thick, green] plot ({\x}, {-1*(\x-0.2)^2+0.91});
\draw [thick, green] (0.5,0.82) -- (1,0.52);
\draw [thick, red] (0,0.81) -- (0.3,0.81) -- (0.81,0.3) -- (0.81,0);
\draw (0.8,0.8) node {$\mathbf{1}$};
\draw (0.1,0.9) node {$\mathbf{e}_{2}$};
\draw (0.9,0.1) node {$\mathbf{e}_{1}$};
\draw [dotted] (0.3,0) node [below] {\strut $a^{1}$} -- (0.3,0.81);
%\draw [dotted] (0.1,0) node [below] {\strut $x_{\zeta,1}$} -- (0.1,0.81);
%\draw [dotted] (0.7,0) node [below] {\strut $x_{\zeta,2}$} -- (0.7,0.7);
\draw [dashed] (0.3,0.81) -- (0.3,1);
\draw [dashed] (0.81,0.3) -- (1,0.3);
\draw (0.5,-0.3) node {(b) Deterministic mechanism};
\end{tikzpicture}
\caption{Corollaries \ref{Cor_Pure} and \ref{Cor_Mixed}.}
\label{Fig_Mixed}
\end{figure}

If the revenue-maximizing deterministic mechanism is a pure bundling mechanism, then we let $a^{1}=0$, and \Cref{Cor_Mixed} degenerates to \Cref{Cor_Pure}. \Cref{Fig_Mixed} graphically illustrates the two corollaries.

\Cref{tab:bundling-comparison} contrasts our \Cref{Cor_Pure,Cor_Mixed} to existing results in the literature. In the table, ``S'' represents ``sufficient condition'', whereas ``N\&S'' represents ``necessary and sufficient condition''.

\begin{table}[htbp]
    \centering
    \caption{Results on Optimality of Bundling Strategies}
    \label{tab:bundling-comparison}
    \footnotesize
    \begin{tabular}{l l l l}
        \toprule
        \textbf{Strategy}
        & \textbf{Paper}
        & \textbf{Results}
        & \textbf{Assumptions} \\
        \midrule

        \multirow{4}{*}{Pure bundling}
            & \citet{menicucci2015optimality}
            & S
            & \makecell[l]{Type space bounded away\\from origin + regularity} \\

            & \citet{hart2014good}
            & S
            & \makecell[l]{Unbounded type space\\+ reversed regularity ($\phi \leq 0$)} \\

            & This paper
            & N\&S
            & SRS \\

        \midrule

        \multirow{2}{*}{Deterministic}
            & \makecell[l]{\citet{manelli2006bundling} \\ \citet{hart2010revenue}}
            & S
            & i.i.d. + regularity\\

            & This paper
            & N\&S
            & SRS \\

        \midrule

        Separate selling
            & This paper
            & Suboptimal
            & SRS \\

        \bottomrule
    \end{tabular}
\end{table}

\subsection{Summary}
\label{Sec_Additive_Summary}

The key contribution of this section is to show that, despite the multidimensional nature of the buyer's valuation, the optimal mechanism for two products can be understood through a one-dimensional object. Under strict regularity and symmetry, the seller can characterize the mechanism through how the quantity of one product ($q_{1}$) varies with the buyer's value for that product ($x_{1}$) when the value of the other product is maximal ($x_{2}=1$). This reduction makes it possible to characterize the optimal mechanism directly from the distribution of valuations.

Central to our analysis is the type-dependent posted price $\zeta$, which can be interpreted as the monopoly price that is optimal at each value of the other product. The optimal mechanism therefore seeks to follow these pointwise monopoly prices subject to the global incentive constraints of nonlinear pricing. In this sense, the logic parallels the standard Myerson--Bulow--Roberts approach to single-product monopoly: one starts from the pointwise solution to $\MR=0$ and ``irons'' it to satisfy the convexity constraint. In the i.i.d.\ case, $\zeta$ can moreover be expressed through demand elasticity, directly linking the optimal pricing rule to the distribution of valuations. \footnote{To the best of our knowledge, the function $\zeta$ has appeared only implicitly in papers with independent distributions and additive valuations, without an explicit link to marginal revenue and posted price. See \citet[p. 13]{hart2010revenue}, DDT's Example 3, and \citet[Theorem 2]{giannakopoulos2018selling}.}

This perspective also explains why probabilistic bundling and even infinite menus can arise in multiproduct monopoly \citep{carroll2017robustness}. As shown in \Cref{Sec_Additive_Menu}, these features arise either when the pointwise monopoly solution itself calls for randomization, or when the solution cannot be implemented directly and must instead be ironed. The same framework yields sharp conditions for the optimality of familiar bundling strategies in \Cref{Sec_Additive_Bundle}, and thus provides necessary and sufficient conditions under which simple mechanisms are optimal among all feasible mechanisms.\footnote{Our use of sold-alone prices as a verification criterion parallels \citet{ghili2023characterization} and \citet{yang2025nested}, but we additionally incorporate the $\zeta$ function due to the multidimensional nature of types.}

More broadly, \Cref{Sec_Additive} addresses an open question raised by DDT: are there broad conditions under which the optimal mechanism admits a simple closed-form characterization or takes the form of \eqref{Eqn_MM}? Our results show that SRS is such an environment and provide a systematic way to determine when simple menus are sufficient and when richer mechanisms are necessary. Thus, the two-product analysis is not only a tractable solution to a canonical multiproduct monopoly problem, but also illustrates how marginal-revenue analysis can reveal the economic forces governing bundling and menu complexity.

%While strict regularity has appeared in the literature since \citet{mcafee1988multidimensional}, our paper is the first to invoke the marginal revenue approach to solve optimal mechanisms for SRS distributions from primitives.

%There are other applications of \Cref{Prop_Additive}. In \Cref{Sec_Additive_Bundle}, we provide sufficient and necessary conditions for the optimality of pure bundling or deterministic mechanisms. These results complement a large body of existing work that seeks sufficient conditions to eliminate randomization, such as \citet{mcafee1988multidimensional,manelli2006bundling}, and \citet{hart2010revenue}. 

\section{Applications}
\label{Sec_Applications}

This section studies three departures from the benchmark: a regulatory objective that places weight on the firm's information rent, an aggregate quantity floor, and nonadditive product values. The first two modify the MR measure, while the third changes the allocation space and hence the geometry of demand sets.

Let $\omega\in[0,1)$ denote the welfare weight on information rent and $\lambda\geq 0$ the multiplier on the quantity floor. The first two applications share the modified MR measure
\begin{equation}
\label{Eqn_mu_lambda_omega}
\dd\mu_{\lambda,\omega}=\dd\mu+\lambda[f(\mathbf{x})(\mathbf{1}\cdot\mathbf{n}_{\mathbf{x}})\dd\sigma(\mathbf{x})-(\mathbf{1}\cdot\nabla f(\mathbf{x}))\dd\mathbf{x}]+\omega f(\mathbf{x})\dd\mathbf{x}.
\end{equation}
Increasing $\omega$ reduces the penalty on information rent by adding positive interior mass. The multiplier $\lambda$ changes the boundary weights and, in general, the interior density through a perturbation with zero total mass. Consequently, strict regularity now refers to
\begin{equation*}
\phi_{\lambda,\omega}(\mathbf{x})=\phi(\mathbf{x})+\lambda(\mathbf{1}\cdot\nabla f(\mathbf{x}))-\omega f(\mathbf{x})>0,\quad \forall\, \mathbf{x}\in \Int (\mathcal{X}).
\end{equation*}
For $\mathcal{X}=[0,1]^{2}$, the corresponding type-dependent posted price generalizes \eqref{Eqn_Zeta}:
\begin{equation}
\label{Eqn_zeta_lambda_omega}
\begin{aligned}
\zeta_{\lambda,\omega}(x_{1})=\sup\biggl\{\zeta\in[0,1]:{}&\int^{1}_{\zeta}[\phi(\mathbf{x})+\lambda(\mathbf{1}\cdot\nabla f(\mathbf{x}))-\omega f(\mathbf{x})]\dd x_{2}\\
&{}-(1+\lambda)f(x_{1},1)\geq 0\biggr\},\quad \forall\,x_{1}\in(0,1),
\end{aligned}
\end{equation}
with the convention $\sup\varnothing=0$.
%When $\lambda>0$, the analysis must also account for the additional MR mass on the lower boundaries.

\Cref{Sec_Regulation} sets $\lambda=0$ and characterizes composition delegation and the conditions under which relative costs cease to affect allocation. \Cref{Sec_QuantityConstraints} sets $\omega=0$ and shows how a quantity floor changes total output and menu composition. Finally, \Cref{Sec_SubsComp} introduces substitutability and complementarity through the allocation space $\mathcal{Q}$, showing how changes in demand geometry affect optimal bundling.

\subsection{Regulation}
\label{Sec_Regulation}

Regulators frequently contract with a single firm that can supply a homogeneous service using several production technologies. For example, an electricity producer may use multiple generation technologies, each with its own privately observed cost, to produce the same electricity. The regulator must then decide whether to prescribe output from each technology separately or impose an aggregate generation requirement and delegate the composition of production to the firm.

The classic benchmark, \citet{baron1982regulating}, studies regulating a monopolist with a scalar private cost. In a recent work, \citet{armstrong2018multiproduct} study multiproduct firms and show that regulation can reduce to one-dimensional screening over composite cost under homothetic benefits and independence between composite and relative costs. In this section, we allow composite and relative costs to be dependent and employ the MR approach to find optimal regulatory schemes.

Consider a monopolist supplying $N=2$ products, each with unit capacity. The regulator chooses a target supply $\mathbf{q}\in \mathcal {Q}=[0,1]^{2}$ and pays a transfer $T(\mathbf{q})$ to the seller. A vector of supply $\mathbf{q}$ generates known social benefit $\mathbf{1}\cdot \mathbf{q}$ and incurs private cost $\mathbf{c}\cdot \mathbf{q}$ to the seller, so the seller's profit is $T(\mathbf{q})-\mathbf{c} \cdot \mathbf{q}$. As in \citet{baron1982regulating}, the regulator maximizes social benefit minus the transfer plus a welfare weight $\omega\in[0,1)$ on firm profit:
\begin{equation*}
\mathbf{1}\cdot \mathbf{q}-T(\mathbf{q})+\omega [T(\mathbf{q})-\mathbf{c} \cdot \mathbf{q}].
\end{equation*}
%Because both services generate the same unit benefit, a given aggregate-output target can be assigned a transfer that is independent of which service the firm supplies.
%The assumption that social benefit equals $\mathbf{1}\cdot \mathbf{q}$ is important for the aggregation result below. 

To reformulate the problem, let $\mathbf{x}=\mathbf{1}-\mathbf{c}$ denote the products' \emph{net efficiencies}, and let $p(\mathbf{q})=\mathbf{1}\cdot\mathbf{q}-T(\mathbf{q})$ denote the net payoff that the contract leaves to society. The seller's profit can then be written as $\mathbf{x}\cdot\mathbf{q}-p(\mathbf{q})$: she chooses an output vector that generates total surplus $\mathbf{x}\cdot\mathbf{q}$, leaves $p(\mathbf{q})$ to society, and retains the remainder as profit. Her profit is $u(\mathbf{x})=\max_{\mathbf{q}\in\mathcal Q}\{\mathbf{x}\cdot \mathbf{q}-p(\mathbf{q})\}$.
%Participation thus requires $u(\mathbf{x})\geq 0$ for every $\mathbf{x}$.

Assume that $\mathbf{x}\in \mathcal{X}=[0,1]^{2}$ and that its density $f$ satisfies the standing assumptions in \Cref{Sec_Model}. The regulator's objective becomes
\begin{equation*}
%\label{Eqn_Reg_ObjectiveMeasure}
\int_{\mathcal{X}}u\dd \mu+\omega\int_{\mathcal{X}}u f\dd\mathbf x=\int_{\mathcal X}u\dd\mu_\omega,
\end{equation*}
where $\mu_\omega$ is obtained from setting $\lambda=0$ in \eqref{Eqn_mu_lambda_omega}.
If, in addition, $f$ is symmetric and $\phi_{\omega}=\phi -\omega f>0$, we are back to the SRS environment studied in \Cref{Sec_Additive}.

In particular, \Cref{Lem_MM} suggests that every symmetric optimal mechanism admits an equivalent implementation using aggregate-output targets: Since all active bundles take the form of $(q,1)$ or $(1,q)$, we can construct an alternative mechanism in which every nonzero aggregate-output target $s\in [1,2]$ is associated with a transfer $s-\overline{p}(s-1,1)$ to the firm, and the zero target incurs zero transfer. A type-$\mathbf{x}$ firm with $u(\mathbf{x})>0$ in the direct mechanism chooses $s$ and $(q_{1},q_{2})$, subject to $q_{1}+q_{2}=s$, to maximize
\begin{equation*}
s-\overline{p}(s-1,1)-(1-\mathbf{x}) \cdot \mathbf{q}=\mathbf{x}\cdot \mathbf{q}-\overline{p}(s-1,1).
%=\min \{c_{1}+(s-1)c_{2},(s-1)c_{1}+c_{2}\}.
\end{equation*}
If $x_{1}<x_{2}$ (or $x_{1}>x_{2}$), the firm must choose $q_{1}=s-1$ (or $q_{2}=s-1$) and obtain the same profit $\mathbf{x}\cdot \mathbf{q}-p(\mathbf{q})$ as in the direct mechanism. Formally, when $x_{1}<x_{2}$ and $u(\mathbf{x})>0$,
\begin{equation*}
u(\mathbf{x})=\max_{q_{1}\in [0,1],q_{2}=1}\{\mathbf{x}\cdot \mathbf{q}-\overline{p}(\mathbf{q})\}=\max_{s\in [1,2]} \{\mathbf{x}\cdot (s-1,1)-\overline{p}(s-1,1)\}.
\end{equation*}

Moreover, one can define the type-dependent posted price $\zeta_{\omega}$ by setting $\lambda=0$ in \eqref{Eqn_zeta_lambda_omega}, and find optimal mechanisms using the recipe established in \Cref{Sec_Additive_Menu}.
\Cref{Exp_RegulationUniform} illustrates the uniform case.

%has boundary density $(\mathbf x\cdot\mathbf n_{\mathbf x})f(\mathbf x)$ and interior density $-\phi_\omega(\mathbf x)=-[\mathbf x\cdot\nabla f(\mathbf x) +(N+1-\omega)f(\mathbf x)]$.
%Therefore, most of our results from MR apply with $\mu$ replaced by $\mu_{\omega}$. Since a higher $p(\mathbf q)$ corresponds to a lower transfer, MR measures the welfare effect of a local subsidy cut.

\begin{example}
\label{Exp_RegulationUniform}
Suppose that types are uniformly distributed. Then $\zeta_{\omega}=(2-\omega)/(3-\omega)$, and the optimal regulation mechanism is characterized by the cutoff $\omega^{*}=2-\sqrt{2}$:
\begin{enumerate}
\item \label{Exp_RegulationUniform_1} If $\omega\in [0,\omega^{*})$, active (target) bundles are $\mathbf{e}_{1}$, $\mathbf{e}_{2}$, and $\mathbf{1}$, with transfers
\begin{equation*}
T(\mathbf{e}_{1})=T(\mathbf{e}_{2})=1-\zeta_{\omega}=\frac{1}{3-\omega},
\quad
T(\mathbf{1})=\frac{2+\sqrt2}{3-\omega}.
%\label{Eqn_Reg_MixedTransfers}
\end{equation*}
Equivalently, the regulator can offer total-output targets
from $\{0,1,2\}$ with the same transfers.
\item \label{Exp_RegulationUniform_2} If $\omega\in [\omega^{*},1)$, only $\mathbf{1}$ is active, with
\begin{equation*}
T(\mathbf{1})=2-\sqrt{\frac{2(1-\omega)}{3-\omega}}.
%\label{Eqn_Reg_PureMenu}
\end{equation*}
Equivalently, the regulator offers targets from $\{0,2\}$, and the firm chooses $Q=2$ if and only if $c_{1}+c_{2}\leq T(\mathbf{1})$.
See \Cref{Fig_Regulation} for an illustration.
\end{enumerate}
\end{example}

\begin{figure}
\centering
\begin{tikzpicture}[baseline={(0,0)},x=4.5cm,y=4.5cm]
\fill [red!10] (1/3,2/3) -- (2/3,2/3) -- (1/2,1/2);
\fill [blue!10] (2/3,2/3) -- (1,1) -- (1,2/3);
\draw (0,0) node [below left] {0} -- (1,0) node [below] {1} -- (1,1) -- (0,1) node [left] {1} -- (0,0);
\draw [dotted] (0,0) -- (1,1);
\draw [thick, green] (1/3,2/3) -- (1,2/3);
\draw [thick, red] (0,2/3) node [black, left] {$\zeta_{\omega}$} -- (1/3,2/3) -- (2/3,1/3) -- (2/3,0);
\draw (0.2,0.8) node {$\mathbf{e}_{2}$};
\draw (0.8,0.2) node {$\mathbf{e}_{1}$};
\draw (0.8,0.8) node {$\mathbf{1}$};
\draw [dashed] (2/3,1/3) -- (1,1/3);
\draw [dashed] (1/3,2/3) -- (1/3,1);
\draw (0.5,-0.3) node {\footnotesize (i) $\omega\in [0,\omega^{*})$};
\end{tikzpicture}%
\begin{tikzpicture}[baseline={(0,0)},x=4.5cm,y=4.5cm]
\fill [red!10] (0,1/3) -- (0,1/2) -- (1/2,1/2) -- (1/6,1/6);
\fill [blue!10] (1/2,1/2) -- (1,1) -- (1,1/2);
\draw (0,0) node [below left] {0} -- (1,0) node [below] {1} -- (1,1) -- (0,1) node [left] {1} -- (0,0);
\draw [dotted] (0,0) -- (1,1);
\draw [thick, green] (0,1/2) -- (1,1/2);
\draw [thick, red] (0,1/3) node [black, left] {$\zeta_{\omega}$} -- (1/3,0);
\draw (0.7,0.7) node {$\mathbf{1}$};
\draw (0.5,-0.3) node {\footnotesize (ii) $\omega\in [\omega^{*},1)$};
\draw (1.5,0.5) node [text width=3cm] {
\begin{equation*}
\begin{aligned}
\phi -\omega f & =3-\omega \\
\zeta_{\omega} & =\frac{2-\omega}{3-\omega}
\end{aligned}
\end{equation*}
};
\end{tikzpicture}
\caption{Optimal Regulation}
\label{Fig_Regulation}
\end{figure}

By \Cref{Exp_RegulationUniform}, placing greater weight on the regulated firm's rent reduces output distortions, echoing the scalar-type result of \citet{baron1982regulating}. 
Our multidimensional setting, however, reveals an additional margin: changing this weight can also alter which dimensions of private information affect the allocation. Below $\omega^{*}$, allocation depends on relative costs, whereas above $\omega^{*}$, relative costs become irrelevant, and allocation depends only on aggregate costs.
In contrast, \citet{armstrong2018multiproduct} allow for an arbitrary number of products while assuming that aggregate costs and relative costs are independent. As a result, the firm's choice of relative quantity ($q_{1}/q_{2}$) is independent of its total quantity target ($q_{1}+q_{2}$), and the regulator does not benefit from screening over relative costs.

\subsection{Quantity Constraints}
\label{Sec_QuantityConstraints}

Quantity constraints arise when a seller or contractor must deliver a minimal or maximal level of aggregate output.
In a recent paper, \citet{loertscher2022monopoly} study a monopolist facing fixed inventories at different quality levels and show that nonconcavity of market-clearing revenue can make rationing and conflation optimal. This section studies a multidimensional counterpart: the seller faces an aggregate quantity floor across several products. As a result, the seller can satisfy the requirement not only by increasing total delivery, but also by changing the combinations of products offered to different buyers.

%\footnote{For example, the World Bank's procurement guidance for sewage treatment plants gives, based on a real case scenario, a minimum treatment-capacity requirement: the plant must be designed to handle an average daily flow of at least a specified number of million liters per day. See the World Bank's \href{https://documents1.worldbank.org/curated/en/099325208202512194/pdf/IDU-a052dfd8-ab44-44c1-b854-870b4cccacdb.pdf}{Procurement Guidance} for more details.}

For tractability, we abstract from \citeauthor{loertscher2022monopoly}'s quality dimension and instead study an expected quantity floor across all products under $\mathcal{X}=\mathcal{Q}=[0,1]^{2}$:
\begin{equation*}
\int_{\mathcal{X}} (q_{1}(\mathbf{x})+q_{2}(\mathbf{x}) )f(\mathbf{x})\dd \mathbf{x}\geq K,\text{ where }K\in (0,2).
\end{equation*}
Integrating by parts and duality imply that an optimal mechanism solves
\begin{equation*}
\max_{u}\int_{\mathcal{X}}u\dd\mu_{\lambda}
\end{equation*}
for some multiplier $\lambda \geq 0$, together with complementary slackness, where $\mu_{\lambda}$ is obtained from setting $\omega=0$ in \eqref{Eqn_mu_lambda_omega}.
For a more intuitive representation, define
\begin{equation*}
%\label{Eqn_QuantityShift}
\mathbf{y}=\mathbf{x}+\lambda \mathbf{1},
\quad
f_{\lambda} (\mathbf{y})=f(\mathbf{y}-\lambda \mathbf{1}),
\quad
p_{\lambda} (\mathbf{q})=p(\mathbf{q})+\lambda \mathbf{1}\cdot\mathbf{q}.
\end{equation*}
Then, modulo the complementary slackness constraint, the quantity floor changes the seller's incentives exactly as if all buyers' values had increased by $\lambda$ and the resulting subsidy were subsequently deducted from prices.

%\Cref{Exp_CapacityUniform} studies an example with $\mathcal{X}=\mathcal{Q}=[0,1]^{2}$ and a uniform type distribution. 

%In this case, the interior MR density is unchanged; only boundary weights move with the shadow value: $\mu_{\lambda}$ assigns mass $1+\lambda$ to each top boundary and mass $-\lambda$ to each bottom boundary. This reweighting can change the optimal menu, and its effect is inherently multidimensional: in a single-product problem (without the quality dimension), a quantity requirement mainly moves the cutoff type, whereas here it also changes how products are combined and how buyers are screened.
%The main results of \Cref{Exp_CapacityUniform} are adapted from \citet{pavlov2011optimal}. For $\lambda >0$, we first identify the optimal mechanism for types uniformly distributed on $[\lambda,1+\lambda]^2$, then map this menu back using \eqref{Eqn_QuantityShift}.

Again, one can define the type-dependent posted price $\zeta_{\lambda}$ by setting $\omega=0$ in \eqref{Eqn_zeta_lambda_omega}. A crucial difference from the benchmark in \Cref{Sec_Additive} is that now $\mu_{\lambda}$ places negative mass on the bottom boundary $[0,1]\times \{0\}$ and $\{0\}\times [0,1]$. Modulo that, one can generally follow the recipe in \Cref{Sec_Additive_Menu}. \Cref{Exp_CapacityUniform} illustrates the uniform case.

\begin{example}
\label{Exp_CapacityUniform}
Under uniform distribution, $\zeta_{\lambda}=\max\{ (2-\lambda)/3,0\}$.
Let $K_{0}=(8+2\sqrt{2})/9$ denote the expected output under the unconstrained optimal mechanism. Then there exist $\lambda^{*}>0$ and its associated $K_{\lambda^{*}}\in (K_{0},2)$, such that:
\begin{enumerate}
\item If $K\leq K_{0}$, the minimum-service requirement does not change the optimal mechanism. Mixed bundling is optimal.
%, with prices
%\begin{equation*}
%p(\mathbf{e}_{1})=p(\mathbf{e}_{2})=\frac{2}{3},
%\quad
%p(\mathbf{1})=\frac{4-\sqrt{2}}{3}.
%\end{equation*}

\item \label{Exp_CapacityUniform_2} If $K_{0}<K<K_{\lambda^{*}}$, the requirement binds and the nonzero active bundles are $(q,1)$, $(1,q)$, and $\mathbf{1}$, where $q \in (0,1)$ varies with $K$.

\item If $K_{\lambda^{*}}\leq K<2$, pure bundling is optimal, with $p(\mathbf{1})=\sqrt{2-K}$.
\end{enumerate}
Numerical calculation gives $\lambda^{*}\approx 0.077$ and $K_{\lambda^{*}}\approx 1.41$.
\end{example}

\begin{figure}
\centering
\begin{tikzpicture}[baseline={(0,0)},x=4.5cm,y=4.5cm]
\fill [red!10] (1/3,2/3) -- (2/3,2/3) -- (1/2,1/2);
\fill [blue!10] (2/3,2/3) -- (1,1) -- (1,2/3);
\draw (0,0) node [below left] {0} -- (1,0) node [below] {1} -- (1,1) -- (0,1) node [left] {1} -- (0,0);
\draw [dotted] (0,0) -- (1,1);
\draw [thick, green] (1/3,2/3) -- (1,2/3);
\draw [thick, red] (0,2/3) node [black, left] {$\zeta_{\omega}$} -- (1/3,2/3) -- (2/3,1/3) -- (2/3,0);
\draw (0.2,0.8) node {$\mathbf{e}_{2}$};
\draw (0.8,0.2) node {$\mathbf{e}_{1}$};
\draw (0.8,0.8) node {$\mathbf{1}$};
\draw [dashed] (2/3,1/3) -- (1,1/3);
\draw [dashed] (1/3,2/3) -- (1/3,1);
\draw (0.5,-0.3) node {\footnotesize (i) $K\leq K_{0}$};
\end{tikzpicture}%
\begin{tikzpicture}[baseline={(0,0)},x=4.5cm,y=4.5cm]
\fill [blue!10] (0,1/2) -- (0,2/3) -- (1/6,7/12);
\fill [red!10] (1/6,7/12) -- (7/12,7/12) -- (5/12,5/12) -- (1/3,1/2);
\fill [blue!10] (2/3,2/3) -- (1,1) -- (1,2/3);
\draw (0,0) node [below left] {0} -- (1,0) node [below] {1} -- (1,1) -- (0,1) node [left] {1} -- (0,0);
\draw [dotted] (0,0) -- (1,1);
\draw [thick, green] (0,7/12) node [black, left] {$\zeta_{\omega}$} -- (1,7/12);
\draw [thick, red] (0,2/3) -- (1/3,1/2) -- (1/2,1/3) -- (2/3,0);
\draw (0.2,0.8) node {$(q,1)$};
\draw (0.8,0.2) node {$(1,q)$};
\draw (0.8,0.8) node {$\mathbf{1}$};
\draw [dashed] (1/2,1/3) -- (1,1/3);
\draw [dashed] (1/3,1/2) -- (1/3,1);
\draw (0.5,-0.3) node {\footnotesize (ii) $K\in (K_{0},K_{\lambda^{*}})$};
\end{tikzpicture}%
\begin{tikzpicture}[baseline={(0,0)},x=4.5cm,y=4.5cm]
\fill [red!10] (0,1/3) -- (0,1/2) -- (1/2,1/2) -- (1/6,1/6);
\fill [blue!10] (1/2,1/2) -- (1,1) -- (1,1/2);
\draw (0,0) node [below left] {0} -- (1,0) node [below] {1} -- (1,1) -- (0,1) node [left] {1} -- (0,0);
\draw [dotted] (0,0) -- (1,1);
\draw [thick, green] (0,1/2) -- (1,1/2);
\draw [thick, red] (0,1/3) node [black, left] {$\zeta_{\omega}$} -- (1/3,0);
\draw (0.7,0.7) node {$\mathbf{1}$};
\draw (0.5,-0.3) node {\footnotesize (iii) $K\geq K_{\lambda^{*}}$};
%\draw (1.5,0.5) node [text width=3cm] {
%\begin{equation*}
%\begin{aligned}
%\phi -\omega f & =3-\omega \\
%\zeta_{\omega} & =\frac{2-\omega}{3-\omega}
%\end{aligned}
%%};
\end{tikzpicture}
\caption{Optimal Mechanism under Quantity Floor}
\label{Fig_Capacity}
\end{figure}

\begin{comment}
The two thresholds in \Cref{Exp_CapacityUniform} have different economic meanings. The first determines when the policy begins to affect behavior. The second marks a change in menu topology: screening through product composition disappears and the mechanism collapses to a single nonzero bundle.

The pure-bundling case also makes the distributional consequences transparent. Since the bundle is purchased by a mass $K/2$ of buyers, total revenue, consumer surplus, and total surplus are
\begin{equation}
\label{Eqn_QuantityWelfare}
\begin{aligned}
\TR (K)
&=\frac{K\sqrt{2-K}}{2},\\
\mathrm{CS}(K)
&=1-\sqrt{2-K}+\frac{(2-K)^{3/2}}{6},\\
\mathrm{W}(K)
&=1-\frac{(2-K)^{3/2}}{3}.
\end{aligned}
\end{equation}
Because $K_{\lambda^{*}}>4/3$, tightening the requirement throughout the pure-bundling region reduces revenue while increasing consumer surplus and total surplus. The requirement therefore affects both the quantity delivered and the way surplus is divided.
\end{comment}

\Cref{Exp_CapacityUniform} provides a transparent comparative static.  A modest binding floor turns each single-product bundle into a lottery that includes a positive probability of the other product. As the floor becomes more demanding, screening through these probabilistic bundles becomes less valuable relative to delivering both products. The lotteries eventually disappear, and the menu collapses to pure bundling. Quantity restrictions consequently operate through two channels. The \emph{direct channel} changes the amount supplied. The \emph{composition channel} changes relative prices and randomization. This complements the quantity-and-quality effects in \citet{loertscher2022monopoly}: their randomization changes which qualities clear the market, whereas ours changes how several products are combined within a screening contract. Finally, we note that a similar approach can be applied to study (linear) production costs, taxation, or subsidies.

\begin{comment}
\begin{figure}[htbp]
\centering
\input{Fig_Capacity1}
%\caption{Caption}
\label{Fig_Capacity1}
\end{figure}
\end{comment}

\subsection{Substitutes and Complements}
\label{Sec_SubsComp}

The benchmark model assumes that the values of the two products are additive. Nevertheless, the framework allows us to study substitutability and complementarity by encoding these features in the allocation space $\mathcal{Q}$.
Motivated by the existing literature, such as \citet{armstrong2013more} and \citet{zhong2026bundling}, we assume that a type-$\mathbf{x}$ buyer receives utility $k\mathbf{x}\cdot \mathbf{1}=k(x_{1}+x_{2})$ from consuming the grand bundle $\mathbf{1}$. Thus, the two products are perfect substitutes if $k= 1/2$, partial substitutes if $1/2<k<1$, and partial complements if $k>1$.

The seller is allowed to offer lotteries, so the allocation space $\mathcal{Q}_{k}$ is the convex hull of $\{\mathbf{0},\mathbf{e}_{1},\mathbf{e}_{2},\mathbf{k}\}$. Note that $\mathcal{Q}_{k}$ with $k<1/2$ is the same as $\mathcal{Q}_{1/2}$, so we omit this case and focus solely on $k\geq 1/2$. For consistency, we still refer to elements in $\mathcal{Q}_{k}$ as bundles. When there is no risk of confusion, we also use $(q,1)_{k}$ to represent a bundle that assigns probability $q$ to the grand bundle and probability $1-q$ to product 2. That is,
\begin{equation*}
(q,1)_{k}=q \mathbf{k}+(1-q)\mathbf{e}_{2}=(qk,qk+1-q).
\end{equation*}
Again, the density function $f$ is SRS, and we focus on symmetric optimal mechanisms.

\begin{figure}[ht]
\centering
\begin{minipage}{0.33\textwidth}
\begin{center}
\begin{tikzpicture}[baseline={(0,0)},x=4.5cm,y=4.5cm]
\draw [<->] (0,1) node [left] {$q_{2}$} -- (0,0) node [below left] {0} -- (1,0) node [below] {$q_{1}$};
\draw (0.15,0.15) node {$\mathcal{Q}_{k}$};
\draw [thick, red] (0,0.6) -- (0.6,0);
\draw [dotted] (0,0.3) node [left] {$k$} -- (0.3,0.3) -- (0.3,0) node [below] {$k$};
\draw (0,0.6) node [left] {1};
\draw (0.6,0) node [below] {1};
\draw [->] (0.3,0.3) -- (0.4,0.4) node [midway, below right] {$\mathbf{n}_{k}$};
%\draw (0.7,-0.4) node {(a) $k\leq 1/2$};
\end{tikzpicture}
\begin{tikzpicture}[baseline={(0,0)},x=4.5cm,y=4.5cm]
\draw (0,0) node [below left] {0} -- (1,0) node [below] {1} -- (1,1) -- (0,1) node [left] {1} -- (0,0);
\draw [dotted] (0,0) -- (1,1);
\draw [thick, red] (0,0.5) -- (0.2,0.4) -- (0.4,0.2) -- (0.5,0);
\draw [dashed, red] (0.2,0.4) -- (0.8,1);
\draw [dashed, red] (0.4,0.2) -- (1,0.8);
\draw [dashed, red] (0,0.6) -- (0.4,1);
\draw [dashed, red] (0.6,0) -- (1,0.4);
\draw [->] (0.5,0.7) -- (0.6,0.8) node [midway, below right] {$\mathbf{n}_{k}$};
\draw (0.12,0.88) node {\footnotesize $(0,1)$};
\draw (0.3,0.7) node {\footnotesize $(q,1-q)$};
\draw (0.5,0.5) node {\footnotesize $\left(\frac{1}{2},\frac{1}{2}\right)$};
\draw (0.7,0.3) node {\footnotesize $(1-q,q)$};
\draw (0.88,0.12) node {\footnotesize $(1,0)$};
\draw (0.5,-0.2) node {\footnotesize (a) $k=1/2$};
\end{tikzpicture}
\end{center}
\end{minipage}%
\begin{minipage}{0.33\textwidth}
\begin{center}
\begin{tikzpicture}[baseline={(0,0)},x=4.5cm,y=4.5cm]
\draw [<->] (0,1) node [left] {$q_{2}$} -- (0,0) node [below left] {0} -- (1,0) node [below] {$q_{1}$};
\draw (0.2,0.2) node {$\mathcal{Q}_{k}$};
\draw [thick, red] (0,0.6) -- (0.36,0.36) -- (0.6,0);
\draw [dotted] (0,0.36) node [left] {$k$} -- (0.36,0.36) -- (0.36,0) node [below] {$k$};
\draw (0,0.6) node [left] {1};
\draw (0.6,0) node [below] {1};
\draw [->] (0.18,0.48) -- (0.26,0.6) node [midway, below right] {$\mathbf{n}_{k}$};
\draw [->] (0.48,0.18) -- (0.6,0.26) node [midway, above left] {$\mathbf{n}_{k}'$};
%\draw (0.7,-0.4) node {(b) $1/2<k<1$};
\end{tikzpicture}
\begin{tikzpicture}[baseline={(0,0)},x=4.5cm,y=4.5cm]
\draw (0,0) node [below left] {0} -- (1,0) node [below] {1} -- (1,1) -- (0,1) node [left] {1} -- (0,0);
\draw [dotted] (0,0) -- (1,1);
\draw [thick, red] (0,0.5) -- (0.2,0.4) -- (0.4,0.2) -- (0.5,0);
\draw [dashed, red] (0.2,0.4) -- (0.6,1);
\draw [dashed, red] (0.4,0.2) -- (1,0.6);
\draw [dashed, red] (0,0.55) -- (0.3,1);
\draw [dashed, red] (0.55,0) -- (1,0.3);
\draw [->] (0.4,0.7) -- (0.48,0.82) node [midway, below right] {$\mathbf{n}_{k}$};
\draw [->] (0.7,0.4) -- (0.82,0.48) node [midway, above left] {$\mathbf{n}_{k}'$};
\draw (0.11,0.89) node {\footnotesize $(0,1)$};
\draw (0.27,0.73) node {\footnotesize $(q_{1},q_{2})$};
\draw (0.5,0.5) node {\footnotesize $(k,k)$};
\draw (0.73,0.27) node {\footnotesize $(q_{2},q_{1})$};
\draw (0.89,0.11) node {\footnotesize $(1,0)$};
\draw (0.5,-0.2) node {\footnotesize (b) $1/2<k<1$};
\end{tikzpicture}
\end{center}
\end{minipage}%
\begin{minipage}{0.33\textwidth}
\begin{center}
\begin{tikzpicture}[baseline={(0,0)},x=4.5cm,y=4.5cm]
\draw [<->] (0,1) node [left] {$q_{2}$} -- (0,0) node [below left] {0} -- (1,0) node [below] {$q_{1}$};
\draw (0.4,0.4) node {$\mathcal{Q}_{k}$};
\draw [thick, red] (0,0.6) -- (0.72,0.72) -- (0.6,0);
\draw [dotted] (0,0.72) node [left] {$k$} -- (0.72,0.72) -- (0.72,0) node [below] {$k$};
\draw (0,0.6) node [left] {1};
\draw (0.6,0) node [below] {1};
\draw [->] (0.36,0.66) -- (0.34,0.78) node [midway, right] {$\mathbf{n}_{k}$};
\draw [->] (0.66,0.36) -- (0.78,0.34) node [midway, above] {$\mathbf{n}_{k}'$};
%\draw (0.7,-0.4) node {(c) $k>1$};
\end{tikzpicture}
\begin{tikzpicture}[baseline={(0,0)},x=4.5cm,y=4.5cm]
\draw (0,0) node [below left] {0} -- (1,0) node [below] {1} -- (1,1) -- (0,1) node [left] {1} -- (0,0);
\draw [dotted] (0,0) -- (1,1);
\draw [thick, red] (0,0.4) -- (0.2,0.4) -- (0.35,0.35) -- (0.4,0.2) -- (0.4,0);
\draw [dashed, red] (0.2,0.4) -- (0.1,1);
\draw [dashed, red] (0.4,0.2) -- (1,0.1);
\draw [dashed, red] (0.52,0.52) -- (0.44,1);
\draw [dashed, red] (1,0.44) -- (0.52,0.52);
\draw [dashed, red] (0.35,0.35) -- (0.52,0.52);
\draw [->] (0.48,0.76) -- (0.46,0.88) node [midway, right] {$\mathbf{n}_{k}$};
\draw [->] (0.76,0.48) -- (0.88,0.46) node [midway, above] {$\mathbf{n}_{k}'$};
\draw (0.1,0.6) node {\footnotesize $(0,1)$};
\draw (0.35,0.65) node {\footnotesize $(q_{1},q_{2})$};
\draw (0.7,0.7) node {\footnotesize $(k,k)$};
\draw (0.65,0.35) node {\footnotesize $(q_{2},q_{1})$};
\draw (0.6,0.1) node {\footnotesize $(1,0)$};
\draw (0.5,-0.2) node {\footnotesize (c) $k>1$};
\end{tikzpicture}
\end{center}
\end{minipage}%
\caption{Active bundles under different $k$.}
\label{Fig_SubsComp}
\end{figure}

% $\mathbf{q}$ to represent a generic bundle. 
%In this new setting, part \ref{Prop_Additive_1} of \Cref{Prop_Additive} can be extended to characterize the set of (potentially) active bundles. Specifically, if $\mathbf{q}$ is active with $q_{1}\leq q_{2}$, then $q_{2}=(1-1/k)q_{1}+1$.\footnote{\citet{pavlov2011property} also characterizes the set of (potentially) active bundles when $\mathcal{X}$ is a hyperrectangle. See \Cref{Sec_Active} for a detailed discussion.} 

%According to \Cref{Fig_SubsComp}, if $\mathbf{q}$ is active with $q_{1}\leq q_{2}$, then $\mathbf{n}_{k}=(1/k-1,1)$ is an outward normal to $\mathcal{Q}_{k}$ at $\mathbf{q}$, that is, $\mathbf{q}\cdot\mathbf{n}_{k}=1$.
%To highlight the relationship between active bundles and $\mathbf{n}_{k}$, we denote by $\hat{\mathbf{n}}_{k}=(1,1-1/k)$ the vector orthogonal to $\mathbf{n}_{k}$ and use $(q,1)_{k}$ to represent a generic active bundle. We also let $\hat{\mathbf{n}}_{k}'=(1-1/k,1)$.

%\subsection{Optimal Mechanisms}
%\label{Sec_SC_Menu}

We start with a geometric description of how the solution recipe in \Cref{Sec_Additive_Menu} extends to the present setting.
For different values of $k$, the three pictures in the top row of \Cref{Fig_SubsComp} illustrate the allocation space $\mathcal{Q}_{k}$ and highlight the (potentially) active bundles in red. Given these bundles, the three pictures in the bottom row of \Cref{Fig_SubsComp} use dashed red lines to depict the relationship between demand sets. The vectors $\mathbf{n}_{k}=(1/k-1,1)$ and $\mathbf{n}_{k}'=(1,1/k-1)$ are outward normals to $\mathcal{Q}_{k}$ and also determine the boundaries of the corresponding demand sets. Moreover, active bundles lie on supporting hyperplanes with these outward normals: $(q,1)_{k}\cdot \mathbf{n}_{k}=(1,q)_{k}\cdot \mathbf{n}_{k}'=1$. Thus, the indirect utility of type $(x_{1},1)$ determines the indirect utility of type $(x_{1},1)-\tau \mathbf{n}_{k}$ when the latter is above the 45-degree line. To simplify notation, we let
\begin{equation*}
\gamma_{k}(x)=\frac{1-k}{k}(1-x).
\end{equation*}
Then the Cartesian coordinates of $(x_{1},1)-(1-x_{2})\mathbf{n}_{k}$ are $(x_{1}-\gamma_{k}(x_{2}),x_{2})$.
Similar to \eqref{Eqn_Zeta}, we define the type-dependent posted price $\zeta_{k}$ for each $k$ as:
\begin{equation*}
\zeta_{k} (x_{1}) =\sup\bigg\{\zeta\in [0,1]:\int^{1}_{\zeta }\phi (x_{1}-\gamma_{k}(x_{2}),x_{2})\dd x_{2}-f(x_{1},1)\geq 0\bigg\},\quad \forall\, x_{1}\in (0,1),
\end{equation*}
with similar conventions as those imposed for $\zeta$.\footnote{To simplify notation, we let $\phi (x_{1}-\gamma_{k}(x_{2}),x_{2})=0$ whenever $(x_{1}-\gamma_{k}(x_{2}),x_{2})\notin\mathcal{X}$.}
Just as in \Cref{Sec_Additive}, for every $x\in (0,1)$, the sign of $\Phi_{k}(x)$ is the same as the sign of $\zeta_{k}(x)-\max(1-\overline{u}(x),z_{k}(x))$.

The shape of $\zeta_{k}$ provides useful information about the optimal menu size. To illustrate, we revisit the example of a uniform type distribution, which can be proved by an argument similar to that in Corollaries \ref{Cor_Inc} and \ref{Cor_Cvx}. Recall that this example is graphically illustrated in \Cref{Fig_Uniform_k}.

\begin{example}
%[Uniform Distribution Revisited]
\label{Exp_Uniform_k}
Under uniform distribution, when $k\in [1/2,1)$,
\begin{equation*}
\zeta_{k}(x)=
\begin{cases}
0 & x< (1-k)/(3k), \\
2/3 & x\geq (1-k)/(3k),
\end{cases}
\end{equation*}
and when $k>1$,
\begin{equation*}
\zeta_{k}(x)=
\begin{cases}
2/3 & x\leq (2k+1)/(3k), \\
0 & x> (2k+1)/(3k).
\end{cases}
\end{equation*}
The set of active bundles in optimal mechanisms varies with $k$:
\begin{enumerate}
\item If $k=1/2$, then only $\mathbf{e}_{1}$ and $\mathbf{e}_{2}$ are active.
\item If $k\in (1/2,1)$, then only $\mathbf{e}_{1}$, $\mathbf{e}_{2}$, and $\mathbf{k}$ are active.
\item If $k>1$, then at most three bundles are active, namely, $(q,1)_{k}$, $(1,q)_{k}$, and $\mathbf{k}$, where $q\in (0,1)$.
\label{Exp_Uniform_k_3}
\end{enumerate}
We perform numerical simulations for part \ref{Exp_Uniform_k_3}. The results show that the optimal mechanism requires randomization when $k$ is close to 1 and becomes pure bundling when $k$ is large. The cutoff is approximately 1.21.\footnote{\citet{zhong2026bundling} independently studies the same scaled-grand-bundle specification in a broader multi-product environment. For two uniformly distributed values, he also identifies $k=1/2$ as the threshold above which the grand bundle must be offered and provides a sufficient condition for pure bundling at $k\approx1.253$. His approach is dual and geometric, whereas ours uses primal marginal-revenue conditions and also covers substitutes.}
\end{example}

%In contrast to \Cref{Sec_Additive_Bundle}, \Cref{Exp_Uniform_k} suggests that separate selling is optimal when products are perfect substitutes and types are uniformly distributed. 

For other shapes of $\zeta_{k}$, we summarize our results and contrast them to prior works in \Cref{Tab_1}.
Their formal statements are relegated to Corollaries \ref{Cor_k=1/2}--\ref{Cor_k>1} in \Cref{Sec_SubsComp_Omitted}.

\begin{table}[ht]
\centering
\caption{Results on Optimal Menu Size.}
\label{Tab_1}
\resizebox{\textwidth}{!}{
\begin{threeparttable}
\begin{tabular}{ccccc}
\toprule
 & \textbf{Additive} & \textbf{Perfect Substitutes} & \textbf{Partial Substitutes} & \textbf{Complements} \\
% \hline
 & $k=1$ & $k=1/2$ & $1/2<k<1$ & $k>1$ \\
\midrule
Uniform & Mixed bundling\tnote{1} & Separate selling\tnote{2} & Mixed bundling & Menu size $\leq 3$\tnote{3} \\
Nondecreasing $\zeta$ & No randomization\tnote{1} & Separate selling & No randomization & Menu size $\leq 3$\\
Convex $\zeta$ & Menu size $\leq 3$ & Menu size $\leq 4$ & Menu size $\leq 5$ & Menu size $\leq 3$ \\
Concave $\zeta$ & \multicolumn{4}{c}{\multirow{2}{*}{Finite or infinite menu size characterized by $\zeta$}\tnote{4}} \\
Nonincreasing and concave $\zeta$ & \\
\bottomrule
\end{tabular}
\begin{tablenotes}
\item [1] \citet{manelli2006bundling,hart2010revenue}.
\item [2] \citet{pavlov2011optimal}.
\item [3] \citet{kash2016optimal} study a special case where $k=2$.
\item [4] \citet{giannakopoulos2018selling}.
\end{tablenotes}
\end{threeparttable}
}
\end{table}

Although several papers study the sale of multiple products to nonadditive buyers, they either consider one-dimensional types \citep{bergemann2021optimality,ghili2023characterization,yang2025nested,kwak2025when} or focus on a particular class of mechanisms \citep{haghpanah2021pure}.\footnote{\citet{armstrong1996multiproduct} and \citet{rochet1998ironing} allow for general utility functions and multidimensional types, but do not specifically study substitutes or complements.}
The one-dimensional settings yield tractable representations of total revenue and allow for heterogeneous ordinal rankings over products, but randomization typically plays no role.

This section therefore offers a complementary perspective: it summarizes product interactions through a single parameter $k$, which determines the relevant normal vector $\mathbf{n}_{k}$ and thereby changes the geometry of demand sets and the associated MR conditions.
Of course, requiring $k$ to be common across types is not without loss, but we view this as a useful tradeoff that yields two advantages. First, it gives rise to a rich set of optimal mechanisms that remain largely unexplored in the literature. Second, the parameter $k$ allows us to conduct comparative statics and isolate how substitutability and complementarity shape optimal mechanisms. As \Cref{Exp_Uniform_k} illustrates, increasing $k$ moves the optimal mechanism from separate selling to mixed bundling, then to probabilistic bundling, and eventually to pure bundling. In particular, partial substitutability ($1/2<k<1$) can induce mixed bundling, while weak complementarity ($1<k<1.21$) can make randomization optimal.

%The cases where $1/2<k<1$ or $k>1$ are largely unexplored in the literature. By appealing to the marginal revenue approach, we are able to solve for optimal mechanisms for uniform distribution and characterize the optimal menu size for other SRS distributions.

%% Palfrey (1983)?

\section{Beyond Two Products}
\label{Sec_General}

When the type and the allocation spaces are general convex sets, obtaining sharp menu-size bounds is more difficult than in the two-product case. However, \Cref{Thm_MR} still yields novel insights into the structure of optimal mechanisms.

\subsection{Pooling Bundles}
\label{Sec_Pooling}

This section applies \eqref{Eqn_SOC} to identify bundles that must be pooling in an optimal mechanism. Since production is costless, a bundle that some type strictly prefers to every alternative generates the highest gross value for that type. This strict preference is sufficient for the bundle to be purchased by a positive mass of buyers.

\begin{proposition}
\label{Prop_Pooling}
Let $\mathcal{Q}$ be a convex polytope. Suppose that $\mu (\underline{\bd}(\mathcal{X}))=0$ and that $f$ is bounded away from zero on $\overline{\bd}(\mathcal{X})$. Then a bundle $\mathbf{q}\neq \mathbf{0}$ is pooling in every optimal mechanism if there exists $\mathbf{x}\in\mathcal{X}$, such that
\begin{equation}
\label{Eqn_Pooling}
\mathbf{x} \cdot(\mathbf{q}-\hat{\mathbf{q}})>0, \quad \forall\,\hat{\mathbf{q}}\in\mathcal{Q}\setminus\{\mathbf{q}\}.
\end{equation}
\end{proposition}

Intuitively, \eqref{Eqn_Pooling} says that $\mathbf{q}$ uniquely maximizes type $\mathbf{x}$'s gross utility among all feasible bundles. Since $\mathcal{Q}$ is a polytope, $\mathbf{q}$ must be a vertex of $\mathcal{Q}$. The proposition shows that every vertex identified in this way must be pooling, so the number of such vertices provides an almost distribution-free lower bound on the optimal menu size.

To see the intuition behind the proof, suppose that $\mathbf{q}^{*}$ is not pooling and slightly reduce its price. The resulting demand set has positive measure but becomes increasingly thin as the price reduction converges to zero. Its interior contribution to marginal revenue therefore vanishes faster than its negative top-boundary contribution, because $f$ is bounded away from zero. A further price reduction would then raise revenue, contradicting \eqref{Eqn_SOC}.

We next apply \Cref{Prop_Pooling} to a setting with an arbitrary number of products. Maintain the assumptions of the proposition, let $\mathcal{X}=[0,1]^{N}$, and let $\mathcal{N}=\{1,\dotsc,N\}$ denote the set of products. A type-$\mathbf{x}$ buyer's utility from consuming a subset $\mathcal{M}\subseteq\mathcal{N}$ is
\begin{equation*}
k_{\mathcal{M}}\left(\mathbf{x}\cdot\sum_{n\in\mathcal{M}}\mathbf{e}_{n}\right)=k_{\mathcal{M}}\sum_{n\in\mathcal{M}}x_{n},
\end{equation*}
where $k_{\mathcal{M}}>0$ for every nonempty $\mathcal{M}$ and $k_{\varnothing}=0$. The allocation space $\mathcal{Q}$ is the convex hull of
\[
\left\{k_{\mathcal{M}}\sum_{n\in\mathcal{M}}\mathbf{e}_{n}:\mathcal{M}\subseteq\mathcal{N}\right\}.
\]

Consider the type $\mathbf{x}=\sum_{n\in\mathcal{M}}\mathbf{e}_{n}$. This type obtains utility $|\mathcal{M}|k_{\mathcal{M}}$ from bundle $\mathcal{M}$ and utility $|\mathcal{M}\cap\mathcal{M}'|k_{\mathcal{M}'}$ from bundle $\mathcal{M}'$. Since $\mathcal{Q}$ is the convex hull of these bundles, \Cref{Prop_Pooling} implies that $\mathcal{M}$ is pooling if
\begin{equation}
\label{Eqn_Pooling_Alt}
|\mathcal{M}|k_{\mathcal{M}}-|\mathcal{M}\cap\mathcal{M}'|k_{\mathcal{M}'}>0,\quad\forall\,\mathcal{M}'\subseteq\mathcal{N}\text{ such that }\mathcal{M}'\neq\mathcal{M}.
\end{equation}

\Cref{Cor_GeneralPooling} summarizes three major implications of \eqref{Eqn_Pooling_Alt}.

\begin{corollary}
\label{Cor_GeneralPooling}
\begin{enumerate}
\item \label{Cor_GeneralPooling_1}
Suppose that $k_{\mathcal{M}}\leq k_{\mathcal{M}'}$ for any $\mathcal{M}\subset\mathcal{M}'$, that is, all products are weak complements. Then the grand bundle $\mathcal{N}$ is pooling in every optimal mechanism.

\item \label{Cor_GeneralPooling_2}
Suppose that $k_{\mathcal{M}}>k_{\mathcal{M}'}$ for any $\varnothing\neq\mathcal{M}\subset\mathcal{M}'$, that is, all products are substitutes. Then bundles $\{1\},\dotsc,\{N\}$ are all pooling in every optimal mechanism.

\item \label{Cor_GeneralPooling_3}
Suppose that $\mathcal{M}^{*}$ is the unique maximizer of $|\mathcal{M}|k_{\mathcal{M}}$. Then $\mathcal{M}^{*}$ is pooling in every optimal mechanism.
\end{enumerate}
\end{corollary}

Under weak complementarity, a type who values every product strictly prefers the grand bundle because adding products weakly increases $k_{\mathcal{M}}$. Under substitutability, a type who values only product $n$ strictly prefers $\{n\}$: bundles excluding product $n$ provide no utility, while larger bundles containing it have a lower value multiplier. More generally, if $\mathcal{M}^{*}$ uniquely maximizes $|\mathcal{M}|k_{\mathcal{M}}$, then the type who values exactly the products in $\mathcal{M}^{*}$ strictly prefers that bundle to every alternative.

Each subset satisfying \eqref{Eqn_Pooling_Alt} adds a pooling bundle and therefore gives a lower bound on the optimal menu size. When $k_{\mathcal{M}}=1$ for every nonempty $\mathcal{M}$, we have $\mathcal{Q}=[0,1]^{N}$, and part \ref{Cor_GeneralPooling_1} implies that the grand bundle is always pooling. 
As a comparison, \citet[Theorem 16]{manelli2007multidimensional} establish that the grand bundle is pooling when $\mathcal{X}=\mathcal{Q}=[0,1]^{N}$, types are independently distributed, and the optimal mechanism is piecewise linear. To the best of our knowledge, \Cref{Prop_Pooling} is the first result of this kind to accommodate arbitrary optimal mechanisms on general type and allocation spaces.

%When $N=2$, $k_{\{1\}}=k_{\{2\}}=1$, and $k_{\{1,2\}}=k$, part \ref{Cor_GeneralPooling_3} gives the same conclusion whenever $k>1/2$. Thus, under the boundary-positivity assumption maintained here, the corollary recovers Corollaries \ref{Cor_GrandBundle} and \ref{Cor_GrandBundle_k}.\footnote{The assumption that $f$ is bounded away from zero on $\overline{\bd}(\mathcal{X})$ is implied by SRS. See \Cref{Ftn_Top_density}.}

%When the seller offers an arbitrary number of products, existing studies typically impose regularity assumptions on the type distribution or on the shape of each bundle's demand or profit function \citep{haghpanah2021pure,ghili2023characterization,yang2025nested}. \Cref{Prop_Pooling} takes an alternative route by providing an almost distribution-free characterization of pooling bundles. In doing so, it yields a lower bound on the optimal menu size and thus extends the menu-size conditions in \citet[Theorem 3]{daskalakis2017strong} to general type and allocation spaces.

\subsection{Active Bundles}
\label{Sec_Active}

This section applies (\hyperref[Eqn_FOC]{FOCs}) to address a question complementary to that in \Cref{Sec_Pooling}: which bundles can be active at all? Under strict regularity, the answer depends on a simple geometric alignment between the top boundary of the type space and the boundary of the allocation space.

\begin{proposition}
\label{Prop_Active}
Let $\mathcal{X}$ and $\mathcal{Q}$ be convex polytopes, and suppose that $f$ is strictly regular. Then in every optimal mechanism, a bundle $\mathbf{q}\neq \mathbf{0}$ is active only if there exists an outward unit normal $\mathbf{n}$ of some facet $\mathcal{S}(\mathbf{n})\subseteq\overline{\bd}(\mathcal{X})$, such that
\begin{equation}
\label{Eqn_Active}
\mathbf{q}+\tau\mathbf{n}\notin\mathcal{Q}, \quad \forall\,\tau>0.
\end{equation}
\end{proposition}

Economically, \Cref{Prop_Active} gives a directional version of the ``no distortion at the top'' property from unidimensional screening problems. An active bundle must exhaust feasibility in at least one direction valued most strongly by a top facet of types. Conditional on strict regularity, the candidate set characterized in \Cref{Prop_Active} depends only on the geometry of $\mathcal{X}$ and $\mathcal{Q}$ and excludes every interior bundle of $\mathcal{Q}$. When $\mathcal{X}=[0,1]^{N}$, an active bundle must be unaugmentable in at least one coordinate, recovering Proposition 2 in \citet{pavlov2011property}. If, in addition, $\mathcal{Q}=[0,1]^{N}$, this condition becomes $q_{n}=1$ for some $n$, recovering Lemma 2 in \citet{mcafee1988multidimensional}.

The role of general type and allocation spaces is illustrated by the model of selling $N$ identical products in \citet{bikhchandani2022selling,bikhchandani2024rank}. Let $x_{n}$ be the buyer's marginal value for the $n$-th unit and $q_{n}$ the probability of receiving that unit. Since the $(n+1)$-th unit can be allocated only if the $n$-th unit is allocated,
\begin{equation*}
\mathcal{Q}=\{\mathbf{q}\in[0,1]^{N}:q_{1}\geq\cdots\geq q_{N}\}.
\end{equation*}
For $c_{n}\in(0,1]$, the type spaces under decreasing marginal values (DMV) and increasing marginal values (IMV) are, respectively,
\begin{equation*}
\begin{aligned}
\text{DMV:}\quad \mathcal{X}&=\{\mathbf{x}\in[0,1]^{N}:x_{n+1}\leq c_{n}x_{n},\quad \forall\,1\leq n\leq N-1\},\\
\text{IMV:}\quad \mathcal{X}&=\{\mathbf{x}\in[0,1]^{N}:x_{n}\leq c_{n}x_{n+1},\quad \forall\,1\leq n\leq N-1\}.
\end{aligned}
\end{equation*}

Under DMV, the top boundary is a single facet $\mathcal{S}(\mathbf{e}_{1})=\{\mathbf{x}\in\mathcal{X}:x_{1}=1\}$. Since $\mathbf{q}+\tau\mathbf{e}_{1}$ is feasible for sufficiently small $\tau>0$ if and only if $q_{1}<1$, every active bundle must satisfy $q_{1}=1$: the first unit is allocated with probability one. Under IMV, the top boundary is a single facet $\mathcal{S}(\mathbf{e}_{N})=\{\mathbf{x}\in\mathcal{X}:x_{N}=1\}$. Since $\mathbf{q}+\tau\mathbf{e}_{N}$ is feasible for sufficiently small $\tau>0$ if and only if $q_{N}<q_{N-1}$, every active bundle must instead satisfy $q_{N}=q_{N-1}$. Thus, with the same allocation space, activity exhausts the upper-capacity constraint under DMV but the sequential-allocation constraint under IMV.

\begin{corollary}
\label{Cor_BM}
Let $N\geq 2$ and suppose that $f$ is strictly regular.
\begin{enumerate}
\item \label{Cor_BM_1} Under DMV, $\mathbf{q}$ is active only if $q_{1}=1$.
\item \label{Cor_BM_2} Under IMV, $\mathbf{q}$ is active only if $q_{N-1}=q_{N}$.
\end{enumerate}
\end{corollary}

When $N=2$, \Cref{Cor_BM} specializes to Propositions 1 and 7 of \citet{bikhchandani2022selling}.
%\Cref{Fig_BM} gives a graphical illustration of both cases.

%\begin{figure}[ht]
%\centering
%\begin{minipage}{0.4\textwidth}
%\begin{center}
%\input{Fig_BM1}
%\end{center}
%\end{minipage}%
%\begin{minipage}{0.4\textwidth}
%\begin{center}
%\input{Fig_BM2}
%\end{center}
%\end{minipage}
%\caption{Active bundles when $N=2$.}
%\label{Fig_BM}
%\end{figure}

\subsection{Exclusion Set}
\label{Sec_Exclusion}

The \emph{exclusion set}, $\mathcal{D}(\mathbf{0})$, consists of buyers who purchase the zero bundle. Economically, exclusion governs the extensive margin of screening: leaving some low types unserved allows the seller to extract more revenue from participating buyers. \Cref{Prop_Zero} identifies conditions under which this trade-off necessarily excludes a positive mass of buyers, so the zero bundle plays the same role as a pooling nonzero bundle.

\begin{proposition}
\label{Prop_Zero}
In every optimal mechanism,
\begin{equation}
\label{Eqn_ZeroMass}
\mu(\mathcal{D}(\mathbf{0}))\leq -1.
\end{equation}
Moreover, $\mathcal{D}(\mathbf{0})$ has positive measure with respect to $f$ if either $\mu(\underline{\bd}(\mathcal{X}))>-1$, or $N\geq2$ and $\mathcal{X}$ is strictly convex.
\end{proposition}

Inequality \eqref{Eqn_ZeroMass} says that the exclusion set must carry at least one unit of negative MR mass. If the exclusion set had zero measure with respect to $f$, it would have to lie on a single exposed face of the bottom boundary. This is impossible if the entire bottom boundary carries less than one unit of negative MR mass. It is also impossible when $N\geq2$ and $\mathcal{X}$ is strictly convex, because every exposed face is then a singleton and has zero $\mu$-measure.

When $\mathcal{X}=\mathcal{Q}=[0,1]^{N}$, \Cref{Prop_Zero} implies that $\mathbf{0}$ is always sold to a positive mass of buyers, thereby generalizing Theorem 16 of \citet{manelli2007multidimensional}.
In their Sections 6 and 7, DDT construct optimal mechanisms via an exclusion set in several examples, but do not establish when such a set must exist. In this sense, \Cref{Prop_Zero} can be viewed as a foundation for DDT's construction.
In an alternative setting where the allocation space is unbounded, \citet[Proposition 1]{armstrong1996multiproduct} establishes the existence of a nondegenerate exclusion set for a strictly convex type space. \citet{rochet2003economics}, however, argue that this conclusion is sensitive to the geometry of the type space. Although our model differs, \Cref{Prop_Zero} delivers a parallel result for strictly convex type spaces when $N\geq 2$. Examples with zero-measure exclusion sets are also provided in \citet{thirumulanathan2019optimal2}.

\subsection{Optimality of Bundling Strategies}

This section extends \Cref{Sec_Additive_Bundle} to an arbitrary number of products.
Rather than characterizing the entire optimal mechanism, we ask whether two simple benchmark strategies can withstand local price deviations. For additive buyers, we study pure bundling; for unit-demand buyers, we study symmetric separate selling, in which only $\mathbf{e}_{1},\dotsc,\mathbf{e}_{N}$ are active and sold at the same price.

We begin with additive buyers, which directly extends the setting in \Cref{Sec_Additive}.
\Cref{Prop_Pure_N} provides three necessary conditions for the optimality of pure bundling, where $\mathbf{x}_{-n}$ denotes a generic element of $[\underline{x},\overline{x}]^{N-1}$.

\begin{proposition}
\label{Prop_Pure_N}
Let $N\geq 2$, $\mathcal{X}=[\underline{x},\overline{x}]^{N}$, and $\mathcal{Q}=[0,1]^{N}$. Then a pure bundling mechanism $p$ is optimal only if
\begin{equation}
\label{Eqn_Pure_N_1}
\mu(\mathcal{D}(\mathbf{1}))=0.
\end{equation}
If, in addition, $f$ is positive almost everywhere on $\overline{\bd}(\mathcal{X})$, then
\begin{equation}
\label{Eqn_Pure_N_2}
p(\mathbf{1})-(N-1)\underline{x}\leq\overline{x}.
\end{equation}
If, further, $\mu(\underline{\bd}(\mathcal{X}))=0$, then
\begin{equation}
\label{Eqn_Pure_N_3}
\int^{\overline{x}}_{p(\mathbf{1})-(N-1)\underline{x}}
\phi(x_{n},\underline{\mathbf{x}}_{-n})\dd x_{n}
-\overline{x}f(\overline{x},\underline{\mathbf{x}}_{-n})
\geq 0,
\quad \forall\,n\in\mathcal{N}.
\end{equation}
\end{proposition}

Intuitively, \eqref{Eqn_Pure_N_1} follows from \eqref{Eqn_FOC}. If \eqref{Eqn_Pure_N_2} is violated, we can show that a small discount on product $n$ alone constitutes a profitable deviation, thereby violating \eqref{Eqn_FOC-}. The left-hand side of \eqref{Eqn_Pure_N_3} is the limit of ``MR density'' of product $n$, and \eqref{Eqn_SOC} requires this limiting marginal revenue to be nonnegative.

Whenever the assumptions in the third part of \Cref{Prop_Pure_N} hold and $f$ is strictly regular, we may define the generalized
type-dependent posted price
\begin{equation*}
\zeta(\mathbf{x}_{-n})
=\sup\left\{\zeta\in[\underline{x},\overline{x}]:
\int_{\zeta}^{\overline{x}}\phi(x_n,\mathbf{x}_{-n})\dd x_n
-\overline{x}f(\overline{x},\mathbf{x}_{-n})\geq0\right\},
\end{equation*}
where we adopt the convention $\sup\varnothing=\underline{x}$. Then \eqref{Eqn_Pure_N_3} can be written as
\begin{equation*}
p(\mathbf{1})-(N-1)\underline{x}
\leq\zeta(\underline{\mathbf{x}}_{-n}),
\quad \forall\,n\in\mathcal{N}.
\end{equation*}
Thus, the implicit cutoff for product $n$ under pure bundling cannot exceed the zero-marginal-revenue cutoff obtained when all other valuations are at their lowest levels.

%In a two-product setting, \citet[Proposition 3]{menicucci2015optimality} provide sufficient conditions for the optimality of pure bundling, which require that the type space is bounded away from the origin.
%DDT's Theorem 6 shows that pure bundling is suboptimal for additive buyers with uniformly distributed types when $N$ is sufficiently large. \Cref{Cor_Pure_N} generalizes this result to strictly regular type distributions.

%Consider a specific setting where $\mathcal{X}=[\underline{x},\underline{x}+1]^{N}$, and $x_{1},\dotsc,x_{N}$ are i.i.d. random variables with a strictly regular p.d.f. $g$.  In other words, pure bundling cannot be optimal when $N$ is sufficiently large.

When $N$ is sufficiently large, DDT's Theorem 6 shows that pure bundling is suboptimal for additive buyers with uniformly distributed types. \Cref{Cor_Pure_N} extends this result to any fixed i.i.d.\ marginal density satisfying mild regularity and boundary-positivity conditions.

\begin{corollary}
\label{Cor_Pure_N}
Fix $0\leq\underline{x}<\overline{x}$ and let $g$ be a p.d.f.\ satisfying the standing smoothness and integrability assumptions and strictly positive on $(\underline{x},\overline{x}]$.
For each $N\geq 2$, let $\mathcal{X}=[\underline{x},\overline{x}]^{N}$, $\mathcal{Q}=[0,1]^{N}$, and let $x_{1},\dotsc,x_{N}$ be i.i.d.\ with p.d.f.\ $g$.
Then there exists an integer $N^{*}$ such that pure bundling is suboptimal whenever $N\geq N^{*}$.
\end{corollary}

We next consider unit-demand buyers, which extends the perfect-substitutes environment with $k=1/2$ in \Cref{Sec_SubsComp}.
\Cref{Prop_Separate_N} provides two necessary conditions for the optimality of a symmetric separate selling mechanism.

\begin{proposition}
\label{Prop_Separate_N}
Let $N\geq2$, $\mathcal{X}=[\underline{x},\overline{x}]^{N}$, and $\mathcal{Q}=\{\mathbf{q}\in[0,1]^{N}:\mathbf{q}\cdot\mathbf{1}\leq 1\}$.
Then a symmetric separate selling mechanism $p$ is optimal only if
\begin{align}
\label{Eqn_Separate_N_1}
\mu(\mathcal{D}(\mathbf{e}_{n})) & =0,\quad \forall\,n\in\mathcal{N},\\
\label{Eqn_Separate_N_2}
\int_{p(\mathbf{e}_{n})}^{\overline{x}}\phi(x,\dotsc,x)\dd x-\overline{x}f(\overline{\mathbf{x}}) & \geq 0,\quad \forall\,n\in\mathcal{N}.
\end{align}
\end{proposition}

The intuition of \Cref{Prop_Separate_N} is similar to that of \Cref{Prop_Pure_N}: \eqref{Eqn_Separate_N_1} follows from \eqref{Eqn_FOC}. \eqref{Eqn_Separate_N_2} follows from applying \eqref{Eqn_SOC} to bundle $(1/N)\mathbf{1}$. Again, when $f$ is strictly regular, it is convenient to define the diagonal
counterpart of the type-dependent posted price:
\begin{equation*}
\zeta_{1/N}
=\sup\left\{\zeta\in[\underline{x},\overline{x}]:
\int_{\zeta}^{\overline{x}}\phi(x,\dotsc,x)\dd x
-\overline{x}f(\overline{\mathbf{x}})\geq0\right\},
\end{equation*}
where we adopt the convention $\sup\varnothing=\underline{x}$. Then \eqref{Eqn_Separate_N_2} becomes
\begin{equation*}
p(\mathbf{e}_n)\leq\zeta_{1/N},
\quad \forall\,n\in\mathcal{N}.
\end{equation*}
Thus, the common separate-selling price cannot exceed the highest price that is protected against entry by a discounted equal-probability lottery.

%Consider a specific setting where $\mathcal{X}=[\underline{x},\underline{x}+1]^{N}$ and $x_{1},\dotsc,x_{N}$ are i.i.d. random variables with a nonincreasing p.d.f. $g$. 
When buyers are unit-demand with uniformly distributed types, \citet[Theorem 2]{hashimoto2025selling} prove that symmetric separate selling is optimal if and only if $\underline{x}$ is sufficiently small. \Cref{Cor_Separate_N} complements their findings by showing that symmetric separate selling is suboptimal when $\underline{x}$ is large and $g$ is nonincreasing.

For i.i.d.\ uniform values on $[c,c+1]^{N}$, \citet[Theorem 2]{hashimoto2025selling} derive an exact cutoff $c^{*}\in (0,N)$ such that symmetric separate selling is optimal if and only if $c\leq c^{*}$. \Cref{Cor_Separate_N} sacrifices sharpness for robustness: it requires only that the joint density be nonincreasing along the main diagonal and therefore does not require independence or symmetry.

%By invoking \Cref{Prop_Separate_N}, we complement the finding of \citet{hashimoto2025selling} by showing that symmetric separate selling is suboptimal when $\underline{x}$ is large and $g$ is nondecreasing. For any $c\in [\underline{x},\overline{x}]$, 

%\begin{equation*}
%\int^{\overline{x}}_{c}\phi (x,\dotsc,x)\dd x-\overline{x}f(\overline{\mathbf{x}})=N\int^{\overline{x}}_{c}g(x)^{N}\dd x-cg(c)^{N}\leq (N(\overline{x}-c)-c)g(c)^{N},
%\end{equation*}
%where the right-hand side is nonpositive if $N\leq \underline{x}$. That is, symmetric separate selling cannot be optimal when $N$ is sufficiently small.

%Propositions \ref{Prop_Pure_N} and \ref{Prop_Separate_N} can be used to analyze how changes in the number of products affect the optimality of pure bundling and separate selling. To wit, consider a situation where 

%Suppose that $\mathcal{Q}=\{\mathbf{q}\in [0,1]^{N}:\mathbf{q}\cdot\mathbf{1}\leq 1\}$ and $g$ is nonincreasing. 

%\Cref{Cor_N} summarizes our results from the two cases.

\begin{corollary}
\label{Cor_Separate_N}
Let $N\geq 2$, $\mathcal{X}=[\underline{x},\overline{x}]^{N}$, and $\mathcal{Q}=\{\mathbf{q}\in[0,1]^{N}:\mathbf{q}\cdot\mathbf{1}\leq1\}$. Suppose that $x\mapsto f(x,\dotsc,x)$ is nonincreasing. Then symmetric separate selling is suboptimal if $N\leq \underline{x}/(\overline{x}-\underline{x})$.
\end{corollary}

%\begin{equation*}
%\begin{aligned}
%-\mu (\mathcal{D}(\mathbf{0})) & \geq (p(\mathbf{e}_{1})-\underline{x})^{N}(N\underline{x}\underline{\xi} +N+1)\underline{g}^{N}+N\underline{x}(p(\mathbf{e}_{1})-\underline{x})^{N-1}\underline{g}^{N} \\
%& \leq N\left((\underline{x}+1)\overline{\xi} +1+\frac{1}{N}+\frac{\underline{x}}{p(\mathbf{e}_{1})-\underline{x}}\right)[p(\mathbf{e}_{1})-\underline{x}]^{N}\overline{g}^{N},
%\end{aligned}
%\end{equation*}
%where the right-hand side of the second inequality goes to infinity as $N\to +\infty$. However, by \eqref{Eqn_Separate_N_1}, $\mu (\mathcal{D}(\mathbf{0}))=-1$, implying the existence of a cutoff $N^{**}$ such that $N<N^{**}$

As far as we know, Propositions \ref{Prop_Pure_N} and \ref{Prop_Separate_N} are both new, except that \eqref{Eqn_Separate_N_2} with $N=2$ can be inferred from \citet[Theorem 1]{thanassoulis2004haggling}. Although the insight bears resemblance to \eqref{Eqn_SOC}, \citet{thanassoulis2004haggling} restrict their analysis to unit-demand buyers with two products and do not employ the $\zeta$ function to study other settings as we do in \Cref{Sec_SubsComp}.

%\footnote{Theorem 1 in \citet{thanassoulis2004haggling} is derived via perturbing the price of some stochastic bundle $(q,1-q)$. } 
%In \Cref{Sec_Verify_2D}, we revisit the problem of selling two products in the SRS environment. With the help of $\zeta$, we establish sufficient and necessary conditions for the optimality of various bundling strategies.

\section{Concluding Remarks}
\label{Sec_Conclusion}

This paper demonstrates how FOC and SOC based on price perturbations pin down optimal selling mechanisms of a multiproduct monopoly. For SRS distributions, the approach characterizes the optimal mechanisms: Depending on the shape of $\zeta$, they feature pure bundling, mixed bundling, and probabilistic bundling with either finite or infinite menu sizes.\footnote{We also conjecture that a similar approach can be applied recursively to derive explicit solutions for specific high-dimensional distributions such as the $N$-item uniform distribution analyzed in \citet{giannakopoulos2018duality}. The challenge lies in how to utilize \eqref{Eqn_SOC} for $N\geq 3$, as the demand sets for inactive bundles can take various shapes in low dimensions.} As a result, we illustrate how a multiproduct monopolist accounts for multidimensional characteristics such as correlation and substitution or complementarity. Conversely, our approach also nails down conditions on the type distribution for a specific selling mechanism to be optimal. For an arbitrary number of products, FOC and SOC lead to novel insights on pooling bundles and active bundles for general distributions. 

Our paper complements the recent literature on optimal informationally robust selling mechanisms for multiproduct monopoly. For instance, \citet{carroll2017robustness}, \citet{deb2023multi}, and \citet{che2025robustly} solve optimal multiproduct mechanisms that maximize the seller's worst-case revenue over information structures subject to partial distributional information such as marginal, moment, or mean-preserving contraction constraints. 
%Depending on the constraints, their results identify separate selling, pure bundling, or categorical bundling as optimal informationally robust mechanisms. %These papers appeal to a saddle-point argument via finding a minmax information structure and a maxmin mechanism whose optimal revenue coincides. 
%In contrast, our result contributes to Bayesian mechanism design which demonstrates how optimal selling mechanisms factor in distributional information for bunching or bundling. 
Comparing optimal Bayesian mechanisms with their optimal informationally robust mechanisms identifies the value of acquiring distributional information beyond moments or marginals, as well as how more sophisticated bundling strategies account for the full distributional information. 

Second, our applications connect the MR approach to broader research on regulation and competition in multiproduct markets. Introducing the quality dimension to our analysis would yield a natural multidimensional extension of \citet{loertscher2022monopoly}. Recent contributions by \citet{armstrong2010competitive,zhou2017competitive,zhou2021mixed} study various bundling strategies in multiproduct oligopolies. Comparing the optimal revenue of a merged multiproduct monopolist with the equilibrium revenues of competing firms would isolate the merger's revenue synergy. Moreover, our characterization of optimal mechanisms for substitutes and complements suggests extending this comparison to settings in which richer bundling strategies arise, along the lines of \citet{armstrong2013more}.

Finally, beyond the multiproduct monopoly setting, there are other multidimensional (and even single-dimensional) screening problems in which randomization is known to improve the objective. Examples include delegation \citep[e.g.,][]{kovac2009stochastic,kleiner2022optimal}, costly state verification \citep[e.g.,][]{mookherjee1989optimal}, and adverse selection with type-dependent participation constraints \citep[e.g.,][]{jullien2000participation}. The FOC-SOC approach may likewise offer insights into the sources and magnitude of the gains from randomization in these environments.

\bibliographystyle{aea}
\bibliography{Screening,InfoDesign,CSV,General,Math}

\clearpage
\appendix

\counterwithin{lemma}{section}
\counterwithin{corollary}{section}
\counterwithin{proposition}{section}
\counterwithin{example}{section}

\numberwithin{equation}{section}
\counterwithin{figure}{section}

\begin{center}
\LARGE Appendix
\end{center}

\section{Omitted Proofs}

\setcounter{subsection}{-1}
\subsection{Preliminary properties of demand sets}

As preparation, this section collects some properties of demand sets that will be used repeatedly below. \Cref{Lem_Cont} speaks to the ``continuity'' of demand sets, whereas \Cref{Lem_NormalCone} speaks to the ``monotonicity'' of them.

\begin{lemma}
\label{Lem_Cont}
Let $p$ be a feasible mechanism and let $Q\subseteq\mathcal{Q}$ be nonempty and closed, with $\mathbf{0}\notin Q$. For every $\varepsilon>0$, write $\mathcal{D}_{\varepsilon}(Q)=\mathcal{D}(Q;\overline{p}-\varepsilon\mathbb{I}_{Q})$. Then
\begin{enumerate}
\item \label{Lem_Cont_Outer}
For every $\delta>0$, there exists $\overline{\varepsilon}>0$ such that $\mathcal{D}_{\varepsilon}(Q)\subseteq B_{\delta}(\mathcal{D}(Q))$ whenever $\varepsilon\in(0,\overline{\varepsilon})$.
\item \label{Lem_Cont_Inner}
For every $\varepsilon>0$, there exists $\delta>0$ such that $B_{\delta}(\mathcal{D}(Q))\subseteq\mathcal{D}_{\varepsilon}(Q)$.
\end{enumerate}
\end{lemma}

\begin{proof}
For every $\mathbf{x}\in\mathcal{X}$, let
\begin{equation*}
u_{Q}(\mathbf{x})=\max_{\mathbf{q}\in Q}\{\mathbf{x}\cdot\mathbf{q}-\overline{p}(\mathbf{q})\}.
\end{equation*}
Since $Q$ is compact, $u_{Q}$ is continuous. Hence $u-u_{Q}$ is continuous and nonnegative, with
\begin{equation*}
\begin{aligned}
\mathcal{D}(Q) & =\{\mathbf{x}\in\mathcal{X}:u(\mathbf{x})=u_{Q}(\mathbf{x})\}, \\
\mathcal{D}_{\varepsilon}(Q) & =\{\mathbf{x}\in\mathcal{X}:u(\mathbf{x})-u_{Q}(\mathbf{x})\leq\varepsilon\}.
\end{aligned}
\end{equation*}

Fix $\delta>0$. If $\mathcal{X}\setminus B_{\delta}(\mathcal{D}(Q))$ is empty, part \ref{Lem_Cont_Outer} is immediate. Otherwise, continuity and compactness imply
\begin{equation*}
\overline{\varepsilon}
:=\min_{\mathbf{x}\in\mathcal{X}\setminus B_{\delta}(\mathcal{D}(Q))}
\{u(\mathbf{x})-u_{Q}(\mathbf{x})\}>0.
\end{equation*}
By construction, if $\varepsilon\in (0,\overline{\varepsilon})$, no $\mathbf{x}\in \mathcal{X}\setminus B_{\delta}(\mathcal{D}(Q))$ belongs to $\mathcal{D}_{\varepsilon}(Q)$. This proves part \ref{Lem_Cont_Outer}.

Fix $\varepsilon>0$. The continuity of $u-u_{Q}$ implies the existence of $\delta>0$ such that $u(\mathbf{x})-u_{Q}(\mathbf{x})<\varepsilon$ whenever $\mathbf{x}\in B_{\delta}(\mathcal{D}(Q))\cap\mathcal{X}$. This proves part \ref{Lem_Cont_Inner}.
\end{proof}

For any $\mathbf{q}\in\mathcal{Q}$, denote by $\NC(\mathbf{q})$ the \emph{normal cone} of $\mathbf{q}$ to $\mathcal{Q}$, consisting of all outward normals to $\mathcal{Q}$ at $\mathbf{q}$.

\begin{lemma}
\label{Lem_NormalCone}
For any feasible mechanism, if $\mathbf{x}\in \mathcal{D}(\mathbf{q})$ for some $\mathbf{q}\in \bd (\mathcal{Q})$ and $\mathbf{n}$ is an outward normal to $\mathcal{Q}$ at $\mathbf{q}$, then for any $\tau>0$ such that $(\mathbf{x}+\tau \mathbf{n})\in \mathcal{X}$, $(\mathbf{x}+\tau \mathbf{n})\in \mathcal{D}(\mathbf{q})$. As a result, $((\mathbf{x}+\NC (\mathbf{q}))\cap \mathcal{X})\subseteq \mathcal{D}(\mathbf{q})$.
\end{lemma}

\begin{proof}
For any $\hat{\mathbf{q}}\in \mathcal{Q}$, $\mathbf{x}\cdot \hat{\mathbf{q}}-\overline{p}(\hat{\mathbf{q}})\leq \mathbf{x}\cdot \mathbf{q}-\overline{p}(\mathbf{q})$ because $\mathbf{x}\in \mathcal{D}(\mathbf{q})$. Therefore,
\begin{equation*}
[(\mathbf{x}+\tau \mathbf{n})\cdot \hat{\mathbf{q}}-\overline{p}(\hat{\mathbf{q}})]-[(\mathbf{x}+\tau \mathbf{n})\cdot \mathbf{q}-\overline{p}(\mathbf{q})]\leq \tau \mathbf{n}\cdot (\hat{\mathbf{q}}-\mathbf{q})\leq 0,
\end{equation*}
implying that $(\mathbf{x}+\tau \mathbf{n})\in \mathcal{D}(\mathbf{q})$.
\end{proof}

\subsection{Proof of Lemma \ref{Lem_MR}}

First, fix a closed set $Q$ with $\mathbf{0}\notin Q$. Let $Q^{c}=\mathcal{Q}\setminus Q$. For every $\mathbf{x}\in \mathcal{X}$, let
\begin{equation*}
\begin{aligned}
u_{Q}(\mathbf{x}) & =\max_{\hat{\mathbf{q}}\in Q}\{\mathbf{x}\cdot \hat{\mathbf{q}}-\overline{p}(\hat{\mathbf{q}})\}, \\
u_{Q^{c}}(\mathbf{x}) & =\sup_{\hat{\mathbf{q}}\in Q^{c}}\{\mathbf{x}\cdot \hat{\mathbf{q}}-\overline{p}(\hat{\mathbf{q}})\},
\end{aligned}
\end{equation*}
where the maximum defining $u_{Q}$ is attained because $Q$ is closed.
It can be verified that $u=\max \{u_{Q},u_{Q^{c}}\}$, and for $\varepsilon>0$, $u_{-\varepsilon,Q}=\max \{u_{Q}+\varepsilon,u_{Q^{c}}\}$. It follows that
\begin{equation*}
\lim_{\varepsilon\to 0}\left(\frac{u_{-\varepsilon,Q}-u}{-\varepsilon}\right)=-\mathbb{I}_{\{u_{Q}\geq u_{Q^{c}}\}}=-\mathbb{I}_{\{u=u_{Q}\}}=-\mathbb{I}_{\mathcal{D}(Q)}.
\end{equation*}
For every $\varepsilon>0$, the difference quotient $\left(\frac{u_{-\varepsilon,Q}-u}{-\varepsilon}\right)$ belongs to $[-1,0]$. Since $|\mu|$ is finite, Lebesgue's dominated convergence theorem gives
\begin{equation*}
\MR_{-}(Q)=\lim_{\varepsilon \to 0}\int_{\mathcal{X}}\left(\frac{u_{-\varepsilon,Q}-u}{-\varepsilon}\right)\dd \mu =-\mu (\mathcal{D}(Q)).
\end{equation*}

Next, fix $\delta>0$ sufficiently small such that $\mathbf{0}\notin B_{\delta}(Q)$, and let $B_{\delta}^{c}(Q)=\mathcal{Q}\setminus B_{\delta}(Q)$. For every $\mathbf{x}\in\mathcal{X}$, let
\begin{equation*}
\begin{aligned}
u_{B_{\delta}(Q)}(\mathbf{x})&=\sup_{\hat{\mathbf{q}}\in B_{\delta}(Q)}\{\mathbf{x}\cdot\hat{\mathbf{q}}-\overline{p}(\hat{\mathbf{q}})\},\\
u_{B_{\delta}^{c}(Q)}(\mathbf{x})&=\max_{\hat{\mathbf{q}}\in B_{\delta}^{c}(Q)}\{\mathbf{x}\cdot\hat{\mathbf{q}}-\overline{p}(\hat{\mathbf{q}})\},
\end{aligned}
\end{equation*}
where the maximum defining $u_{B_{\delta}^{c}(Q)}$ is attained because $B_{\delta}^{c}(Q)$ is closed. It can be verified that
\begin{equation*}
u=\max\{u_{B_{\delta}(Q)},u_{B_{\delta}^{c}(Q)}\},\qquad
u_{+\varepsilon, B_{\delta}(Q)}   =\max\{u_{B_{\delta}(Q)}-\varepsilon,u_{B_{\delta}^{c}(Q)}\}.
\end{equation*}
It follows that
\begin{equation*}
\lim_{\varepsilon\downarrow0}\frac{u_{+\varepsilon, B_{\delta}(Q)}   -u}{\varepsilon}
=-\mathbb{I}_{\mathcal{X}\setminus\mathcal{D}(B_{\delta}^{c}(Q))}.
\end{equation*}
Since the difference quotient belongs to $[-1,0]$, Lebesgue's dominated convergence theorem gives
\begin{equation*}
\lim_{\varepsilon\downarrow0}\frac{\TR(\overline{p}+\varepsilon\mathbb{I}_{B_{\delta}(Q)})-\TR(\overline{p})}{\varepsilon}
=-\mu(\mathcal{X}\setminus\mathcal{D}(B_{\delta}^{c}(Q))).
\end{equation*}
Moreover,
\begin{equation*}
\mathcal{X}\setminus\mathcal{D}(B_{\delta}^{c}(Q))
=\{\mathbf{x}\in\mathcal{X}:\partial_{\mathcal{Q}} u(\mathbf{x})\subseteq B_{\delta}(Q)\}.
\end{equation*}
Since $Q$ is closed, as $\delta\downarrow0$ these sets decrease to $\mathcal{D}^{*}(Q)$. Continuity from above of the finite signed measure $\mu$ therefore yields
\begin{equation*}
\MR_{+}(Q)=-\mu(\mathcal{D}^{*}(Q)).
\end{equation*}
The last statement of the lemma immediately follows from the definition of MR.

\begin{comment}
Now we are ready to compute $\MR_{+}(Q)$. For all sufficiently small $\delta>0$, $B_{\delta}(Q)$ is relatively open and does not contain $\mathbf{0}$. By previous analysis,
\begin{equation*}
\MR_{+}(B_{\delta}(Q))=-\mu (\mathcal{X}\setminus\mathcal{D}(B_{\delta}(Q)^{c})).
\end{equation*}
Indeed, $\mathcal{X}\setminus\mathcal{D}(B_{\delta}(Q)^{c})$ consists precisely of types whose demanded bundles all lie in $B_{\delta}(Q)$, and $\bigcap_{\delta>0}B_{\delta}(Q)=Q$. Hence, as $\delta\to 0$, $\mathcal{X}\setminus\mathcal{D}(B_{\delta}(Q)^{c})\downarrow \mathcal{D}^{*}(Q)$.
Continuity from above of the finite signed measure $\mu$ therefore yields $\MR_{+}(Q)=-\mu (\mathcal{D}^{*}(Q))$.
\end{comment}

\subsection{Proof of Theorem \ref{Thm_MR}}

Note that \hyperref[Eqn_FOC]{(FOCs)} follow immediately from the optimality of $\overline{p}$, so we only prove \eqref{Eqn_SOC}.

By the proof of \Cref{Lem_MR}, for every $\varepsilon,\varepsilon'>0$,
\begin{equation*}
|u_{-\varepsilon,Q}-u_{-\varepsilon',Q}|=|\max \{u_{Q}+\varepsilon,u_{Q^{c}}\}-\max \{u_{Q}+\varepsilon',u_{Q^{c}}\}|\leq |\varepsilon -\varepsilon'|.
\end{equation*}
Consequently,
\begin{equation*}
|\TR (\overline{p}-\varepsilon \mathbb{I}_{Q})-\TR (\overline{p}-\varepsilon' \mathbb{I}_{Q})|=\left|\int_{\mathcal{X}}(u_{-\varepsilon,Q}-u_{-\varepsilon',Q})\dd \mu \right|\leq |\varepsilon -\varepsilon'|\cdot |\mu|(\mathcal{X}).
\end{equation*}
Hence, $\TR (\overline{p}-\varepsilon \mathbb{I}_{Q})$ as a function of $\varepsilon$ is Lipschitz and is thus absolutely continuous.

Since $\overline{p}$ is optimal, for every $\overline{\varepsilon}>0$,
\begin{equation*}
\int^{\overline{\varepsilon}}_{0}\TR' (\overline{p}-\varepsilon \mathbb{I}_{Q})\dd \varepsilon =\TR (\overline{p}-\overline{\varepsilon} \mathbb{I}_{Q})-\TR (\overline{p})\leq 0.
\end{equation*}
Hence, there exists $\varepsilon \in(0,\overline{\varepsilon})$ such that $\TR'(\overline{p}-\varepsilon \mathbb{I}_{Q})$ exists and is nonpositive. By definition, $\TR'(\overline{p}-\varepsilon \mathbb{I}_{Q})=-\MR_{-}(Q;\overline{p}-\varepsilon \mathbb{I}_{Q})$.

\subsection{Proof of Corollary \ref{Cor_MR}}

By \Cref{Lem_MR} and \Cref{Thm_MR}, $\MR_{-}(Q)-\MR_{+}(Q)=-\mu (\mathcal{D}(Q)\setminus \mathcal{D}^{*}(Q) )\geq 0$. Note that every type in $\mathcal{D}(Q)\setminus \mathcal{D}^{*}(Q)$ demands both a bundle in $Q$ and a distinct bundle outside $Q$. By Theorem 25.5 of \citet{rockafellar1970}, this set has zero Lebesgue measure, so the interior part of $\mu$ vanishes on it. Moreover, since $\mu(\underline{\bd}(\mathcal{X}))=0$, the contribution from the bottom boundary also vanishes, while the remaining top boundary contribution is nonnegative. Therefore, $-\mu (\mathcal{D}(Q)\setminus \mathcal{D}^{*}(Q) )\leq 0$. Combining the two inequalities yields $\mu (\mathcal{D}(Q)\setminus \mathcal{D}^{*}(Q) )=0$.

\subsection{Proof of Propositions \ref{Prop_Additive} and \ref{Prop_SubsComp}}
\label{Prf_Prop1&2}

\noindent \emph{Proof of part \ref{Prop_Additive_2} of both propositions:}
Let $(q,1)_{k}$ be pooling. If $q<1$, its demand set is given by \eqref{Eqn_D_k} except for the case where $k>1$ and $a^{q}=0$. Applying \Cref{Cor_MR} to $\{(q,1)_{k}\}$ gives \eqref{Eqn_IntPhi} and \eqref{Eqn_IntPhi_k}. If $q=1$, the same results follow from symmetry.

We next establish the endpoint inequalities.
Suppose first that $q>0$ and that either $a^{q}>0$ or $k\leq 1$. If $\Phi_{k}(a^{q})<0$, then, by continuity, there exists $\delta>0$ sufficiently small such that
\begin{equation}
\label{Eqn_Endpoint_Neg}
\max(1-\overline{u}(x),z_{k}(x))-\zeta_{k}(x)>\delta
\end{equation}
for $x$ in a neighborhood of $a^{q}$, where the neighborhood is one-sided when $a^{q}=0$.

Suppose first that $a^{q}>0$, and take $\delta<a^{q}$. Since $a^{q}$ is the left endpoint of $\{x:\overline{u}_{-}'(x)\leq qk\leq\overline{u}_{+}'(x)\}$, we have $\overline{u}_{+}'(a^{q}-\delta)<qk$. By the monotonicity of $\partial\overline{u}$, there exists $\hat{q}<q$ satisfying $\overline{u}_{+}'(a^{q}-\delta)<\hat{q}k<qk$ and $[a^{\hat{q}},b^{\hat{q}}] \subset(a^{q}-\delta,a^{q}]$. Indeed, types $x\leq a^{q}-\delta$ have $\overline{u}_{+}'(x)<\hat qk$, while types $x>a^{q}$ have $\overline{u}_{-}'(x)\geq qk>\hat qk$.

Let $Q=\{(\hat q,1)_{k},(1,\hat q)_{k}\}$. Since $Q$ is closed and does not contain $\mathbf{0}$, $\overline{p}-\varepsilon\mathbb{I}_{Q}$ is feasible for every $\varepsilon>0$. By \Cref{Lem_Cont}, for all sufficiently small $\varepsilon>0$, $\mathcal{D}((\hat{q},1)_{k})$ and $\mathcal{D}((1,\hat{q})_{k})$ under the perturbed price schedule $\overline{p}-\varepsilon\mathbb{I}_{Q}$ remain arbitrarily close to $\mathcal{D}((\hat{q},1)_{k})$ and $\mathcal{D}((1,\hat{q})_{k})$ under the optimal price schedule $\overline{p}$. Consequently,
\begin{equation*}
\MR_{-}(Q;\overline{p}-\varepsilon\mathbb{I}_{Q})=-\mu (\mathcal{D}(Q;\overline{p}-\varepsilon\mathbb{I}_{Q}))<0
\end{equation*}
for every sufficiently small $\varepsilon >0$, contradicting \eqref{Eqn_SOC}. Hence $\Phi_{k}(a^{q})\geq 0$.
When $a^{q}=0$ and $k\leq 1$, the same argument applies using the one-sided neighborhood $[0,\delta)$ and a $\hat q<q$ sufficiently close to $q$.

The proof that $\Phi_{k}(b^{q})\geq 0$ when $q<1$ is analogous. If $\Phi_{k}(b^{q})<0$, choose instead a bundle $(\check{q},1)_{k}$ with $\check{q}>q$ sufficiently close to $q$. $Q=\{(\check{q},1)_{k},(1,\check{q})_{k}\}$ then has strictly negative marginal revenue for every sufficiently small perturbation, again contradicting \eqref{Eqn_SOC}. This argument also applies when $k>1$ and $a^{q}=0$, because only a neighborhood of $b^{q}$ is involved.

\noindent \emph{Proof of part \ref{Prop_Additive_3} of both propositions:}
Let $(q,1)_{k}$ be separating. Then $a^{q}=b^{q}$. Suppose first that $q\in (0,1)$ and that either $k\leq 1$ or $a^{q}>0$.
For every sufficiently small $\delta>0$, define the closed set $Q_{\delta}=\{(\hat q,1)_{k}:q-\delta\leq\hat q\leq q+\delta \}$.
By \Cref{Cor_MR} and \eqref{Eqn_D_k},
\begin{equation*}
-\mu(\mathcal{D}(Q_{\delta}))=\int^{b^{q+\delta}}_{a^{q-\delta}}\Phi_{k}(x_{1})\dd x_{1}=0.
\end{equation*}
Since $(q,1)_{k}$ is active, $a^{q-\delta}<b^{q+\delta}$ for arbitrarily small $\delta$. Hence, by the mean value theorem for integrals, for every such $\delta$ there exists $x_{\delta}\in [a^{q-\delta},b^{q+\delta}]$ such that $\Phi_{k}(x_{\delta})=0$. By the monotonicity and closed-graph property of $\partial\overline{u}$, $a^{q-\delta}\to a^{q}$ and $b^{q+\delta}\to b^{q}$, which imply $x_{\delta}\to a^{q}=b^{q}$ as $\delta \to 0$. Continuity of $\Phi_{k}$ then gives $\Phi_{k}(a^{q})=0$. The cases $q=0$ and $q=1$ follow analogously by using one-sided closed intervals of bundles.

\subsection{Proof of Corollary \ref{Cor_Inc}}

Denote by $x_{\zeta}$ the unique solution of $\zeta(x)=x$. Then $\zeta(x)>x$ for $x\in[0,x_{\zeta})$ and $\zeta(x)<x$ for $x\in(x_{\zeta},1]$. By \Cref{Cor_GrandBundle}, $\mathbf{1}$ is pooling, so $\max (1-\overline{u}(a^{1}),a^{1})\leq \zeta (a^{1})$. Together with \eqref{Eqn_IntPhi}, this implies $a^{1}<x_{\zeta}$. Therefore, on $[0,a^{1}]$, the sign of $\Phi$ is the same as the sign of $\zeta+\overline{u}-1$. Since both $\zeta$ and $\overline{u}$ are nondecreasing, $\zeta+\overline{u}-1$ is nondecreasing.

Suppose that $(q,1)$ is pooling for some $q\in (0,1)$. By part \ref{Prop_Additive_2} of \Cref{Prop_Additive}, $\Phi(a^{q})\geq 0$, so $\Phi\geq 0$ on $[a^{q},b^{q}]$. Equation \eqref{Eqn_IntPhi} then implies $\Phi= 0$ throughout $[a^{q},b^{q}]$, and so is $\zeta+\overline{u}-1$. As a result, both $\zeta$ and $\overline{u}$ are constant on this interval, contradicting $q>0$.

There cannot be any separating bundle $(q,1)$ with $q \in (0,1)$ either. Otherwise, one can always find active bundles arbitrarily close to $(q,1)$. Since no such bundle is pooling, part \ref{Prop_Additive_3} of \Cref{Prop_Additive} implies that $\zeta+\overline{u}-1=0$ on a sequence of points converging to $a^{q}$ (and also at $a^{q}$). Since $\zeta+\overline{u}-1$ is nondecreasing, both $\zeta$ and $\overline{u}$ must be constant on a nondegenerate interval that contains $a^{q}$,  contradicting separation.

Thus, the only possible active bundle $(q,1)$ with $q<1$ is $\mathbf{e}_{2}$.

\subsection{Proof of Corollary \ref{Cor_Cvx}}

Since $\zeta(x)-x$ is convex, $\zeta(0)\geq 0$, and $\zeta(1)<1$, $\zeta(x)=x$ has a unique solution, which we denote by $x_{\zeta}$. As in the proof of \Cref{Cor_Inc}, $a^{1}<x_{\zeta}$, so on $[0,a^{1}]$ the sign of $\Phi$ is the same as the sign of $\zeta+\overline{u}-1$. Since both $\zeta$ and $\overline{u}$ are convex, $\zeta+\overline{u}-1$ is convex.

Suppose that $(q,1)$ is pooling for some $q\in (0,1)$. If $\zeta+\overline{u}-1$ is not identically zero on $[a^{q},b^{q}]$, then, by part \ref{Prop_Additive_2} of \Cref{Prop_Additive}, $\zeta (a^{q})+\overline{u}(a^{q})-1\geq 0$ and $\zeta (b^{q})+\overline{u}(b^{q})-1\geq 0$. By \eqref{Eqn_IntPhi}, there exist $x\in [a^{q},b^{q}]$ such that $\zeta(x)+\overline{u}(x)-1<0$. By convexity, this implies $\zeta(x)+\overline{u}(x)-1>0$ for all $x\in [0,a^{q})\cup (b^{q},a^{1}]$, so there cannot be a second pooling bundle.

If $\zeta+\overline{u}-1=0$ on $[a^{q},b^{q}]$, then since $\zeta$ is convex and $1-\overline{u}$ is concave, both of them must be affine on $[a^{q},b^{q}]$. Moreover, $\overline{u}'=q$ on $(a^{q},b^{q})$ implies that $\zeta(x)+\overline{u}(x)-1>0$ for all $x\in [0,a^{q})\cup (b^{q},a^{1}]$. Again, there cannot be a second pooling bundle.

There cannot be any separating bundle $(q,1)$ with $q<1$ either. Otherwise, as in the proof of \Cref{Cor_Inc}, both $\zeta$ and $\overline{u}$ must be affine on a nondegenerate interval that contains $a^{q}$, contradicting separation.

Therefore, there is at most one active bundle $(q,1)$ with $q<1$.

\subsection{Proof of Corollaries \ref{Cor_Ccv} and \ref{Cor_CcvDec}}

Since $\zeta(x)-x$ is concave, and it can be verified that $\zeta(0)>0$ and $\zeta(1)<1$, $\zeta(x)=x$ has a unique solution, which we denote by $x_{\zeta}$. As in the proof of \Cref{Cor_Inc}, $a^{1}<x_{\zeta}$.

If $(q,1)$ is pooling with $q\in (0,1)$, then, by part \ref{Prop_Additive_2} of \Cref{Prop_Additive}, $\max (1-\overline{u}(x),x)\leq \zeta (x)$ for $x=a^{q},b^{q}$. On $[a^{q},b^{q}]$, the function $\overline{u}$ is affine, so $\max(1-\overline{u}(x),x)$ is convex. Since $\zeta$ is concave, there must be $\max(1-\overline{u}(x),x)\leq \zeta(x)$ throughout $[a^{q},b^{q}]$. By \eqref{Eqn_IntPhi} and the continuity of $\Phi$, this further implies $\max(1-\overline{u}(x),x)=\zeta(x)$ throughout $[a^{q},b^{q}]$. Since $\zeta (x)>x$ on $[0,a^{1}]$, there must be $1-\overline{ u}=\zeta$ on $[a^{q},b^{q}]$.
Since this holds for every pooling bundle, together with part \ref{Prop_Additive_3} of \Cref{Prop_Additive}, we have essentially proved that $1-\overline{u}=\zeta $ on the entire $[b^{0},a^{1}]$.

Suppose, in addition, that $\zeta$ is nonincreasing. Then $\overline{u}$ is constant on $[0,b^{0}]$, so $\max(1-\overline{u}(b^{0}),b^{0})\leq \zeta (b^{0})$ implies $\max (1-\overline{u}(0),0)\leq \zeta (0)$. Thus, $\mathbf{e}_{2}$ is pooling only if $1-\overline{u}(x)=\zeta (x)$ for all $x\in [0,b^{0}]$. This proves that $\overline{u}(x)=1-\zeta (x)$ for all $x\in [0,a^{1}]$.

%As depicted in \Cref{Fig_Prop_2D}, the single-dipped function $\max(1-\overline{u}(x),x)$ lies (weakly) below $\zeta (x)$ on $[a^{q},b^{q}]$. The only way to satisfy \eqref{Eqn_IntPhi} is that $1-\overline{u}(x)>x$ and $1-\overline{u}(x)=\zeta (x)$ for all $x\in [a^{q},b^{q}]$. This result and part \ref{Prop_Additive_3} of \Cref{Prop_Additive} jointly imply that $\overline{u}(x)=1-\zeta (x)$ for all $x\in [b^{0},a^{1}]$.

\subsection{Proof of Proposition \ref{Prop_Sufficiency}}
\label{Prf_Prop_Sufficiency}

With a slight abuse of notation, we write the total revenue $\TR$ as a functional of $\overline{u}$:
\begin{equation*}
\TR (\overline{u})=2\int^{1}_{0}\overline{u}(x)f(x,1)\dd x-2\int_{\{\mathbf{x}\in [0,1]^{2}:x_{1}\leq x_{2}\}}\max ( \overline{u}(x_{1})-(1-x_{2}),0)\phi (\mathbf{x})\dd \mathbf{x}.
\end{equation*}
The set of feasible $\overline{u}$ is
\begin{equation*}
\begin{aligned}
\overline{\mathcal{U}} = \{ \overline{u}:[0,1]\to \mathbb{R}_{+} \mid {} & \overline{u}\text{ is nondecreasing, convex, and 1-Lipschitz} \\
& \text{with respect to the $\ell_{1}$-norm}, \overline{u}(0)\leq 1 \}. \end{aligned}
\end{equation*}
Hence, our problem is equivalent to maximizing $\TR (\overline{u})$ over $\overline{\mathcal{U}}$. It can be verified that $\overline{\mathcal{U}}$ is convex and compact under the uniform norm. Its extreme points are given by
\begin{equation*}
\ext \overline{\mathcal{U}}=\{\overline{u}:[0,1]\to \mathbb{R}_{+} \mid \overline{u}(x)=\overline{u}(0)+\max (x-y,0),\text{ where }\overline{u}(0)\in \{0,1\},\ y\in [0,1]\}.
\end{equation*}
Also note that $\TR (\overline{u})$ is concave in $\overline{u}$. For each fixed $x_1$, the kink set $\{x_2:\overline u(x_1)=1-x_2\}$ contains at most one point and hence has two-dimensional Lebesgue measure zero. Since $\phi$ is integrable, dominated convergence gives the Gateaux derivative of $\TR (\overline{u})$ in direction $\psi$:
\begin{equation*}
\lim_{\delta\to 0}\frac{\TR (\overline{u}+\delta \psi )-\TR (\overline{u})}{\delta}=-2\int^{1}_{0}\Phi_{\overline{u}} (x)\psi (x)\dd x,
\end{equation*}
where $\Phi_{\overline{u}}$ is the pointwise marginal revenue $\Phi$ defined in \eqref{Eqn_Phi}; we attach the subscript $\overline{u}$ to highlight its dependence on $\overline{u}$.

The proof consists of several steps. First, we show that 
$\overline{u}$ is optimal if and only if
\begin{equation}
\label{Eqn_Linear}
\int^{1}_{0}\Phi_{\overline{u}}(x)(\overline{u}(x)-\hat{u}(x))\dd x\leq 0,\quad \forall\, \hat{u}\in \ext \overline{\mathcal{U}}.
\end{equation}

\emph{The ``only if'' direction}. 
For any $\hat{u}\in \overline{\mathcal{U}}$ and $\delta \in [0,1]$, $((1-\delta )\overline{u}+\delta \hat{u}) \in \overline{\mathcal{U}}$. By the optimality of $\overline{u}$,
\begin{equation*}
\frac{\TR ((1-\delta )\overline{u}+\delta \hat{u})-\TR (\overline{u})}{\delta }\leq 0, \quad \forall\, \delta \in (0,1].
\end{equation*}
Taking $\delta \to 0$ yields
\begin{equation*}
\int^{1}_{0}\Phi_{\overline{u}} (x)(\overline{u}(x)-\hat{u}(x))\dd x\leq 0,\quad \forall \hat{u}\in \overline{\mathcal{U}},
\end{equation*}
which implies \eqref{Eqn_Linear}.

\emph{The ``if'' direction}.
Since $\overline{\mathcal{U}}$ is convex, by Choquet's theorem, for any $\hat{u}\in \overline{\mathcal{U}}$, there exists a probability distribution $\nu$ supported on $\ext \overline{\mathcal{U}}$ such that $\hat{u}=\int_{\overline{\mathcal{U}}} u_{\ext} \dd \nu$.
By the concavity of $\TR$ and the Fubini-Tonelli theorem,
\begin{equation*}
\begin{aligned}
\TR(\overline{u}) & \geq \TR(\hat{u})-2\int^{1}_{0}\Phi_{\overline{u}} (x)(\overline{u}(x)-\hat{u}(x))\dd x \\
& =\TR(\hat{u})-2\int^{1}_{0}\Phi_{\overline{u}} (x)\left (\int_{\overline{\mathcal{U}}} (\overline{u}(x)-u_{\ext}(x)) \dd \nu \right)\dd x \\
& =\TR(\hat{u})-2\int_{\overline{\mathcal{U}}}\left (\int^{1}_{0}\Phi_{\overline{u}} (x)(\overline{u}(x)-u_{\ext}(x))\dd x \right)\dd \nu.
\end{aligned}
\end{equation*}
The optimality of $\overline{u}$ follows from \eqref{Eqn_Linear}.

Next, we show that \eqref{Eqn_Linear} implies \eqref{Eqn_R0} and \eqref{Eqn_Supp}.
For any $\hat{u}\in \overline{\mathcal{U}}$, integration by parts yields
\begin{equation*}
\begin{aligned}
\int^{1}_{0}\Phi_{\overline{u}}(x) \hat{u}(x)\dd x & =-\int^{1}_{0} \hat{u}(x)\dd \left(\int^{1}_{x}\Phi_{\overline{u}}(z)\dd z\right) \\
& =\hat{u}(0)\int^{1}_{0}\Phi_{\overline{u}}(z)\dd z +\int^{1}_{0} \left(\int^{1}_{x}\Phi_{\overline{u}}(z)\dd z\right)\hat{u}'(x)\dd x \\
& =\hat{u}(0)\int^{1}_{0}\Phi_{\overline{u}}(z)\dd z-\int^{1}_{0} \mathcal{R}_{\overline{u}}(x)\dd \hat{u}'(x),
\end{aligned}
\end{equation*}
where the last equality treats $\hat{u}'$ as a Lebesgue--Stieltjes measure.
Hence, \eqref{Eqn_Linear} can be reformulated as
\begin{equation*}
\overline{u}(0)\int^{1}_{0}\Phi_{\overline{u}}(z)\dd z-\int^{1}_{0} \mathcal{R}_{\overline{u}}(x)\dd \overline{u}'(x)\leq \hat{u}(0)\int^{1}_{0}\Phi_{\overline{u}}(z)\dd z-\int^{1}_{0} \mathcal{R}_{\overline{u}}(x)\dd \hat{u}'(x), \quad \forall\, \hat{u}\in \ext \overline{\mathcal{U}}.
\end{equation*}

If $\mathcal{R}_{\overline{u}}'(0)<0$, then the optimal $\overline{u}$ must have $\overline{u}(0)=1$. However, this in turn implies
\begin{equation*}
\int^{1}_{0}\Phi_{\overline{u}}(x)\dd x=\int^{1}_{0}\left[\int^{1}_{x_{1}}\phi (x_{1},x_{2})\dd x_{2}-f(x_{1},1)\right]\dd x_{1}=-\mu (\{\mathbf{x}\in \mathcal{X}:x_{1}\leq x_{2}\})=\frac{1}{2},
\end{equation*}
a contradiction.

If $\mathcal{R}_{\overline{u}}'(0)>0$, then the optimal $\overline{u}$ must have $\overline{u}(0)=0$. Denote by $x^{*}=\inf\{x\in [0,1]: \int^{1}_{x}\Phi_{\overline{u}}(z)\dd z\leq 0\}$. Then $\mathcal{R}_{\overline{u}}'(x)=\int^{1}_{x}\Phi_{\overline{u}}(z)\dd z>0$ for all $x \in [0,x^{*})$, which means $\overline{u}(x)=0$ for all $x \in [0,x^{*})$. However, this in turn implies
\begin{equation*}
\mathcal{R}_{\overline{u}}''(x)=-\Phi_{\overline{u}}(x)=f(x,1)>0,\quad \forall\, x\in (0,x^{*}),
\end{equation*}
namely, $\mathcal{R}_{\overline{u}}'$ is strictly increasing on $[0,x^{*})$, a contradiction to our definition of $x^{*}$.

Thus, there must be $\mathcal{R}_{\overline{u}}'(0)=0$, which is \eqref{Eqn_R0}. Inserting this to the reformulation of \eqref{Eqn_Linear} gives \eqref{Eqn_Supp}.

Finally, we show that \eqref{Eqn_R0} and \eqref{Eqn_Supp} jointly imply \eqref{Eqn_Linear}. Given \eqref{Eqn_R0}, it suffices to prove
\begin{equation*}
-\int^{1}_{0} \mathcal{R}_{\overline{u}}(x)\dd \overline{u}'(x)\leq -\int^{1}_{0} \mathcal{R}_{\overline{u}}(x)\dd \hat{u}'(x), \quad \forall\, \hat{u}\in \ext \overline{\mathcal{U}},
\end{equation*}
which is implied by \eqref{Eqn_Supp}.
This completes the proof.

\subsection{Proof of Corollary \ref{Cor_Pure}}

Part \ref{Cor_Pure_1} follows from \Cref{Prop_Sufficiency}. The ``if'' direction of part \ref{Cor_Pure_3} follows from part \ref{Cor_Pure_2}, and the ``only if'' direction of part \ref{Cor_Pure_3} follows from \Cref{Prop_Additive}. Thus, we will only prove part \ref{Cor_Pure_2}.

%and part \ref{Cor_Pure_3} follows from part \ref{Cor_Pure_2}, .

Since $p$ is a revenue-maximizing pure bundling mechanism, the first-order condition with respect to $p(\mathbf{1})$ gives $\MR (\mathbf{1})=2\int^{1}_{0}\Phi(x)\dd x=0$, which implies \eqref{Eqn_R0}. 
Denote by $x_{\zeta}$ the intersection of $\zeta$ and the 45-degree line. By \eqref{Eqn_Pure_Suf}, $\mathcal{R}_{\overline{u}}''(x)=-\Phi(x)\leq 0$ for $x\in [0,x_{\zeta})$ and $\geq 0$ for $(x_{\zeta},1]$. 
Since $\mathcal{R}_{\overline{u}}'(0)=0$, $\mathcal{R}_{\overline{u}}'(x)\leq 0$ for all $x$, which implies \eqref{Eqn_Supp}. Hence $p$ is optimal.

\subsection{Proof of Corollary \ref{Cor_Mixed}}

Part \ref{Cor_Mixed_1} follows from \Cref{Prop_Sufficiency}. The ``only if'' direction of part \ref{Cor_Mixed_2} follows from \Cref{Prop_Additive}. Thus, we will only prove the ``if'' direction of part \ref{Cor_Mixed_2}.

Without loss of generality, suppose that $\mathbf{e}_{2}$ is active (and thus pooling) under $p$. Then
\begin{equation*}
\max (1-\overline{u}(x),x)=
\begin{cases}
\max (p(\mathbf{e}_{2}),x) & \text{if }x\in [0,a^{1}), \\
\max (p(\mathbf{e}_{2})+a^{1}-x,x) & \text{if }x\in [a^{1},1],
\end{cases}
\end{equation*}
which means $\max (1-\overline{u}(x),x)$ is convex separately on $[0,a^{1}]$ and $[a^{1},1]$. By \eqref{Eqn_Mixed}, $\max (1-\overline{u}(a^{1}),a^{1})\leq \zeta (a^{1})$. Let
\begin{equation*}
\begin{aligned}
x_{\zeta,1} & =\inf\{x\in [0,a^{1}]:\max (1-\overline{u}(x),x)\leq \zeta (x)\}, \\
x_{\zeta,2} & =\sup\{x\in [a^{1},1]:\max (1-\overline{u}(x),x)\leq \zeta (x)\}.
\end{aligned}
\end{equation*}
Since $\zeta$ is concave, $\zeta (x)-\max (1-\overline{u}(x),x)$ is concave on each of these two intervals. Hence, $\zeta (x)-\max (1-\overline{u}(x),x)\geq 0$ on $[x_{\zeta,1},x_{\zeta,2}]$ and $<0$ otherwise. This implies $\mathcal{R}_{\overline{u}}''(x)
=-\Phi(x)\leq 0$ on $[x_{\zeta,1},x_{\zeta,2}]$ and $>0$ otherwise.
Moreover, since $p$ is a revenue-maximizing deterministic bundling mechanism, first-order conditions with respect to $p(\mathbf{e}_{1})$, $p(\mathbf{e}_{2})$, and $p(\mathbf{1})$ yield
\begin{equation*}
\MR (\mathbf{e}_{2})=\int_{0}^{a^{1}}\Phi(x)\dd x=0,\text{ and }\MR (\mathbf{1})=2\int_{a^{1}}^{1}\Phi(x)\dd x=0.
\end{equation*}
Together, they imply \eqref{Eqn_R0}.
It follows that $\mathcal{R}_{\overline{u}}'\geq 0$ on $[0,a^{1}]$ and $\leq 0$ on $[a^{1},1]$, which implies \eqref{Eqn_Supp}. Hence $p$ is optimal.

\subsection{Proof of Example \ref{Exp_RegulationUniform}}

Under uniform distribution, $\mu_{\omega}$ has negative interior density $-(3-\omega)$, positive density on the top boundary, and no mass on the bottom boundary. Consequently, the arguments for \Cref{Lem_MM} and \Cref{Cor_Inc,Cor_Pure} apply with $\phi$ replaced by $3-\omega$ and $\zeta$ replaced by $\zeta_{\omega}=(2-\omega)/(3-\omega)$. Since $\zeta_{\omega}$ is constant, \Cref{Cor_Inc} implies that an optimal symmetric mechanism is deterministic.

We first consider pure bundling. Maximizing the regulator's objective over the bundle price gives
\begin{equation*}
p(\mathbf{1})=\sqrt{\frac{2(1-\omega)}{3-\omega}}.
\end{equation*}
By \Cref{Cor_Pure}, this mechanism is optimal if and only if $p(\mathbf{1})\leq \zeta_{\omega}$. Substituting the two expressions and simplifying gives $\omega\geq 2-\sqrt{2}=\omega^{*}$. Part \ref{Exp_RegulationUniform_2} follows from $T(\mathbf{1})=2-p(\mathbf{1})$.

When $\omega<\omega^{*}$, pure bundling is not optimal, so $\mathbf{e}_{1}$ and $\mathbf{e}_{2}$ must be active. Let $a^{1}=p(\mathbf{1})-p(\mathbf{e}_{2})$. As in the proof of \Cref{Cor_Inc}, $0<a^{1}<\zeta_{\omega}$. The first-order condition for the price of $\mathbf{e}_{2}$ is
\begin{equation*}
\int_{0}^{a^{1}}
\left[(3-\omega)\left(1-\max\{p(\mathbf{e}_{2}),x\}\right)-1\right]\dd x=0.
\end{equation*}
If $p(\mathbf{e}_{2})\leq a^{1}$, the integrand would be strictly positive throughout this interval. Hence $p(\mathbf{e}_{2})>a^{1}$, and the first-order condition gives $p(\mathbf{e}_{2})=\zeta_{\omega}$. By symmetry, $p(\mathbf{e}_{1})=\zeta_{\omega}$ as well.

The grand bundle's demand set has area $(1-a^{1})^{2}-(\zeta_{\omega}-a^{1})^{2}/2$ and intersects each upper boundary in an interval of length $1-a^{1}$. Its first-order condition is therefore
\begin{equation*}
(3-\omega)
\left[(1-a^{1})^{2}-\frac{(\zeta_{\omega}-a^{1})^{2}}{2}\right]
-2(1-a^{1})=0.
\end{equation*}
Substituting $\zeta_{\omega}$ gives
$a^{1}=(2-\omega-\sqrt{2})/(3-\omega)$, and hence
$p(\mathbf{1})=[2(2-\omega)-\sqrt{2}]/(3-\omega)$.
Part \ref{Exp_RegulationUniform_1} follows from $T(\mathbf{q})=\mathbf{1}\cdot\mathbf{q}-p(\mathbf{q})$.

%Finally, because the two single-product targets have the same transfer, a firm choosing total output one assigns production to its lower-cost technology. Thus the menu can be implemented by aggregate-output targets from $\{0,1,2\}$. At $\omega=\omega^{*}$, $a^{1}=0$, so the single-product targets cease to be active and the menu reduces to targets from $\{0,2\}$.

%For the uniform distribution, $\mu_\omega$ has interior density $-(3-\omega)$ and unit density on the top and right boundaries. The associated type-dependent posted price is $(2-\omega)/(3-\omega)$. By \Cref{Cor_Inc}, every symmetric optimal mechanism is deterministic. Results in the example thus follow from direct computation.

\subsection{Proof of Example \ref{Exp_CapacityUniform}}
\label{Prf_Prop_Capacity}

\textbf{Step 1.}
We first justify the multiplier formulation. Since subtracting $u(\mathbf{0})$ increases revenue without changing allocations, we can normalize $u(\mathbf{0})=0$. The feasible normalized utilities are convex, nondecreasing in each coordinate, and 1-Lipschitz with respect to the $\ell_{1}$-norm. They form a compact convex set under the uniform norm. Revenue is a continuous linear functional on this set. Under uniform distribution, total supply is also a continuous linear functional, since
\begin{equation*}
\int_{\mathcal{X}}\mathbf{1}\cdot\nabla u(\mathbf{x})\dd\mathbf{x}=\int_{\bd (\mathcal{X})}u(\mathbf{x})(\mathbf{1}\cdot \mathbf{n}_{\mathbf{x}})\dd \sigma (\mathbf{x}).
\end{equation*}

The utility $u(\mathbf{x})=\mathbf{1}\cdot\mathbf{x}$ supplies both products to every buyer and therefore strictly satisfies the quantity constraint whenever $K<2$. The Lagrange multiplier theorem for convex programs in \citet[Chapter 8]{luenberger1997optimization} consequently applies. Hence we have the following strong duality
\begin{equation*}
%\label{Eqn_Dual}
\max_{u}\left\{\int_{\mathcal{X}}u\dd \mu : \int_{\bd (\mathcal{X})}u(\mathbf{x})(\mathbf{1}\cdot \mathbf{n}_{\mathbf{x}})\dd \sigma (\mathbf{x})\geq K\right\}=\min_{\lambda\geq 0}\left\{-\lambda K+\max_{u}\int_{\mathcal{X}}u\dd\mu_{\lambda}\right\},
\end{equation*}
where the maximization is over feasible normalized utilities. Both optima are attained, and complementary slackness requires
\begin{equation*}
\lambda \left(\int_{\bd (\mathcal{X})}u(\mathbf{x})(\mathbf{1}\cdot \mathbf{n}_{\mathbf{x}})\dd \sigma (\mathbf{x})-K\right)=0.
\end{equation*}

\bigskip
\noindent \textbf{Step 2.}
We then identify active bundles.
Fix $\lambda\geq 0$ and a symmetric optimal mechanism for the dual problem. Since $\mu_{\lambda}$ is negative in the interior and nonpositive on the lower boundaries, the argument for \Cref{Lem_MM} continues to apply. 
Define
\begin{equation*}
\Phi_{\lambda}(x)
=
2-\lambda-3\max\{1-\overline{u}(x),x\},
\quad x\in[0,1].
\end{equation*}
The arguments for the MR conditions remain valid whenever the affected demand sets avoid the lower boundaries, with $\Phi$ replaced by $\Phi_{\lambda}$.
Since $\Phi_{\lambda}(1)=-(1+\lambda)<0$, the same local SOC argument as in \Cref{Cor_GrandBundle} shows that $\mathbf{1}$ is pooling.

Suppose that the optimal mechanism is not pure bundling. Then $a^{1}>0$, and $\mathcal{D}(\mathbf{1})$ avoids the lower boundaries.
Since $\Phi_{\lambda}(x)\leq 2-\lambda-3x$, part \ref{Prop_Additive_2} of  \Cref{Prop_Additive} implies
$a^{1}<(2-\lambda)/3=\zeta_{\lambda}$, and particularly $\lambda<2$.
On $[0,a^{1}]$, the sign of $\Phi_{\lambda}$ therefore coincides with that of the nondecreasing function
$\zeta_{\lambda}+\overline{u}-1$.

Suppose that $(q,1)$ is pooling for some $q\in(0,1)$ and $a^{q}>0$. Its demand set and the local perturbations at its left endpoint avoid the lower boundaries.
By part \ref{Prop_Additive_2} of  \Cref{Prop_Additive}, $\Phi_{\lambda}\geq 0$ throughout $[a^{q},b^{q}]$, and $\overline{u}=1-\zeta_{\lambda}$ throughout this interval, contradicting $q>0$.
Moreover, a pooling interval with slope zero necessarily starts at zero, because $\overline{u}'$ is nonnegative and nondecreasing. Hence every nongrand pooling interval must start at zero, and there can be at most one such interval.

Separating bundles are impossible as well. A separating bundle cannot be isolated in the active set. Since there is at most one nongrand pooling bundle, the existence of a separating bundle would imply infinitely many separating bundles with distinct and positive endpoints. By part \ref{Prop_Additive_3} of  \Cref{Prop_Additive}, $\overline{u}=1-\zeta_{\lambda}$ at every positive endpoint, a contradiction.

Thus, every symmetric optimum is either pure bundling or has exactly three nonzero active bundles as illustrated in part \ref{Exp_CapacityUniform_2} of the example. 

\bigskip
\noindent\textbf{Step 3.}
Next, we provide a sufficient condition and identify cutoffs.
Similar to the proof of \Cref{Prop_Sufficiency}, we can rewrite the objective as a functional of $\overline{u}$:
\begin{equation*}
J_{\lambda}(\overline{u})=2(1+\lambda)\int_0^1 \overline{u}(x)\dd x
-6\int_0^1\int_x^1 \max(\overline{u}(x)+y-1,0)\dd y\dd x
-\lambda \overline{u}(0)^2.
\end{equation*}
Here \(\overline{u}\) is nonnegative, convex, nondecreasing, 1-Lipschitz, and
\(\overline{u}(0)\leq 1\). The final term is the exact contribution of the bottom boundary. This functional is concave, and its Gateaux derivative in direction $\psi$ is
\[
DJ_{\lambda}(\overline{u})[\psi]=-2\left[\int_0^1\Phi_\lambda(x) \psi(x)\dd x
+\lambda \overline{u}(0)\psi (0)\right].
\] 
Every optimum has \(0<\overline{u}(0)<1\): at \(\overline{u}(0)=0\), \(\overline{u}(x)\le x\) gives
\(\int_0^1\Phi_\lambda(x) \dd x\leq -1/4-\lambda\), so raising the intercept improves
the objective; at \(\overline{u}(0)=1\), the intercept expression is
\(\int_0^1\Phi_\lambda(x) \dd x+\lambda \overline{u}(0)=1/2>0\), so lowering it improves the
objective.
Define \[
R_{\lambda}(x)=-\int_x^1\int_y^1\Phi_\lambda(z)\dd z\dd y.
\] The analogue of Proposition 2 is the following
necessary and sufficient conditions: \[
R_{\lambda}'(0)+\lambda \overline{u}(0)=0,
\qquad \supp \overline{u}'\subseteq\arg\max R_{\lambda}, 
\] with similar endpoint conventions. A pure bundling mechanism with $0<\overline{u}<1$ has $\supp \overline{u}'=\{0\}$, so it is optimal if and only if 
\begin{align}
\label{Eqn_Capacity_R0}
R_{\lambda}'(0)+\lambda \overline{u}(0) & =\int_{0}^{1}\Phi_{\lambda}(x)\dd x +\lambda \overline{u}(0)=0, \\
R_{\lambda}(x) & \leq R_{\lambda}(0), \quad \forall\, x\in [0,1].
\end{align}
For the second inequality,
\begin{align}
\label{Eqn_Capacity_Supp}
R_{\lambda}(x)-R_{\lambda}(0) & =x\left[\frac{(3p(\mathbf{1})-2+\lambda )x}{2}-\frac{x^{2}}{2}-\lambda (1-p(\mathbf{1}))\right],\quad 0\leq x\leq \frac{p(\mathbf{1})}{2}, \\
R_{\lambda}'(x) & =(1-x)\left(\frac{1}{2}-\lambda-\frac{3}{2}x\right)<0,\quad \frac{p(\mathbf{1})}{2}<x\leq 1.
\end{align}
The bracket in \eqref{Eqn_Capacity_Supp} is maximized at $x=(3p(\mathbf{1})-2+\lambda )/2$. Hence \eqref{Eqn_Capacity_Supp} is equivalent to
\begin{equation*}
\frac{1}{8}(3p(\mathbf{1})-2+\lambda )^{2}\leq \lambda (1-p(\mathbf{1})).
\end{equation*}
Applying \eqref{Eqn_Capacity_R0} to solve the inequality gives pure bundling when $\lambda\geq \lambda^{*}\approx 0.077$. 

When the optimal mechanism has three active bundles, $\supp \overline{u}'=\{0,a^{1}\}$, so the cutoffs are given by
\begin{equation}
\begin{aligned}
R_{\lambda}'(0)+\lambda \overline{u}(0)=0 \quad & \implies \quad \int_{0}^{1}\Phi_{\lambda}(x)\dd x+\lambda \overline{u}(0)=0, \\
R_{\lambda}(0)=R_{\lambda}(a) \quad & \implies \quad \int_{0}^{a^{1}}x\Phi_{\lambda}(x)\dd x=0.
\end{aligned}
\end{equation}

\bigskip
\noindent \textbf{Step 4.}
It remains to relate $\lambda$ to the quantity floor. We first prove uniqueness of the optimal $\overline{u}$: Suppose that $\overline{u}$ and $\widehat{u}$ are both optimal. Their midpoint is optimal by concavity of $J_{\lambda}$. Since both intercepts are less than one, for all sufficiently small $x$,
\begin{equation*}
\int_{x}^{1}\max\{\overline{u}(x)+y-1,0\}\dd y
=\frac{\overline{u}(x)^{2}}{2},
\end{equation*}
and the same identity holds for $\widehat{u}$. Equality in the concavity inequality therefore forces $\overline{u}=\widehat{u}$ on a neighborhood of zero.

By the preceding menu characterization, the two profiles have the same intercept and initial slope. If that slope is one, they coincide everywhere. Otherwise, denote their common initial slope by $q<1$. If their grand-bundle cutoffs differed, their midpoint would have slope $(1+q)/2$ on the interval between these cutoffs. This would be a nongrand pooling interval with a strictly positive left endpoint, contradicting the preceding characterization. Hence the optimal symmetric profile is unique.

Symmetrizing any optimal utility preserves both its dual objective and its expected output. Thus all dual optima have the same expected output, which we denote by $K_{\lambda}$. Compactness of the feasible normalized utilities and continuity of the objective imply that every limit point of optima for multipliers converging to $\lambda$ is optimal at $\lambda$. Since expected output is continuous and unique at each multiplier, $K_{\lambda}$ is continuous.

For $\widehat{\lambda}>\lambda$, adding the two optimality inequalities gives
$(\widehat{\lambda}-\lambda)(K_{\widehat{\lambda}}-K_{\lambda})\geq 0$.
If $K_{\widehat{\lambda}}=K_{\lambda}$, either optimal utility would also be optimal at the other multiplier. Subtracting \eqref{Eqn_Capacity_R0} at the two multipliers would then give
$(\widehat{\lambda}-\lambda)[\overline{u}(0)-1]=0$,
contradicting $\overline{u}(0)<1$. Hence $K_{\lambda}$ is strictly increasing.

Moreover, comparison with the full-supply mechanism gives
\begin{equation*}
0\leq 2-K_{\lambda}
\leq \frac{\sup_{u}\TR(u)}{\lambda}
\longrightarrow 0
\quad \text{as }\lambda\to\infty,
\end{equation*}
where the supremum is over feasible normalized utilities and is finite by compactness. Consequently, every $K\in (K_{0},2)$ equals $K_{\lambda}$ for a unique $\lambda>0$. Strong duality and complementary slackness imply that the corresponding dual optimum solves the quantity-constrained problem and that the constraint binds.

Under pure bundling, $K_{\lambda}=2-p(\mathbf{1})^{2}$. In particular, $K_{\lambda^{*}}=2-(p^{*})^{2} \approx 1.411$.
Thus $K_{0}<K<K_{\lambda^{*}}$ corresponds to the three nonzero bundles with $q\in (0,1)$, whereas $K_{\lambda^{*}}\leq K<2$ corresponds to pure bundling with $p(\mathbf{1})=\sqrt{2-K}$.

\subsection{Proof of Proposition \ref{Prop_Pooling}}

Suppose that $\mathbf{x}^{*}$ and $\mathbf{q}^{*}$ satisfy the conditions of the proposition, but $\mathcal{D}(\mathbf{q}^{*})$ has zero measure with respect to $f$. Our goal is to prove that \eqref{Eqn_SOC} fails for $\mathbf{q}^{*}$. 
%Since the maximum defining $\overline{p}(\mathbf{q}^{*})$ is attained, $\mathcal{D}(\mathbf{q}^{*})\neq\varnothing$.

First, we claim that $\mathcal{D}(\mathbf{q}^{*})\subseteq \mathcal{X}_{P}=\{\mathbf{x}\in\bd(\mathcal{X}):\forall\,\tau>0,\ (\mathbf{x}+\tau\mathbf{x}^{*})\notin\mathcal{X}\}$.
Indeed, suppose that $\mathcal{D}(\mathbf{q}^{*})\setminus\mathcal{X}_{P}\neq\varnothing$. Then there exist $\mathbf{x}\in\mathcal{D}(\mathbf{q}^{*})\setminus\mathcal{X}_{P}$ and $\tau>0$ such that $(\mathbf{x}+\tau\mathbf{x}^{*})\in\mathcal{X}$. Since $\mathcal{Q}$ is a polytope, \eqref{Eqn_Pooling} implies $\tau\mathbf{x}^{*}\in\Int\NC(\mathbf{q}^{*})$. It can be verified that $\Int (\mathbf{x}+\NC(\mathbf{q}^{*}) )\cap\Int(\mathcal{X})$ is nonempty and is a subset of $\mathcal{D}(\mathbf{q}^{*})$ because of \Cref{Lem_NormalCone}. This contradicts the supposition that  $\mathcal{D}(\mathbf{q}^{*})$ has zero measure with respect to $f$.

For every $\varepsilon>0$, let $\mathcal{D}_{\varepsilon}(\mathbf{q}^{*})$ denote the demand set of $\mathbf{q}^{*}$ under the perturbed price schedule $\overline{p}-\varepsilon\mathbb{I}_{\{\mathbf{q}^{*}\}}$. By \Cref{Lem_Cont}, $\mathcal{D}_{\varepsilon}(\mathbf{q}^{*})$ contains a relative neighborhood of $\mathcal{D}(\mathbf{q}^{*})$ in $\mathcal{X}$, so it has positive measure with respect to $f$. Moreover, by \Cref{Lem_NormalCone}, if $\mathbf{x}\in\mathcal{D}_{\varepsilon}(\mathbf{q}^{*})$ and $(\mathbf{x}+\tau\mathbf{x}^{*})\in\mathcal{X}$ for some $\tau>0$, then $(\mathbf{x}+\tau\mathbf{x}^{*})\in\mathcal{D}_{\varepsilon}(\mathbf{q}^{*})$. Since $\mathcal{X}$ is compact and convex, every line parallel to $\mathbf{x}^{*}$ that intersects $\mathcal{D}_{\varepsilon}(\mathbf{q}^{*})$ has a unique terminal point in direction $\mathbf{x}^{*}$, which belongs to $\mathcal{D}_{\varepsilon}(\mathbf{q}^{*})\cap\mathcal{X}_{P}$.

Next, for every $\mathbf{x}\in\mathcal{D}_{\varepsilon}(\mathbf{q}^{*})\cap\mathcal{X}_{P}$, define $\overline{\tau}_{\varepsilon}(\mathbf{x})=\max\{\tau\geq0:\mathbf{x}-\tau\mathbf{x}^{*}\in\mathcal{D}_{\varepsilon}(\mathbf{q}^{*})\}$, where the maximum is attained because $\mathcal{D}_{\varepsilon}(\mathbf{q}^{*})$ is compact. We claim that
\begin{equation}
\label{Eqn_Pooling_Uniform_Depth}
\lim_{\varepsilon\downarrow 0}\sup_{\mathbf{x}\in\mathcal{D}_{\varepsilon}(\mathbf{q}^{*})\cap\mathcal{X}_{P}}\overline{\tau}_{\varepsilon}(\mathbf{x})=0.
\end{equation}
Otherwise, there would exist $\delta>0$, a sequence $\varepsilon_{i}\downarrow 0$, and $\mathbf{x}_{i}\in\mathcal{D}_{\varepsilon_{i}}(\mathbf{q}^{*})\cap\mathcal{X}_{P}$ such that $\overline{\tau}_{\varepsilon_{i}}(\mathbf{x}_{i})>\delta$ for all $i$. Upward closure of the demand set then implies $(\mathbf{x}_{i}-\delta\mathbf{x}^{*})\in\mathcal{D}_{\varepsilon_{i}}(\mathbf{q}^{*})$ for all $i$. Since $\mathcal{X}$ is compact, $\{\mathbf{x}_{i}\}$ admits a convergent subsequence, and we denote its limit by $\mathbf{y}\in\mathcal{X}$. By \Cref{Lem_Cont} and the closedness of $\mathcal{D}(\mathbf{q}^{*})$, we have $(\mathbf{y}-\delta\mathbf{x}^{*})\in\mathcal{D}(\mathbf{q}^{*})$. However, $(\mathbf{y}-\delta\mathbf{x}^{*})+\delta\mathbf{x}^{*}=\mathbf{y}\in\mathcal{X}$, so $\mathbf{y}-\delta\mathbf{x}^{*}\notin\mathcal{X}_{P}$, contradicting $\mathcal{D}(\mathbf{q}^{*})\subseteq\mathcal{X}_{P}$.

Recall that the outward unit normal $\mathbf{n}_{\mathbf{x}}$ exists for $\sigma$-almost every $\mathbf{x}\in\bd(\mathcal{X})$. We claim that $\mathbf{x}^{*}\cdot\mathbf{n}_{\mathbf{x}}\geq0$ at every such $\mathbf{x}\in\mathcal{X}_{P}$. Otherwise, there is $\mathbf{x}^{*}\cdot\mathbf{n}_{\mathbf{x}}<0$, and thus $\mathbf{x}+\tau\mathbf{x}^{*}\in\mathcal{X}$ for every sufficiently small $\tau>0$, contradicting $\mathbf{x}\in\mathcal{X}_{P}$.
Consequently, the tangential Jacobian of $(\mathbf{x},\tau)\mapsto\mathbf{x}-\tau\mathbf{x}^{*}$ is $|\mathbf{x}^{*}\cdot\mathbf{n}_{\mathbf{x}}|=\mathbf{x}^{*}\cdot\mathbf{n}_{\mathbf{x}}$ for $\sigma$-almost every $\mathbf{x}\in\mathcal{X}_{P}$. Since every line parallel to $\mathbf{x}^{*}$ has a unique terminal point in $\mathcal{X}_{P}$, the area formula applied on the smooth pieces of $\bd(\mathcal{X})$ gives, for every bounded Borel function $\psi$,
\begin{equation}
\label{Eqn_Pooling_General_COV}
\int_{\mathcal{D}_{\varepsilon}(\mathbf{q}^{*})}\psi(\mathbf{y})\dd\mathbf{y}
=\int_{\mathcal{D}_{\varepsilon}(\mathbf{q}^{*})\cap\mathcal{X}_{P}}\int_{0}^{\overline{\tau}_{\varepsilon}(\mathbf{x})}\psi(\mathbf{x}-\tau\mathbf{x}^{*})(\mathbf{x}^{*}\cdot\mathbf{n}_{\mathbf{x}})\dd\tau\dd\sigma(\mathbf{x}).
\end{equation}
Taking $\psi=1$ in \eqref{Eqn_Pooling_General_COV} gives
\begin{equation*}
\int_{\mathcal{D}_{\varepsilon}(\mathbf{q}^{*})}\dd\mathbf{y}
=\int_{\mathcal{D}_{\varepsilon}(\mathbf{q}^{*})\cap\mathcal{X}_{P}}\overline{\tau}_{\varepsilon}(\mathbf{x})(\mathbf{x}^{*}\cdot\mathbf{n}_{\mathbf{x}})\dd\sigma(\mathbf{x}).
\end{equation*}
Since the left-hand side is strictly positive, it follows that
\begin{equation}
\label{Eqn_Pooling_Positive_Base}
\int_{\mathcal{D}_{\varepsilon}(\mathbf{q}^{*})\cap\mathcal{X}_{P}}(\mathbf{x}^{*}\cdot\mathbf{n}_{\mathbf{x}})\dd\sigma(\mathbf{x})>0.
\end{equation}

Finally, by \Cref{Lem_MR}, $\MR_{-}(\mathbf{q}^{*};\overline{p}-\varepsilon\mathbb{I}_{\{\mathbf{q}^{*}\}})
=-\mu(\mathcal{D}_{\varepsilon}(\mathbf{q}^{*}))$. Taking $\psi=\phi$ in \eqref{Eqn_Pooling_General_COV} gives
\begin{equation*}
\begin{aligned}
-\mu(\mathcal{D}_{\varepsilon}(\mathbf{q}^{*}))
& \leq\int_{\mathcal{D}_{\varepsilon}(\mathbf{q}^{*})\cap\mathcal{X}_{P}}\left[\int_{0}^{\overline{\tau}_{\varepsilon}(\mathbf{x})}\phi(\mathbf{x}-\tau\mathbf{x}^{*})(\mathbf{x}^{*}\cdot\mathbf{n}_{\mathbf{x}})\dd\tau-(\mathbf{x}\cdot\mathbf{n}_{\mathbf{x}})f(\mathbf{x})\right]\dd\sigma(\mathbf{x})\\
&\leq\int_{\mathcal{D}_{\varepsilon}(\mathbf{q}^{*})\cap\mathcal{X}_{P}}\left[\max_{\mathbf{y}\in\mathcal{X}}|\phi(\mathbf{y})|\,\overline{\tau}_{\varepsilon}(\mathbf{x})-\inf_{\mathbf{y}\in\overline{\bd}(\mathcal{X})}f(\mathbf{y})\right](\mathbf{x}^{*}\cdot\mathbf{n}_{\mathbf{x}})\dd\sigma(\mathbf{x}),
\end{aligned}
\end{equation*}
where the first inequality follows from $\mu(\underline{\bd}(\mathcal{X}))=0$, and the second inequality follows from $(\mathbf{x}-\mathbf{x}^{*})\cdot\mathbf{n}_{\mathbf{x}}\geq 0$.
By \eqref{Eqn_Pooling_Uniform_Depth}, for every sufficiently small $\varepsilon>0$,
\begin{equation*}
\max_{\mathbf{y}\in\mathcal{X}}|\phi(\mathbf{y})|\times\sup_{\mathbf{x}\in\mathcal{D}_{\varepsilon}(\mathbf{q}^{*})\cap\mathcal{X}_{P}}\overline{\tau}_{\varepsilon}(\mathbf{x})<\frac{1}{2}\inf_{\mathbf{y}\in\overline{\bd}(\mathcal{X})}f(\mathbf{y}).
\end{equation*}
Together with \eqref{Eqn_Pooling_Positive_Base}, this implies
\begin{equation*}
\MR_{-}(\mathbf{q}^{*};\overline{p}-\varepsilon\mathbb{I}_{\{\mathbf{q}^{*}\}})<0
\end{equation*}
for every sufficiently small $\varepsilon>0$, a contradiction to \eqref{Eqn_SOC}.

%Therefore, $\mathcal{D}(\mathbf{q}^{*})$ has positive measure with respect to $f$, so $\mathbf{q}^{*}$ is pooling.

\subsection{Proof of Proposition \ref{Prop_Active}}

Let $\{\mathcal{S}(\mathbf{n}_{j})\}_{j=1}^{J}$ be the top facets of $\mathcal{X}$, where each $\mathbf{n}_{j}$ is the corresponding outward unit normal. Denote by $\mathcal{Q}_{A}$ the set of bundles satisfying \eqref{Eqn_Active}. We first show that $\mathcal{Q}_{A}$ is closed.

Write an irredundant half-space representation of $\mathcal{Q}$ as 
\begin{equation*}
\mathcal{Q}=\{\mathbf{q}\in\mathbb{R}^{N}:\mathbf{m}_{\ell}\cdot\mathbf{q}\leq d_{\ell},\ \ell=1,\dotsc,L\}.
\end{equation*}
Then we claim that
\begin{equation}
\label{Eqn_ActiveFacets}
\mathcal{Q}_{A}=\bigcup_{\substack{1\leq j\leq J,\ 1\leq\ell\leq L\\ \mathbf{m}_{\ell}\cdot\mathbf{n}_{j}>0}}\{\mathbf{q}\in\mathcal{Q}:\mathbf{m}_{\ell}\cdot\mathbf{q}=d_{\ell}\}.
\end{equation}
Indeed, if $\mathbf{m}_{\ell}\cdot\mathbf{q}=d_{\ell}$ and $\mathbf{m}_{\ell}\cdot\mathbf{n}_{j}>0$, then $\mathbf{q}+\tau\mathbf{n}_{j}$ violates the $\ell$-th constraint for every $\tau>0$. Conversely, if no such constraint binds at $\mathbf{q}$ for a given $j$, every constraint that becomes tighter in direction $\mathbf{n}_{j}$ is strictly slack. Hence $(\mathbf{q}+\tau\mathbf{n}_{j})\in\mathcal{Q}$ for sufficiently small $\tau>0$. This proves \eqref{Eqn_ActiveFacets}, so $\mathcal{Q}_{A}$ is closed.

Suppose that there exists an active bundle $\mathbf{q}^{*}\notin\mathcal{Q}_{A}$. Since $\mathbf{q}^{*}\neq\mathbf{0}$ and $\mathcal{Q}_{A}$ is closed, there exists $\delta>0$ such that $\overline{B}_{\delta}(\mathbf{q}^{*})\cap\mathcal{Q}_{A}=\varnothing$ and $\mathbf{0}\notin\overline{B}_{\delta}(\mathbf{q}^{*})$. By the activity of $\mathbf{q}^{*}$,
\begin{equation*}
\int_{\mathcal{D}(\overline{B}_{\delta}(\mathbf{q}^{*}))}f(\mathbf{x})\dd\mathbf{x}>0.
\end{equation*}
Note that every interior type in $\mathcal{D}(\overline{B}_{\delta}(\mathbf{q}^{*}))\setminus\mathcal{D}^{*}(\overline{B}_{\delta}(\mathbf{q}^{*}))$ demands both a bundle in $\overline{B}_{\delta}(\mathbf{q}^{*})$ and a distinct bundle outside it, and hence is a nondifferentiability point of $u$. Since a convex function is differentiable almost everywhere, $\mathcal{D}(\overline{B}_{\delta}(\mathbf{q}^{*}))\setminus\mathcal{D}^{*}(\overline{B}_{\delta}(\mathbf{q}^{*}))$ has zero Lebesgue measure. Strict regularity therefore implies
\begin{equation}
\label{Eqn_ActivePositivePhi}
\int_{\mathcal{D}^{*}(\overline{B}_{\delta}(\mathbf{q}^{*}))}\phi(\mathbf{x})\dd\mathbf{x}>0.
\end{equation}
Since $\overline{B}_{\delta}(\mathbf{q}^{*})$ is closed and $\mathbf{0}\notin\overline{B}_{\delta}(\mathbf{q}^{*})$, \Cref{Lem_MR} and \eqref{Eqn_FOC+} imply $\mu(\mathcal{D}^{*}(\overline{B}_{\delta}(\mathbf{q}^{*})))\geq0$. Together with \eqref{Eqn_ActivePositivePhi}, this gives
\begin{equation*}
\int_{\mathcal{D}^{*}(\overline{B}_{\delta}(\mathbf{q}^{*}))\cap\bd(\mathcal{X})}(\mathbf{x}\cdot\mathbf{n}_{\mathbf{x}})f(\mathbf{x})\dd\sigma(\mathbf{x})
\geq\int_{\mathcal{D}^{*}(\overline{B}_{\delta}(\mathbf{q}^{*}))}\phi(\mathbf{x})\dd\mathbf{x}>0.
\end{equation*}
Since the relative interiors of the facets of $\mathcal{X}$ cover its boundary up to a set of zero surface measure and every non-top facet makes a nonpositive contribution to the preceding integral, there exists a top facet $\mathcal{S}(\mathbf{n}_{j})$ such that $\mu(\mathcal{D}^{*}(\overline{B}_{\delta}(\mathbf{q}^{*}))\cap\mathcal{S}(\mathbf{n}_{j}))>0$.

Since $\overline{B}_{\delta}(\mathbf{q}^{*})\cap\mathcal{Q}_{A}=\varnothing$, \eqref{Eqn_ActiveFacets} implies that $\mathbf{m}_{\ell}\cdot\mathbf{q}<d_{\ell}$ for every $\mathbf{q}\in\overline{B}_{\delta}(\mathbf{q}^{*})$ and every $\ell$ satisfying $\mathbf{m}_{\ell}\cdot\mathbf{n}_{j}>0$. The set of such constraints is nonempty because $\mathcal{Q}$ is compact. By compactness of $\overline{B}_{\delta}(\mathbf{q}^{*})$, we can choose $\tau>0$ such that
\begin{equation*}
\tau<\min_{\substack{\mathbf{q}\in\overline{B}_{\delta}(\mathbf{q}^{*})\\ \ell:\,\mathbf{m}_{\ell}\cdot\mathbf{n}_{j}>0}}\frac{d_{\ell}-\mathbf{m}_{\ell}\cdot\mathbf{q}}{\mathbf{m}_{\ell}\cdot\mathbf{n}_{j}}.
\end{equation*}
If $\mathbf{m}_{\ell}\cdot\mathbf{n}_{j}>0$, this choice implies $(\mathbf{q}+\tau\mathbf{n}_{j})\cdot\mathbf{m}_{\ell}<d_{\ell}$; if $\mathbf{m}_{\ell}\cdot\mathbf{n}_{j}\leq0$, then $(\mathbf{q}+\tau\mathbf{n}_{j})\cdot\mathbf{m}_{\ell}\leq\mathbf{q}\cdot\mathbf{m}_{\ell}\leq d_{\ell}$. Thus, $(\mathbf{q}+\tau\mathbf{n}_{j})\in\mathcal{Q}$ for every $\mathbf{q}\in\overline{B}_{\delta}(\mathbf{q}^{*})$.

We next claim that
\begin{equation}
\label{Eqn_ActiveIndividualDemand}
\mathcal{D}(\mathbf{q}+\tau\mathbf{n}_{j})=\mathcal{D}(\mathbf{q})\cap\mathcal{S}(\mathbf{n}_{j})
\end{equation}
for every $\mathbf{q}\in\overline{B}_{\delta}(\mathbf{q}^{*})$ satisfying $\mathcal{D}(\mathbf{q})\cap\mathcal{S}(\mathbf{n}_{j})\neq\varnothing$. Fix such a bundle $\mathbf{q}$ and some $\mathbf{x}\in\mathcal{D}(\mathbf{q})\cap\mathcal{S}(\mathbf{n}_{j})$. Since every bundle has a nonempty demand set under $\overline{p}$, we can also fix $\mathbf{y}\in\mathcal{D}(\mathbf{q}+\tau\mathbf{n}_{j})$. The demand inequalities for $\mathbf{x}$ and $\mathbf{y}$ imply
\begin{equation}
\label{Eqn_ActivePriceDifference}
\tau s(\mathbf{n}_{j})=\tau\mathbf{x}\cdot\mathbf{n}_{j}
\leq\overline{p}(\mathbf{q}+\tau\mathbf{n}_{j})-\overline{p}(\mathbf{q})
\leq\tau\mathbf{y}\cdot\mathbf{n}_{j}
\leq\tau s(\mathbf{n}_{j}).
\end{equation}
Thus, all inequalities in \eqref{Eqn_ActivePriceDifference} are equalities. In particular, $\mathbf{y}\in\mathcal{S}(\mathbf{n}_{j})$, and $\mathbf{y}$ also demands $\mathbf{q}$, which implies $\mathcal{D}(\mathbf{q}+\tau\mathbf{n}_{j})\subseteq\mathcal{D}(\mathbf{q})\cap\mathcal{S}(\mathbf{n}_{j})$. Conversely, every type in $\mathcal{D}(\mathbf{q})\cap\mathcal{S}(\mathbf{n}_{j})$ is indifferent between $\mathbf{q}$ and $\mathbf{q}+\tau\mathbf{n}_{j}$, which implies the reverse inclusion and hence \eqref{Eqn_ActiveIndividualDemand}.

Finally, define
\begin{equation*}
Q=\{\mathbf{q}+\tau\mathbf{n}_{j}:\mathbf{q}\in\overline{B}_{\delta}(\mathbf{q}^{*})\text{ and }\mathcal{D}(\mathbf{q})\cap\mathcal{S}(\mathbf{n}_{j})\neq\varnothing\}.
\end{equation*}
The set $Q$ is nonempty by the choice of $\mathcal{S}(\mathbf{n}_{j})$. It is also compact because the demand correspondence has a closed graph. Moreover, $\mathbf{0}\notin Q$. Otherwise, $\mathbf{q}+\tau\mathbf{n}_{j}=\mathbf{0}$ for some $\mathbf{q}\in\mathcal{Q}\subseteq\mathbb{R}_{+}^{N}$, which would imply $\mathbf{n}_{j}\leq\mathbf{0}$. Since $\mathcal{X}\subseteq\mathbb{R}_{+}^{N}$, this would imply $s(\mathbf{n}_{j})\leq0$, contradicting $\mathcal{S}(\mathbf{n}_{j})\subseteq\overline{\bd}(\mathcal{X})$.
Taking the union of \eqref{Eqn_ActiveIndividualDemand} over all $\mathbf{q}\in\overline{B}_{\delta}(\mathbf{q}^{*})$ satisfying $\mathcal{D}(\mathbf{q})\cap\mathcal{S}(\mathbf{n}_{j})\neq\varnothing$ gives
\begin{equation}
\label{Eqn_ActiveAggregateDemand}
\mathcal{D}(Q)=\mathcal{D}(\overline{B}_{\delta}(\mathbf{q}^{*}))\cap\mathcal{S}(\mathbf{n}_{j}).
\end{equation}
By \eqref{Eqn_ActiveAggregateDemand}, $\mathcal{D}(Q)$ is contained in the top facet $\mathcal{S}(\mathbf{n}_{j})$, so it has zero $N$-dimensional Lebesgue measure. Intersections with the other facets have zero surface measure, so $\mu$ restricts to a nonnegative measure on $\mathcal{S}(\mathbf{n}_{j})$. Since $\mathcal{D}^{*}(\overline{B}_{\delta}(\mathbf{q}^{*}))\cap\mathcal{S}(\mathbf{n}_{j})\subseteq\mathcal{D}(Q)$, the choice of $\mathcal{S}(\mathbf{n}_{j})$ implies $\mu(\mathcal{D}(Q))>0$. Therefore,
\begin{equation*}
\MR_{-}(Q)=-\mu(\mathcal{D}(Q))<0,
\end{equation*}
contradicting \eqref{Eqn_FOC-}. Hence $\mathbf{q}^{*}\in\mathcal{Q}_{A}$, proving the proposition.

\subsection{Proof of Proposition \ref{Prop_Zero}}

Since $\mathcal{Q}\setminus\{\mathbf{0}\}$ is open in $\mathcal{Q}$, the price schedule
\begin{equation*}
\overline{p}+\varepsilon\mathbb{I}_{\mathcal{Q}\setminus\{\mathbf{0}\}}
\end{equation*}
is feasible for every $\varepsilon>0$. Applying the same argument as in the proof of \Cref{Lem_MR} directly to this perturbation gives
\begin{equation*}
\lim_{\varepsilon\downarrow0}
\frac{\TR(\overline{p}+\varepsilon\mathbb{I}_{\mathcal{Q}\setminus\{\mathbf{0}\}})-\TR(\overline{p})}{\varepsilon}
=-\mu(\mathcal{X}\setminus\mathcal{D}(\mathbf{0}))
=1+\mu(\mathcal{D}(\mathbf{0})).
\end{equation*}
Optimality implies that the left-hand side is nonpositive, which proves \eqref{Eqn_ZeroMass}.

Suppose that $\mathcal{D}(\mathbf{0})$ has zero measure with respect to $f$. Since $-\mathbf{e}_{n}\in \NC(\mathbf{0})$ for every $n\in \mathcal{N}$, the set $\mathcal{D}(\mathbf{0})+\NC(\mathbf{0})$ has nonempty interior. By \Cref{Lem_NormalCone}, $(\mathcal{D}(\mathbf{0})+\NC(\mathbf{0})) \cap\mathcal{X} \subseteq\mathcal{D}(\mathbf{0})$. Therefore, the interiors of $\mathcal{D}(\mathbf{0})+\NC(\mathbf{0})$ and $\mathcal{X}$ are disjoint. Otherwise, their intersection would be a nonempty open subset of $\mathcal{D}(\mathbf{0})$ and would have positive measure
with respect to $f$.
Both $\mathcal{D}(\mathbf{0})+\NC(\mathbf{0})$ and $\mathcal{X}$ are closed and convex, and $\mathcal{D}(\mathbf{0})$ is contained in both. Hence, the separating-hyperplane theorem gives a nonzero vector $\mathbf{n}\in\mathbb{R}^{N}$ such that
\begin{equation*}
%\label{Eqn_ZeroSeparation}
\mathbf{n}\cdot(\mathbf{x}-\mathbf{y})\leq 0
\leq\mathbf{n}\cdot\mathbf{n}_{0},
\quad
\forall\,\mathbf{x}\in\mathcal{X},\
\mathbf{y}\in\mathcal{D}(\mathbf{0}),\
\mathbf{n}_{0}\in\NC(\mathbf{0}).
\end{equation*}
The first inequality implies
\begin{equation}
\label{Eqn_ZeroFace}
\mathcal{D}(\mathbf{0})\subseteq\mathcal{S}(\mathbf{n}).
\end{equation}
The second inequality, together with $-\mathbf{e}_{n}\in\NC(\mathbf{0})$ for every $n\in\mathcal{N}$, implies $\mathbf{n}\leq\mathbf{0}$. Since $\mathcal{X}\subseteq\mathbb{R}_{+}^{N}$, we have $s(\mathbf{n})\leq 0$.

First, suppose that $N\geq 2$ and $\mathcal{X}$ is strictly convex. Then $\mathcal{S}(\mathbf{n})$ is a singleton and has zero $\mu$-measure. By \eqref{Eqn_ZeroFace}, $\mu(\mathcal{D}(\mathbf{0}))=0$, contradicting \eqref{Eqn_ZeroMass}.

Next, suppose that $\mu(\underline{\bd}(\mathcal{X}))>-1$. If $s(\mathbf{n})=0$, then $\mu (\mathcal{S}(\mathbf{n}))=0$. By \eqref{Eqn_ZeroFace}, $\mu(\mathcal{D}(\mathbf{0}))=0$, again contradicting \eqref{Eqn_ZeroMass}.
If $s(\mathbf{n})<0$, then $\mathcal{S}(\mathbf{n})\subseteq\underline{\bd}(\mathcal{X})$. By \eqref{Eqn_ZeroFace}, $\mu(\mathcal{D}(\mathbf{0})) \geq \mu(\mathcal{S}(\mathbf{n})) \geq \mu(\underline{\bd}(\mathcal{X})) >-1$, which also contradicts \eqref{Eqn_ZeroMass}.

%Hence, under either condition, $\mathcal{D}(\mathbf{0})$ has positive measure with respect to $f$.

\subsection{Proof of Proposition \ref{Prop_Pure_N}}

We first note that $p(\mathbf{1})\in (N\underline{x},N\overline{x})$. Indeed, if $p(\mathbf{1})<N\underline{x}$, the seller can raise the price to $N\underline{x}$ without losing any buyer. If $p(\mathbf{1})=N\underline{x}$, raising the price by $\varepsilon>0$ excludes a set whose probability is $O(\varepsilon^{N})$ and is therefore profitable for all sufficiently small $\varepsilon$, since $N\geq 2$. If $p(\mathbf{1})\geq N\overline{x}$, the seller receives zero revenue.

Under pure bundling, $\mathcal{D}(\mathbf0)\cup\mathcal{D}(\mathbf1)=\mathcal{X}$ and their intersection is $\mu$-negligible. \Cref{Prop_Zero} gives $\mu(\mathcal{D}(\mathbf0))\leq -1$, while applying $\eqref{Eqn_FOC-}$ to $\mathbf{1}$ gives $\mu(\mathcal{D}(\mathbf1))\leq 0$. Condition \eqref{Eqn_Pure_N_1} then follows from $\mu(\mathcal{X})=-1$.

Next, we assume that $f$ is positive almost everywhere on $\overline{\bd}(\mathcal{X})$. Suppose that \eqref{Eqn_Pure_N_2} is false. Direct computation gives
\[
\overline{p}(\mathbf{e}_{n})=\overline{x},\quad
\mathcal{D}(\mathbf{e}_{n})
=\{\mathbf{x}\in\mathcal{X}:x_{n}=\overline{x}\text{ and }\mathbf{x}\cdot\mathbf{1}\leq p(\mathbf{1})\},\quad \forall\, n\in \mathcal{N}.
\]
Since $p(\mathbf{1})-(N-1)\underline{x}>\overline{x}$, this demand set contains a relatively open subset of the top facet $\{\mathbf{x}\in \mathcal{X}: x_{n}=\overline{x}\}$. Hence,
\[
\mu(\mathcal{D}(\mathbf{e}_{n}))
=\overline{x}\int_{\mathcal{D}(\mathbf{e}_{n})}f(\mathbf{x})\dd\sigma(\mathbf{x})>0,
\]
where the inequality follows from our assumption on $f$. This contradicts \eqref{Eqn_FOC-}.

Finally, we additionally assume that $\mu(\underline{\bd}(\mathcal{X}))=0$. By \eqref{Eqn_Pure_N_2}, $p(\mathbf{1})-(N-1)\underline{x} \in(\underline{x},\overline{x}]$. Direct computation gives
\[
\overline{p}(\mathbf{e}_{n})=p(\mathbf{1})-(N-1)\underline{x},\quad \mathcal{D}(\mathbf{e}_{n})
=\{(x_{n},\underline{\mathbf{x}}_{-n})\in\mathcal{X}:x_{n}\geq p(\mathbf{1})-(N-1)\underline{x}\},\quad \forall\, n\in \mathcal{N}.
\]
For $\varepsilon>0$, let $\mathcal{D}_{\varepsilon}(\mathbf{e}_{n})$ denote the demand set of $\mathbf{e}_{n}$ under the perturbed price schedule $\overline{p}-\varepsilon\mathbb{I}_{\{\mathbf{e}_{n}\}}$. For all sufficiently small $\varepsilon$,
\[
\mathcal{D}_{\varepsilon}(\mathbf{e}_{n})
=\left\{\mathbf{x}\in\mathcal{X}:x_{n}\geq p(\mathbf{1})-(N-1)\underline{x}-\varepsilon
\text{ and }\sum_{m\neq n}(x_{m}-\underline{x})\leq\varepsilon\right\}.
\]
Since $\mu(\underline{\bd}(\mathcal{X}))=0$, changing variables $x_{m}=\underline{x}+\varepsilon y_{m}$ for every $m\neq n$ and applying dominated convergence yield
\[
\lim_{\varepsilon\downarrow 0}\frac{(N-1)!}{\varepsilon^{N-1}}
\mu(\mathcal{D}_{\varepsilon}(\mathbf{e}_{n}))
=\overline{x}f(\overline{x},\underline{\mathbf{x}}_{-n})
-\int_{p(\mathbf{1})-(N-1)\underline{x}}^{\overline{x}}\phi(x_{n},\underline{\mathbf{x}}_{-n})\dd x_{n}.
\]
If \eqref{Eqn_Pure_N_3} were violated, the right-hand side would be strictly positive. Consequently,
\[
\MR_{-}\left(\mathbf{e}_{n};\overline{p}-\varepsilon\mathbb{I}_{\{\mathbf{e}_{n}\}}\right)
=-\mu(\mathcal{D}_{\varepsilon}(\mathbf{e}_{n}))<0
\]
for every sufficiently small $\varepsilon>0$, contradicting \eqref{Eqn_SOC}. This proves \eqref{Eqn_Pure_N_3}.

\subsection{Proof of Corollary \ref{Cor_Pure_N}}

Suppose that a pure bundling mechanism $p$ is optimal. It can be verified that $\mathcal{D}(\mathbf{0})$ and $\mathcal{D}(\mathbf{1})$ cover $\mathcal{X}$, and their intersection is $\mu$-negligible. Therefore, \eqref{Eqn_Pure_N_1} and $\mu(\mathcal{X})=-1$ imply $\mu(\mathcal{D}(\mathbf{0}))=-1$.

As shown in the proof of \Cref{Prop_Pure_N}, $p(\mathbf{1})>N\underline{x}$. Moreover, since $g(\overline{x})>0$, the joint density is positive almost everywhere on $\overline{\bd}(\mathcal{X})$, so \eqref{Eqn_Pure_N_2} implies $0<p(\mathbf{1})-N\underline{x}\leq \overline{x}-\underline{x}$.
Consequently, $\mathcal{D}(\mathbf{0})$ is a simplex whose volume is at most $(\overline{x}-\underline{x})^{N}/N!$ and whose intersection with $\bd (\mathcal{X})$ has surface measure at most $N(\overline{x}-\underline{x})^{N-1}/(N-1)!$.

Let $\overline{g}=\max_{x\in[\underline{x},\overline{x}]}g(x)$. By the definition of $\mu$,
\[
-\mu(\mathcal{D}(\mathbf{0}))\leq\int_{\mathcal{D}(\mathbf{0})}|\phi(\mathbf{x})|\dd\mathbf{x}
+\underline{x}\int_{\mathcal{D}(\mathbf{0})\cap\bd(\mathcal{X})}f(\mathbf{x})\dd\sigma(\mathbf{x}).
\]
Substituting $f(\mathbf{x})=\prod_{n\in\mathcal{N}}g(x_{n})$ into the definition of $\phi$ and applying the simplex bounds give
\begin{equation*}
\begin{aligned}
1=-\mu(\mathcal{D}(\mathbf{0})) & 
\leq
\frac{N\overline{g}^{\,N-1}(\overline{x}-\underline{x})^{N-1}}{(N-1)!}
\int_{\underline{x}}^{\overline{x}}|xg'(x)|\dd x \\
& \qquad +\frac{(N+1)\overline{g}^{\,N}(\overline{x}-\underline{x})^{N}}{N!}
+\frac{N\underline{x}\,\overline{g}^{\,N}(\overline{x}-\underline{x})^{N-1}}{(N-1)!}.
\end{aligned}
\end{equation*}
The three terms respectively bound the derivative terms in $\phi$, the remaining density term in $\phi$, and the contribution from the $N$ bottom facets of $\mathcal{X}$. Since $g$ is fixed, all quantities other than $N$ are constants. Hence, there exists $C\geq1$, independent of $N$, such that
\[
1\leq\frac{N^{2}C^{N}}{(N-1)!}.
\]
The right-hand side converges to zero as $N\to+\infty$, a contradiction.

\subsection{Proof of Proposition \ref{Prop_Separate_N}}

Denote by $p^{*}$ the common price for all $\mathbf{e}_{n}$, $n\in \mathcal{N}$. Similar to the proof of \Cref{Prop_Pure_N}, we can prove that $p^{*}\in (\underline{x},\overline{x})$. Direct computation gives $\overline{p}(\mathbf{q})=p^{*}\sum_{n\in \mathcal{N}}q_{n}$ for every $\mathbf{q}\in\mathcal{Q}$, and
\[
\mathcal{D}(\mathbf{e}_{n})=\{\mathbf{x}\in\mathcal{X}:x_{n}\geq p^{*}\text{ and }x_{n}\geq x_{m}\text{ for every }m\in\mathcal{N}\}.
\]
A type that demands both $\mathbf{e}_{n}$ and another bundle must satisfy either $x_{n}=p^{*}$ or $x_{n}=x_{m}$ for some $m\neq n$. Since $p^{*}\in(\underline{x},\overline{x})$ and $N\geq 2$, these sets have zero Lebesgue measure in the interior and zero surface measure on $\bd(\mathcal{X})$. Hence, $\mu (\mathcal{D}(\mathbf{e}_{n})\setminus\mathcal{D}^{*}(\mathbf{e}_{n}) )=0$, so \eqref{Eqn_Separate_N_1} follows from \eqref{Eqn_FOC}.

Next, let $\mathbf{q}=(1/N)\mathbf{1}$. We have $\overline{p}(\mathbf{q})=p^{*}$ and $\mathcal{D}(\mathbf{q})=\{(x,\dotsc,x)\in\mathcal{X}:x\geq p^{*}\}$. For $\varepsilon>0$, let $\mathcal{D}_{\varepsilon}(\mathbf{q})$ denote the demand set of $\mathbf{q}$ under the perturbed price schedule $\overline{p}-\varepsilon\mathbb{I}_{\{\mathbf{q}\}}$. Direct computation gives
\[
\mathcal{D}_{\varepsilon}(\mathbf{q})=\left\{\mathbf{x}\in\mathcal{X}:\frac{\mathbf{x}\cdot\mathbf{1}}{N}\geq p^{*}-\varepsilon\text{ and }\max_{n\in\mathcal{N}}\{x_{n}\}-\frac{\mathbf{x}\cdot\mathbf{1}}{N}\leq\varepsilon\right\}.
\]
Because $p^{*}>\underline{x}$, this demand set does not intersect the bottom boundary for all sufficiently small $\varepsilon$. Partitioning it according to its largest coordinate, rescaling the deviations from the diagonal, and applying dominated convergence yield
\[
\lim_{\varepsilon\downarrow 0}\frac{(N-1)!}{N^{N}\varepsilon^{N-1}}\mu(\mathcal{D}_{\varepsilon}(\mathbf{q}))
=\overline{x}f(\overline{\mathbf{x}})-\int_{p^{*}}^{\overline{x}}\phi(x,\dotsc,x)\dd x.
\]
If \eqref{Eqn_Separate_N_2} were violated, the right-hand side would be strictly positive. Consequently,
\[
\MR_{-}\left(\mathbf{q};\overline{p}-\varepsilon\mathbb{I}_{\{\mathbf{q}\}}\right)
=-\mu(\mathcal{D}_{\varepsilon}(\mathbf{q}))<0
\]
for every sufficiently small $\varepsilon>0$, contradicting \eqref{Eqn_SOC}. This proves \eqref{Eqn_Separate_N_2}.

\subsection{Proof of Corollary \ref{Cor_Separate_N}}

Suppose that a symmetric separate selling mechanism $p$ is optimal. Denote by $p^{*}$ the common price for all $\mathbf{e}_{n}$, $n\in \mathcal{N}$. By the proof of \Cref{Prop_Separate_N}, $p^{*}\in (\underline{x},\overline{x})$.

Along the main diagonal, the chain rule gives
\begin{equation*}
\phi(x,\dotsc,x)=\frac{\dd}{\dd x}[xf(x,\dotsc,x)]+Nf(x,\dotsc,x).
\end{equation*}
Consequently,
\begin{equation*}
\begin{aligned}
\int_{p^{*}}^{\overline{x}}\phi(x,\dotsc,x)\dd x-\overline{x}f(\overline{\mathbf{x}})
& =N\int_{p^{*}}^{\overline{x}}f(x,\dotsc,x)\dd x-p^{*}f(p^{*},\dotsc,p^{*}) \\
& \leq \left[N(\overline{x}-p^{*})-p^{*}\right]f(p^{*},\dotsc,p^{*}) \\
& <\left[N(\overline{x}-\underline{x})-\underline{x}\right]f(p^{*},\dotsc,p^{*}) \\
& \leq 0.
\end{aligned}
\end{equation*}
The first inequality follows because $f(x,\dotsc,x)$ is nonincreasing, the second inequality follows from $p^{*}>\underline{x}$ and that $f$ is strictly positive in $\Int (\mathcal{X})$, while the last inequality follows from $\underline{x}\geq N(\overline{x}-\underline{x})$. This contradicts \eqref{Eqn_Separate_N_2}.

\section{Omitted Results}

\subsection{Direct Mechanisms}
\label{Sec_DirectMech}

In this section, we reformulate the seller's problem in \Cref{Sec_Model} as finding the optimal direct mechanism. In particular, a direct mechanism is denoted by $(\mathbf{q},r)$, where $\mathbf{q} :\mathcal{X}\to \mathcal{Q}$ maps a buyer type to an allocation, and $r :\mathcal{X}\to \mathbb{R}$ maps a buyer type to a monetary transfer. A mechanism is \emph{feasible} if it satisfies the incentive compatibility and the individual rationality constraints
\begin{align*}
\mathbf{x}\cdot\mathbf{q}(\mathbf{x})-r(\mathbf{x}) & \geq \mathbf{x}\cdot\mathbf{q}(\hat{\mathbf{x}})-r(\hat{\mathbf{x}}) & \forall\, & \mathbf{x},\hat{\mathbf{x}}\in \mathcal{X}, \tag{IC}\label{IC} \\
\mathbf{x}\cdot\mathbf{q}(\mathbf{x})-r(\mathbf{x}) & \geq 0 & \forall\, & \mathbf{x}\in \mathcal{X}.
\tag{IR}\label{IR}
\end{align*}
Let $u(\mathbf{x})=\mathbf{x}\cdot\mathbf{q}(\mathbf{x})-r(\mathbf{x})$ be the indirect utility function. By standard arguments, $(\mathbf{q},r)$ is feasible if and only if $u$ is nonnegative, convex, and satisfies $\mathbf{q}(\mathbf{x})\in \partial_{\mathcal{Q}} u(\mathbf{x})\subseteq \mathcal{Q}$ for every $\mathbf{x}\in \mathcal{X}$. The price schedule induced by $u$ is the convex conjugate of $u$, denoted by $p$:
\begin{equation}
\label{IC_u}
p(\mathbf{q})=\sup_{\mathbf{x}\in \mathcal{X}}\{\mathbf{x}\cdot\mathbf{q}-u(\mathbf{x})\}.
\tag{IC$_u$}
\end{equation}
Let $\mathcal{U}$ be the set of feasible $u$. Then the seller's revenue maximization problem is
\begin{equation*}
\max_{u\in \mathcal{U}}\int_{\mathcal{X}}u\dd \mu.
\tag{U}
\label{Eqn_Problem_u}
\end{equation*}

\Cref{Prop_Equiv} establishes the equivalence between \eqref{Eqn_Problem_p} and \eqref{Eqn_Problem_u}, which is essentially a consequence of the taxation principle \citep{laffont2002theory}.

\begin{proposition}
\label{Prop_Equiv}
A price schedule $p$ is optimal for \eqref{Eqn_Problem_p} if and only if its associated indirect utility function $u$ from \eqref{IC_p} is optimal for \eqref{Eqn_Problem_u}. Conversely, an indirect utility function $u$ is optimal for \eqref{Eqn_Problem_u} if and only if its associated price schedule $p$ from \eqref{IC_u} is optimal for \eqref{Eqn_Problem_p}.
\end{proposition}

%If $p$ is a feasible price schedule for \eqref{Eqn_Problem_p}, $u$ is also a feasible indirect utility function for \eqref{Eqn_Problem_u}. If $u$ is a feasible indirect utility function for \eqref{Eqn_Problem_u}, it is possible that $p(\mathbf{0})<0$, but the seller can always increase the prices for all bundles by a fixed amount and obtain a (weakly) higher revenue. In other words, an optimal $u$ for \eqref{Eqn_Problem_u} must induce an optimal $p$ for \eqref{Eqn_Problem_p}.

\subsection{An Example Where MR Is Not Well-Defined}
\label{Sec_MR_Example}

\begin{example}
\label{Exp_Uniform_2}
Let $N=2$, $\mathcal{X}=[\underline{x},\underline{x}+1]^2$ with $\underline{x}>0$, $\mathcal{Q}=[0,1]^2$, and $f(\mathbf{x})=1$. \citet{pavlov2011optimal} proves that bundles like $(q,1)$ and $(1,q)$ with $q\in (0,1)$ are active in the optimal mechanism when $\underline{x}$ is sufficiently small, and the optimal mechanism becomes pure bundling when $\underline{x}$ is large.
\end{example}

As depicted by the blue lines in \Cref{Fig_Uniform_Positive}, $\mathcal{D}(\mathbf{e}_{2})$ is a subset of the bottom boundary $\{\underline{x}\}\times [\underline{x},\underline{x}+1]$, so $\mu (\mathcal{D}(\mathbf{e}_{2}))<0$. However, every type in $\mathcal{D}(\mathbf{e}_{2})$ also demands either the lottery $(q,1)$ or the grand bundle $\mathbf{1}$. Hence, $\mathcal{D}^{*}(\mathbf{e}_{2})=\varnothing$ and 
\begin{equation*}
\mu (\mathcal{D}(\mathbf{e}_{2})\setminus \mathcal{D}^{*}(\mathbf{e}_{2}))=\mu (\mathcal{D}(\mathbf{e}_{2}))<0,
\end{equation*}
which means $\MR (\mathbf{e}_{2})$ is not well-defined. Intuitively, increasing and decreasing the price of $\mathbf{e}_{2}$ have distinct marginal effects on total revenue.

In a related paper, \citet{kash2016optimal} interpret $-\mu (\mathcal{D}(\mathbf{q}))$ as the marginal revenue of $\mathbf{q}$. Their Section 4.1 claims that $\mu (\mathcal{D}(\mathbf{q}))=0$ for \emph{each} $\mathbf{q}\in \mathcal{Q}$ is necessary for optimality, and their Section 5.1 establishes this result for bundles that are sold to a positive mass of buyers. However, their claim fails when the qualification \eqref{Eqn_Qualification} is not met, as illustrated by the example.

\begin{figure}[ht]
\centering
\begin{tikzpicture}[baseline={(0,0)},x=4.5cm,y=4.5cm]
\draw (0,2/3) -- (0,0) node [below left] {$\underline{x}$} -- (1,0) node [below] {$\underline{x}+1$} -- (1,1) -- (0,1) node [left] {$\underline{x}+1$};
\draw [ultra thick, blue] (0,1) -- (0,2/3) node [midway, left, black] {$\mathcal{D}(\mathbf{e}_{2})$};
\draw [dotted] (0,0) -- (1,1);
%\draw [thick, green] (1/4,2/3) -- (1,2/3);
\draw [thick, red] (0,2/3) -- (1/3,1/2) -- (1/2,1/3) -- (2/3,0);
\draw (1/6,0.8) node {$(q,1)$};
\draw (0.8,1/6) node {$(1,q)$};
\draw (0.8,0.8) node {$\mathbf{1}$};
\draw [dashed] (1/3,1/2) -- (1/3,1);
\draw [dashed] (1/2,1/3) -- (1,1/3);
\draw (0.5,-0.2) node {\footnotesize (a) $\underline{x}$ is small};
\end{tikzpicture}%
\begin{tikzpicture}[baseline={(0,0)},x=4.5cm,y=4.5cm]
\draw (0,2/3) -- (0,0) node [below left] {$\underline{x}$} -- (1,0) node [below] {$\underline{x}+1$} -- (1,1) -- (0,1) node [left] {$\underline{x}+1$};
\draw [ultra thick, blue] (0,1) -- (0,2/3) node [midway, left, black] {$\mathcal{D}(\mathbf{e}_{2})$};
\draw [dotted] (0,0) -- (1,1);
%\draw [thick, green] (1/4,2/3) -- (1,2/3);
\draw [thick, red] (0,2/3) -- (2/3,0);
\draw (0.8,0.8) node {$\mathbf{1}$};
\draw (0.5,-0.2) node {\footnotesize (b) $\underline{x}$ is large};
\end{tikzpicture}
\caption{Uniform distribution on $[\underline{x},\underline{x}+1]^2$ with $\underline{x}>0$.}
\label{Fig_Uniform_Positive}
\end{figure}

%Moreover, their claim alone does not provide the same power for finding optimal mechanisms as \Cref{Thm_MR} due to two reasons. First, our first-order conditions \hyperref[Eqn_FOC]{(FOCs)} are formulated with respect to \emph{sets of bundles} rather than individual bundles. This added generality is essential for identifying zero-measure demand sets (see part (ii) of Propositions \ref{Prop_Additive} and \ref{Prop_SubsComp}). Second, we impose an additional second-order condition \eqref{Eqn_SOC}, which is crucial for characterizing the endpoints of demand sets (see Propositions \ref{Prop_Additive} and \ref{Prop_SubsComp}).

\subsection{Properties of \texorpdfstring{$\zeta$}{zeta} in All Examples of Section \ref{Sec_Additive_Menu}}
\label{Sec_Examples}

When $h$ is twice differentiable and $h'(x)\neq 0$, we can compute $\zeta'$ and $\zeta''$ from \eqref{Eqn_Zeta_iid}:
\begin{align}
\label{Eqn_Zeta'}
\zeta'(x) & =\frac{h'(x)\zeta(x)}{[h(x)+2][h(\zeta (x))+h(x)+3]}, \\
\label{Eqn_Zeta''}
\zeta''(x) & =\frac{h''(x)\zeta'(x)}{h'(x)}-\left[\frac{h(\zeta (x))+2h(x)+4}{\zeta (x)}+\frac{h'(\zeta (x))}{h(\zeta (x))+h(x)+3}\right]\zeta'(x)^{2}.
\end{align}

\subsubsection{Power-law Distribution}

When $\eta\geq 1$, $h$ is nondecreasing, so $\zeta$ is nondecreasing. Moreover, it can be computed that $\kappa =\eta /[(1+\theta)^{\eta}-\theta^{\eta}]$ and $G(x)=[(x+\theta)^{\eta}-\theta^{\eta}]/[(1+\theta)^{\eta}-\theta^{\eta}]$. If $\zeta (x)=x$, then $x$ solves
\begin{equation*}
\frac{xg(x)}{1-G(x)}=h(x)+2,
\end{equation*}
which can be reformulated as
\begin{equation*}
\frac{\eta(x+\theta)^{\eta-1}}{(1+\theta)^{\eta}-(x+\theta)^{\eta}}=\frac{\eta-1}{x+\theta}+\frac{2}{x}.
\end{equation*}
The left-hand side is strictly increasing in $x$ and approaches $+\infty$, whereas the right-hand side is strictly decreasing in $x$ from $+\infty$. Therefore, $\zeta (x)=x$ has a unique solution.

\subsubsection{Truncated Pareto distribution}

Since $h'<0$, $\zeta'<0$. Define $s=x/(1+x)$ and $t=\zeta (x)/(1+\zeta (x))$. Then \eqref{Eqn_Zeta'} and \eqref{Eqn_Zeta''} can be rewritten as
\begin{equation*}
\begin{aligned}
\zeta'(x) & =-\frac{\eta(1-s)^2t}{(1-t)(2-\eta s)[3-\eta(s+t)]}, \\
\zeta''(x) & =-2(1-s)\zeta'(x)-\left[\frac{(1-t)(4-2\eta s-\eta t)}{t}-\frac{\eta(1-t)^2}{3-\eta(s+t)}\right]\zeta'(x)^2.
\end{aligned}
\end{equation*}
It can be verified that $\zeta''(x)>0$ is equivalent to
\begin{equation*}
2(2-\eta s)(3-\eta(s+t))^2-\eta(1-s)[(4-2\eta s-\eta t)(3-\eta(s+t))-\eta t(1-t)]>0.
\end{equation*}
Since $s,t\in [0,1/2]$ and $0<\eta<3$, the left-hand side is bounded below by $(3-\eta)/4$, which implies $\zeta''(x)>0$.

\subsubsection{Uniform Marginals with Correlation}

By direct computation,
\begin{equation*}
\phi (\mathbf{x})=3+3\rho[3-8(x_{1}+x_{2})+20x_{1}x_{2}].
\end{equation*}
When $\rho<0$, $\phi (\mathbf{x})\geq 3+21\rho$. When $\rho>0$, $\phi (\mathbf{x})\geq 3-15\rho$. Hence $f$ is strictly regular if and only if $\rho \in [-1/7,1/5]$.

For $\rho\neq 0$, we use \eqref{Eqn_Zeta} to express $x$ as a function of $\zeta$:
\[
x=\Xi (\zeta)=\frac{3\zeta-2-3\rho\zeta(4\zeta-3)}{6\rho\zeta(4-5\zeta)}.
\]
It suffices to identify the convexity or concavity of $\Xi$.
By direct computation,
\begin{equation*}
\Xi'(\zeta )=\frac{1}{12\rho}\left[\frac{1}{\zeta^{2}}+\frac{5-6\rho}{(4-5\zeta )^{2}}\right],
\end{equation*}
which implies $\sgn \Xi'=\sgn \rho$.
Moreover,
\[
\lim_{\zeta\downarrow 0}\Xi (\zeta)
=
\begin{cases}
-\infty, & \rho>0,\\
+\infty, & \rho<0,
\end{cases}
\qquad
\lim_{\zeta\uparrow4/5}\Xi (\zeta)
=
\begin{cases}
+\infty, & \rho>0,\\
-\infty, & \rho<0.
\end{cases}
\]
Thus, for every $x\in[0,1]$, there is a unique solution in
$(0,4/5)$. There is no additional solution in $(4/5,1]$, because
\[
\Xi (1)=\frac{3\rho-1}{6\rho}
\begin{cases}
<0, & \rho>0,\\
>1, & \rho<0.
\end{cases}
\]
Therefore, the unique solution in $[0,1]$ satisfies $\zeta(x)<4/5$.

We next establish a sufficiently tight lower bound. If $\rho>0$, then
\[
\Xi \left(\frac{21}{40}\right)
=
\frac{-\frac{17}{40}+\frac{567}{400}\rho}
     {\frac{693}{160}\rho}
<0,
\]
where the inequality follows from $\rho\leq1/5$. Since
$\Xi $ is increasing and $x\geq0$, this implies
$\zeta(x)>21/40$. If $\rho<0$, then
\[
\Xi \left(\frac{21}{40}\right)-1
=
\frac{-\frac{17}{40}-\frac{2331}{800}\rho}
     {\frac{693}{160}\rho}
>0,
\]
where the inequality follows from
$\rho\geq-1/7$. Since
$\Xi $ is decreasing and $x\leq 1$, this again implies
$\zeta(x)>21/40$. Hence,
\[
\frac{21}{40}<\zeta(x)<\frac45.
\]

Differentiating $\Xi'$ once more gives
\[
\Xi ''(\zeta)
=
\frac{1}{6\rho}
\left[
\frac{5(5-6\rho)}{(4-5\zeta)^{3}}
-\frac{1}{\zeta^{3}}
\right].
\]
The preceding bounds imply
\[
\left(\frac{4-5\zeta}{\zeta}\right)^{3}
<
\left(\frac{55}{21}\right)^{3}
<19
\leq 5(5-6\rho).
\]
Consequently,
\[
\sgn \Xi''
=
\sgn \Xi'
=
\sgn \rho.
\]
Finally, differentiating
$x=\Xi (\zeta(x))$ twice yields
\[
\zeta''(x)
=
-\frac{\Xi ''(\zeta(x))}
       {\bigl[\Xi '(\zeta(x))\bigr]^{3}}
<0.
\]
Therefore, $\zeta$ is strictly concave whenever $\rho\neq 0$.

\subsubsection{Truncated Normal Distribution}

It is easy to verify that $h'(x)=0$ only if $x=\theta/2$, and in this case $\zeta''(\theta/2)<0$. 
When $h'(x)\neq 0$, inserting $h$ into \eqref{Eqn_Zeta''} yields
\begin{equation*}
\zeta''(x)=-\frac{2\zeta'(x)}{h'(x)}-\left[\frac{h(\zeta (x))+2h(x)+4}{\zeta (x)}+\frac{(\theta-2\zeta (x))}{h(\zeta (x))+h(x)+3}\right]\zeta'(x)^{2}.
\end{equation*}
Since $\zeta'$ and $h'$ share the same sign, it suffices to show that the expression in brackets is positive.
Using integration by parts,
\begin{equation*}
\begin{aligned}
\int^{1}_{x}[1-G(\hat{x})]\dd (\hat{x}-\theta) & =(\theta-x)(1-G(x))+\int^{1}_{x}(\hat{x}-\theta)g(\hat{x})\dd\hat{x} \\
& =(\theta-x)(1-G(x))+g(x)-g(1).
\end{aligned}
\end{equation*}
Since the left-hand side is nonnegative,
\begin{equation*}
\theta-x+\frac{g(x)}{1-G(x)}\geq \frac{g(1)}{1-G(x)}\geq g(1)\geq e^{-(1-\theta)^{2}/2}\geq 1-\frac{(1-\theta)^{2}}{2}\geq \frac{1}{2},
\end{equation*}
where the last inequality comes from $e^{\tau}\geq 1+\tau$. Therefore,
\begin{equation*}
h(\zeta (x))+h(x)+3=1+\zeta (x)(\theta-\zeta (x))+\frac{\zeta (x)g(\zeta (x))}{1-G(\zeta (x))}\geq 1+\frac{\zeta (x)}{2}.
\end{equation*}
Moreover, since $h(x)\geq -1$, $h(\zeta(x))+2h(x)+4 \geq h(\zeta(x))+h(x)+3$.
Hence, the numerator of the expression in brackets is bounded below by
\[
\left(1+\frac{\zeta(x)}{2}\right)^2+\zeta (x)(\theta-2\zeta (x)) \geq \left(1+\frac{\zeta(x)}{2}\right)^2-2\zeta (x)^{2}\geq \frac{1}{4},
\]
which means $\zeta''(x)<0$.

\subsubsection{Truncated Normal Distribution with Correlation}

Suppose that $\mathbf{x}$ follows a bivariate normal distribution with correlation:
\begin{equation*}
f(\mathbf{x})=\kappa \exp \left\{-\frac{(x_{1}-\theta)^{2}-2\rho (x_{1}-\theta)(x_{2}-\theta)+(x_{2}-\theta)^{2}}{2(1-\rho^{2})}\right\}.
\end{equation*}
With a slight abuse of notation, denote by $h(\mathbf{x})=\mathbf{x}\cdot \nabla f(\mathbf{x})/f(\mathbf{x})$. By direct computation,
\begin{equation*}
h(\mathbf{x})=-\frac{x_{1}^{2}-2\rho x_{1}x_{2}+x_{2}^{2}-\theta (1-\rho)(x_{1}+x_{2})}{1-\rho^{2}}.
\end{equation*}
It can be verified that $h$ is concave in $\mathbf{x}$, so $h(\mathbf{x})$ is bounded below by its continuous limit at $(0,0)$, $(1,0)$, and $(1,1)$. Direct computation gives
\begin{equation*}
h(\mathbf{x})\geq \min\left\{-\frac{2-2\theta}{1+\rho},-\frac{1-\theta (1-\rho)}{1-\rho^{2}},0\right\}.
\end{equation*}
Hence, when $\theta\leq (1-2\rho)/(1-\rho)$, $f$ is strictly regular if $1+3\rho +2\theta >0$. When $\theta >(1-2\rho)/(1-\rho)$, then $f$ is strictly regular if $\theta >(3\rho^{2}-2)/(1-\rho)$. When $\theta=0$, $f$ is strictly regular if and only if $\rho \in [-1/3,\sqrt{2/3}]$.

%so it is upper bounded by $\max (2-2\rho-2\theta (1-\rho),1-\theta (1-\rho))$,

\subsubsection{Truncated Gamma Distribution}

Insert $h$ into \eqref{Eqn_Zeta''}:
\begin{equation*}
\zeta''(x)=-\left[\frac{h(\zeta (x))+2h(x)+4}{\zeta (x)}+\frac{-\theta}{h(\zeta (x))+h(x)+3}\right]\zeta'(x)^{2}.
\end{equation*}
By direct computation,
\begin{equation*}
[h(\zeta (x))+2h(x)+4][h(\zeta (x))+h(x)+3]-\theta \zeta (x)\geq 1-\theta \zeta (x)>0,
\end{equation*}
which implies $\zeta''(x)<0$.

\subsection{A General Upper Bound on Menu Size}
\label{Sec_Upperbound}

In the setting of \Cref{Sec_Additive}, even if $\zeta$ does not have a desirable global property, we can obtain an upper bound for the number of pooling bundles in any optimal mechanism. On top of SRS, we only assume that $\zeta$ crosses the 45-degree line exactly once.

Let $x_{\zeta}$ be the unique solution of $\zeta(x)=x$. Define
\begin{equation*}
X_{\supp}
=
\left\{
x\in(0,x_{\zeta}):
\begin{array}{l}
\exists\, q\in(0,1)\text{ and }\varepsilon>0\text{ such that}\\
\zeta(y)\geq \zeta(x)-q(y-x)
\quad
\forall\,y\in(x-\varepsilon,x+\varepsilon)\cap[0,1]
\end{array}
\right\}.
\end{equation*}
Thus, $X_{\supp}$ consists of the points at which $\zeta$ admits a local supporting line from below with slope in $(-1,0)$. Denote the connected components of $X_{\supp}$ by $\{C_{i}\}_{i\in\mathcal{I}}$, so that each $C_{i}$ is an interval.

Moreover, $\zeta$ is convex on every $C_i$. Suppose otherwise. Then there exist $a<b$ in $C_i$ such that $\zeta$ lies strictly above the chord $\ell$ joining $(a,\zeta(a))$ and $(b,\zeta(b))$ at some point in $(a,b)$. Let $x$ be the leftmost maximizer of $\zeta-\ell$ on $[a,b]$. Then $x\in(a,b)$. By the definition of $X_{\supp}$, there exist $q\in(0,1)$ and $\varepsilon>0$ such that, for all $y$ sufficiently close to $x$,
\[
(\zeta-\ell)(y)\geq(\zeta-\ell)(x)+(-q-\ell')(y-x).
\]
Since $x$ maximizes $\zeta-\ell$, $-q-\ell'$ can be neither positive nor negative, and hence must equal zero. The displayed inequality then implies that $\zeta-\ell$ is constant in a neighborhood of $x$, contradicting the choice of $x$ as the leftmost maximizer.

\begin{corollary}
\label{Cor_MenuSize}
Suppose that $\zeta$ crosses the 45-degree line exactly once. Then in any optimal mechanism,
\[
\left|\{q\in(0,1):(q,1)\text{ is pooling}\}\right|\leq|\mathcal{I}|.
\]
In particular, if $|\mathcal{I}|$ is finite, the number of pooling bundles in any optimal mechanism is at most $2|\mathcal{I}|+3$.
\end{corollary}

\begin{proof}
Fix an optimal $\overline{u}$ and define $\Delta (x)=\zeta(x)+\overline{u}(x)-1$. As in the proof of \Cref{Cor_Inc}, $a^{1}<x_{\zeta}$, so on $[0,a^{1}]$ the sign of $\Phi$ is the same as the sign of $\Delta$.

Now suppose that $(q,1)$ is pooling for some $q\in(0,1)$. Then $\overline{u}$ is affine with slope $q$ on $[a^{q},b^{q}]$. By part \ref{Prop_Additive_2} of \Cref{Prop_Additive}, $\Phi(a^{q}),\Phi(b^{q})\geq 0$ and \eqref{Eqn_IntPhi} holds. Since $a^{q},b^{q}\leq a^{1}<x_{\zeta}$, the endpoint inequalities imply $\Delta (a^{q}),\Delta (b^{q})\geq 0$.

Suppose first that $\Phi=0$, which also implies $\Delta =0$ throughout $[a^{q},b^{q}]$. Then $\zeta=1-\overline{u}$ is affine with slope $-q$ on this interval. Choose any $y^{q}\in (a^{q},b^{q})$. Then $y^{q}\in X_{\supp}$ and $\Delta(y^{q})=0$.

Suppose instead that $\Phi$ is not identically zero on $[a^{q},b^{q}]$. Then $\Phi<0$, and thus $\Delta<0$, somewhere in $(a^{q},b^{q})$. Let $y^{q}$ minimize $\Delta $ on $[a^{q},b^{q}]$. Since $\Delta $ is nonnegative at both endpoints, $y^{q}\in (a^{q},b^{q})$ and $\Delta(y^{q})<0$. Since $\overline{u}$ is affine on $[a^{q},b^{q}]$,
\begin{align*}
\Delta(x)\geq\Delta(y^{q})
&\iff
\zeta(x)+\overline{u}(x)
\geq
\zeta(y^{q})+\overline{u}(y^{q})\\
&\iff
\zeta(x)\geq\zeta(y^{q})-q(x-y^{q})
\end{align*}
for every $x\in [a^{q},b^{q}]$. Thus $y^{q}\in X_{\supp}$.

Next, suppose that $(\hat{q},1)$ and $(\check{q},1)$ are two distinct pooling bundles with $y^{\hat{q}},y^{\check{q}}\in C_{i}$. Without loss of generality, let $\hat{q}<\check{q}$. Then $a^{\hat{q}}<b^{\hat{q}}\leq a^{\check{q}}<b^{\check{q}}$, and hence $y^{\hat{q}}<b^{\hat{q}}\leq a^{\check{q}}<y^{\check{q}}$. By part \ref{Prop_Additive_2} of \Cref{Prop_Additive}, $\Delta(b^{\hat{q}}),\Delta(a^{\check{q}})\geq 0$, whereas, by our construction, $\Delta(y^{\hat{q}}),\Delta(y^{\check{q}})\leq 0$. Since $\Delta$ is convex on $C_{i}$, there must be $\Delta(y^{\hat{q}})=\Delta(y^{\check{q}})=0$. Convexity then implies that $\Delta=0$ throughout $[y^{\hat{q}},y^{\check{q}}]$. Thus $\zeta=1-\overline{u}$ on this interval. Since $\zeta$ is convex while $1-\overline{u}$ is concave, both are affine. Therefore $\overline{u}$ is affine on $[y^{\hat{q}},y^{\check{q}}]$, contradicting $\hat{q}\neq\check{q}$.

Therefore, each interval $C_{i}$ contains the selected point $y^{q}$ associated with at most one pooling bundle $(q,1)$ with $q\in(0,1)$, which implies
\[
\left|\{q\in(0,1):(q,1)\text{ is pooling}\}\right|
\leq|\mathcal{I}|.
\]
By symmetry, there are at most $|\mathcal{I}|$ pooling bundles of the
form $(1,q)$ with $q\in(0,1)$. Adding the possibly pooling bundles
$\mathbf{e}_{1}$, $\mathbf{e}_{2}$, and $\mathbf{1}$ gives the bound
$2|\mathcal{I}|+3$ whenever $|\mathcal{I}|$ is finite.
\end{proof}

\subsection{Results Omitted from Section \ref{Sec_SubsComp}}
\label{Sec_SubsComp_Omitted}

As a starting point, \Cref{Lem_MM_k} is an extension of \Cref{Lem_MM} and is largely due to \citet[Proposition 2]{pavlov2011property}.

\begin{lemma}
\label{Lem_MM_k}
Suppose that $f$ is SRS. For any optimal symmetric mechanism, a bundle is active only if it takes the form of $(q,1)_{k}$ or $(1,q)_{k}$. Moreover, whenever the relevant Cartesian point belongs to $\mathcal{X}$,
%whenever $x_{1}-\gamma_{k}(x_{2}),x_{2}\in \mathcal{X}$ is above the 45-degree line,
\begin{equation}
\label{Eqn_MM_k}
\begin{aligned}
u (x_{1}-\gamma_{k}(x_{2}),x_{2}) & =\max (\overline{u}(x_{1})-(1-x_{2}),0), \quad \text{if }x_{1}-\gamma_{k}(x_{2})\leq x_{2}, \\
u  (x_{1},x_{2}-\gamma_{k}(x_{1})) & =\max (\overline{u}(x_{2})-(1-x_{1}),0), \quad \text{if }x_{1}\geq x_{2}-\gamma_{k}(x_{1}),
\end{aligned}
\end{equation}
where $\overline{u}(\cdot)=u(\cdot,1)=u(1,\cdot)$.
\end{lemma}

Next, we characterize demand sets. Let $[a^{q},b^{q}]=\{x\in [0,1]: \overline{u}_{-}'(x)\leq qk\leq \overline{u}_{+}'(x)\}$. Suppose that $(x_{1},1)\in \mathcal{D}((q,1)_{k})$. Then $(x_{1}-\gamma_{k}(x_{2}),x_{2})$ also belongs to $\mathcal{D}((q,1)_{k})$ if
\begin{equation*}
x_{2}\geq x_{1}-\gamma_{k}(x_{2})\geq 0, \text{ and }\overline{u}(x_{1})-(1-x_{2})\geq 0.
\end{equation*}
There solutions are exactly $x_{2}\geq \max(1-\overline{u}(x_{1}),z_{k}(x_{1}))$, where
\begin{equation}
\label{Eqn_z_k}
z_{k}(x_{1})=
\begin{cases}
1-x_{1} & k=1/2, \\
\max \left(\frac{kx_{1}+k-1}{k-1},\frac{kx_{1}+k-1}{2k-1}\right) & 1/2<k<1, \\
\frac{kx_{1}+k-1}{2k-1} & k\geq 1.
\end{cases}
\end{equation}

Fix $q<1$. When either $k\leq 1$ or $a^{q}>0$, the demand set of $(q,1)_{k}$ is
\begin{equation}
\label{Eqn_D_k}
\mathcal{D}((q,1)_{k})=\{(x_{1}-\gamma_{k}(x_{2}),x_{2})\in \mathcal{X}:a^{q}\leq x_{1}\leq b^{q}, x_{2}\geq \max(1-\overline{u}(x_{1}),z_{k}(x_{1}))\},
\end{equation}
and the cumulative MR is
\begin{equation*}
\Phi_{k}(x_{1})=\int^{1}_{\max(1-\overline{u}(x_{1}),z_{k}(x_{1}))}\phi (x_{1}-\gamma_{k}(x_{2}),x_{2})\dd x_{2}-f(x_{1},1).
\end{equation*}

When $k>1$ and $a^{q}=0$, there is an edge case where some buyers with $x_{1}<(k-1)(1-x_{2})/k$ also purchase $(q,1)_{k}$.\footnote{If these buyers purchase another bundle, say, $(\hat{q},1)_{k}$, then $\mathcal{D}((\hat{q},1)_{k})$ will not touch the top boundary. As a result, $\MR (\overline{B}_{\delta}((\hat{q},1)_{k}))$ is strictly positive for sufficiently small $\delta>0$, a contradiction to \eqref{Eqn_FOC}.}
In this case, the demand set of $(q,1)_{k}$ is
\begin{equation*}
\begin{aligned}
\mathcal{D}((q,1)_{k}) & =\{(x_{1}-\gamma_{k}(x_{2}),x_{2})\in \mathcal{X}: 0\leq x_{1}\leq b^{q},x_{2}\geq \max(1-\overline{u}(x_{1}),z_{k}(x_{1}))\} \\
& \quad\ \cup \left\{(x_{1}-\gamma_{k}(x_{2}),x_{2})\in \mathcal{X}: x_{1}<0,\text{ and }x_{2}\geq \max(1-\overline{u}(0)-qkx_{1},z_{k}(x_{1}))\right\}.
\end{aligned}
\end{equation*}

\Cref{Prop_SubsComp} below extends \Cref{Prop_Additive} to $k \geq 1/2$.

\begin{proposition}
\label{Prop_SubsComp}
Suppose that $f$ is SRS. Then for every optimal mechanism and every active bundle $(q,1)_{k}$ such that either $k\leq 1$ or $a^{q}>0$:
\begin{enumerate}
\item \label{Prop_SubsComp_2}
If $(q,1)_{k}$ is pooling, then
\begin{equation}
\label{Eqn_IntPhi_k}
\int_{a^{q}}^{b^{q}}\Phi_{k}(x_{1})\dd x_{1}=0.
\end{equation}
Moreover, $\Phi_{k}(a^{q})\geq 0$ unless $q=0$, and $\Phi_{k}(b^{q})\geq 0$ unless $q=1$.
\item \label{Prop_SubsComp_3}
If $(q,1)_{k}$ is separating, then $a^{q}=b^{q}$ and $\Phi_{k}(a^{q})=0$.
\end{enumerate}
When $k>1$ and $a^{q}=0$, the conclusion $\Phi_{k}(b^{q})\geq 0$ in part \ref{Prop_SubsComp_2} continues to hold
for $q<1$.
\end{proposition}

In what follows, we present several corollaries based on \Cref{Prop_SubsComp}. 
\Cref{Cor_GrandBundle_k} is an extension of \Cref{Cor_GrandBundle}.

\begin{corollary}
\label{Cor_GrandBundle_k}
Let $k>1/2$. Then in every optimal mechanism, the grand bundle $\mathbf{k}$ is pooling.
\end{corollary}

Corollaries \ref{Cor_k=1/2}--\ref{Cor_k>1} are extensions of Corollaries \ref{Cor_Inc}--\ref{Cor_CcvDec}. We only sketch their proofs, emphasizing their distinctions relative to the arguments used for Corollaries \ref{Cor_Inc}--\ref{Cor_CcvDec}.

\begin{corollary}
\label{Cor_k=1/2}
Suppose that $k=1/2$ and there exists $x^{*}\in [0,1]$ such that $\zeta_{k}(x)=0$ for all $x\in [0,x^{*})$. Then for every optimal mechanism:
\begin{enumerate}
\item If $\zeta_{k}$ is nondecreasing on $[x^{*},1]$, then $\overline{u}'(x)=0$ for $x\in [0,1]$. \label{Cor_k=1/2_1}
\item If $\zeta_{k}$ is convex on $[x^{*},1]$, then there exist $x^{*}<b^{0}\leq 1$ and $0<q\leq 1$ such that:
\begin{enumerate}
    \item $\overline{u}'(x)=0$ for $x\in (0,b^{0})$,
    \item $\overline{u}'(x)=qk$ for $x\in (b^{0},1)$.
\end{enumerate}
\item If $\zeta_{k}$ is concave on $[x^{*},1]$, then there exist $x^{*}<b^{0}\leq a^{1}\leq 1$ such that:
\begin{enumerate}
    \item $\overline{u}'(x)=0$ for $x\in (0,b^{0})$,
    \item $\overline{u}(x)=1-\zeta_{k}(x)$ for $x\in (b^{0},a^{1})$,
    \item $\overline{u}'(x)=k$ for $x\in (a^{1},1)$.
\end{enumerate}
\end{enumerate}
\end{corollary}

\begin{proof}[Proof Sketch]
There are two key differences between the cases $k=1/2$ and $k=1$. First, $\mathbf{k}$ is not necessarily pooling when $k=1/2$. Second, by \eqref{Eqn_SOC}, both $\mathbf{e}_{1}$ and $\mathbf{e}_{2}$ must be pooling when $k=1/2$.  Consequently, when $\zeta_{k}$ is nondecreasing, optimal mechanisms are separate selling. When $\zeta_{k}$ is convex (or concave), we allow $b^{0}=1$ (or $a^{1}=1$), but instead require $b^{0}>x^{*}$.
\end{proof}

\begin{corollary}
\label{Cor_1/2<k<1}
Suppose that $1/2<k<1$, that $\zeta_{k}$ crosses the line $x_{2}=(kx_{1}+k-1)/(2k-1)$ exactly once, and there exists $x^{*}\in [0,1]$ such that $\zeta_{k}(x)=0$ for all $x\in [0,x^{*})$. Then for every optimal mechanism:
\begin{enumerate}
\item If $\zeta_{k}$ is nondecreasing on $[x^{*},1]$, then there exists $x^{*}<a^{1}<1$ such that: \label{Cor_1/2<k<1_1}
\begin{enumerate}
    \item $\overline{u}'(x)=0$ for $x\in (0,a^{1})$,
    \item $\overline{u}'(x)=k$ for $x\in (a^{1},1]$.
\end{enumerate}
\item If $\zeta_{k}$ is convex on $[x^{*},1]$, then there exist $x^{*}<b^{0}\leq a^{1}<1$ and $0<q<1$ such that:
\begin{enumerate}
    \item $\overline{u}'(x)=0$ for $x\in (0,b^{0})$,
    \item $\overline{u}'(x)=qk$ for $x\in (b^{0},a^{1})$.
    \item $\overline{u}'(x)=k$ for $x\in (a^{1},1)$.
\end{enumerate}
\item If $\zeta_{k}$ is concave on $[x^{*},1]$, then there exist $x^{*}<b^{0}\leq a^{1}<1$ such that:
\begin{enumerate}
    \item $\overline{u}'(x)=0$ for $x\in (0,b^{0})$,
    \item $\overline{u}(x)=1-\zeta_{k}(x)$ for $x\in (b^{0},a^{1})$,
    \item $\overline{u}'(x)=k$ for $x\in (a^{1},1)$.
\end{enumerate}
\end{enumerate}
\end{corollary}

\begin{proof}[Proof Sketch]
This case is an intermediate scenario between the cases $k=1/2$ and $k=1$. By \eqref{Eqn_SOC}, $\mathbf{e}_{1}$, $\mathbf{e}_{2}$, and $\mathbf{k}$ are all pooling in optimal mechanisms. Consequently, when $\zeta_{k}$ is nondecreasing, optimal mechanisms are deterministic. When $\zeta_{k}$ is convex, optimal mechanisms may have two additional pooling bundles relative to \Cref{Cor_Cvx}. When $\zeta_{k}$ is concave, optimal mechanisms exhibit features analogous to those characterized in \Cref{Cor_Ccv}.
\end{proof}

\begin{corollary}
\label{Cor_k>1}
Suppose that $k>1$, that $\zeta_{k}$ crosses the line $x_{2}=(kx_{1}+k-1)/(2k-1)$ exactly once, and there exists $x^{*}\in [0,1]$ such that $\zeta_{k}(x)=0$ for all $x\in (x^{*},1]$. Then for every optimal mechanism:
\begin{enumerate}
\item If $\zeta_{k}$ is nonincreasing on $[0,x^{*}]$, then both $\mathbf{e}_{1}$ and $\mathbf{e}_{2}$ are inactive. \label{Cor_k>1_1}
\item If $\zeta_{k}$ is nondecreasing on $[0,x^{*}]$, then there exist $0\leq a^{1}<1$ and $0\leq q<1$ such that: \label{Cor_k>1_2}
\begin{enumerate}
    \item $\overline{u}'(x)=qk$ for $x\in (0,a^{1})$,
    \item $\overline{u}'(x)=k$ for $x\in (a^{1},1)$.
\end{enumerate}
\item If $\zeta_{k}$ is convex on $[0,x^{*}]$, then there exist $0\leq a^{1}<1$ and $0\leq q<1$ such that: \label{Cor_k>1_3}
\begin{enumerate}
    \item $\overline{u}'(x)=qk$ for $x\in (0,a^{1})$,
    \item $\overline{u}'(x)=k$ for $x\in (a^{1},1)$.
\end{enumerate}
\item If $\zeta_{k}$ is concave on $[0,x^{*}]$, then there exist $0\leq b^{q}\leq a^{1}<1$ and $0\leq q<1$ such that:\label{Cor_k>1_4}
\begin{enumerate}
    \item $\overline{u}'(x)=qk$ for $x\in (0,b^{q})$,
    \item $\overline{u}(x)=1-\zeta_{k}(x)$ for $x\in (b^{q},a^{1})$,
    \item $\overline{u}'(x)=k$ for $x\in (a^{1},1)$.
\end{enumerate}
\end{enumerate}
\end{corollary}

\begin{proof}[Proof Sketch]
Part \ref{Cor_k>1_1} of the corollary is new and we will prove it by contradiction. Suppose that $\zeta_{k}$ satisfies all the features listed in the corollary, and $\mathbf{e}_{2}$ is active. Then $a^{0}=0$, and
\begin{equation*}
\begin{aligned}
\mathcal{D}(\mathbf{e}_{2}) & =\{(x_{1}-\gamma_{k}(x_{2}),x_{2})\in \mathcal{X}: 0\leq x_{1}\leq b^{0},x_{2}\geq \max(1-\overline{u}(x_{1}),z_{k}(x_{1}))\} \\
& \quad\ \cup \{(x_{1}-\gamma_{k}(x_{2}),x_{2})\in \mathcal{X}: x_{1}<0,\text{ and }x_{2}\geq \max(1-\overline{u}(0),z_{k}(x_{1}))\}.
\end{aligned}
\end{equation*}
Hence, $\mathbf{e}_{2}$ is pooling. If $b^{0}=0$, $\mathcal{D}(\mathbf{e}_{2})$ has no intersection with the top boundary of $\mathcal{X}$, which implies $\MR (\mathbf{e}_{2})>0$. If $b^{0}>0$, then by an argument similar to that of \Cref{Prop_SubsComp}, $\max (1-\overline{u}(b^{0}),z_{k}(b^{0}))\leq \zeta_{k}(b^{0})$. Since $\zeta_{k}$ is nonincreasing, the single-dipped curve $\max (1-\overline{u}(x),z_{k}(x))$ must lie below $\zeta_{k}$ for all $x\in [0,b^{0})$, which also implies $\MR (\mathbf{e}_{2})>0$. Both cases lead to a contradiction to \Cref{Prop_SubsComp}.

Parts \ref{Cor_k>1_2}--\ref{Cor_k>1_4} of the corollary are similar to Corollaries \ref{Cor_k=1/2} and \ref{Cor_1/2<k<1}. According to \Cref{Prop_Pooling}, $\mathbf{e}_{1}$ and $\mathbf{e}_{2}$ may not be pooling, while $\mathbf{k}$ must be pooling in optimal mechanisms. Consequently, we allow $q=0$ but require $a^{1}<1$ in parts \ref{Cor_k>1_2}--\ref{Cor_k>1_4}.
\end{proof}

\begin{table}[ht]
\caption{Potentially active bundles for different $k$.}
\label{Tab_3}
\centering
\resizebox{\textwidth}{!}{
\begin{tabular}{cccc}
\toprule
 & \hspace*{2cm} $k=1/2$ \hspace*{2cm} & \hspace*{2cm} $1/2<k<1$ \hspace*{2cm} & $k>1$ \\
 & \Cref{Cor_k=1/2} & \Cref{Cor_1/2<k<1} & \Cref{Cor_k>1} \\
\midrule
Nondecreasing $\zeta_{k}$ & $\mathbf{e}_{1}$, $\mathbf{e}_{2}$ & $\mathbf{e}_{1}$, $\mathbf{e}_{2}$, $\mathbf{k}$ & $(q,1)_{k}$, $(1,q)_{k}$, $\mathbf{k}$ \\
\cline{2-4}
Convex $\zeta_{k}$
& $\mathbf{e}_{1}$, $\mathbf{e}_{2}$, $(q,1)_{k}$, $(1,q)_{k}$ & $\mathbf{e}_{1}$, $\mathbf{e}_{2}$, $(q,1)_{k}$, $(1,q)_{k}$, $\mathbf{k}$ & $(q,1)_{k}$, $(1,q)_{k}$, $\mathbf{k}$ \\
\cline{2-4}
Concave $\zeta_{k}$ & \multicolumn{3}{c}{General structure: $\overline{u}'(x)=qk$ for $x\in [0,b^{q})$; 
$\overline{u}(x)=1-\zeta_{k}(x)$ for $x\in [b^{q},a^{1})$;
$\overline{u}'(x)=k$ for $x\in [a^{1},1]$.} \\
& \multicolumn{3}{c}{When $k<1$, $\mathbf{e}_{1}$ and $\mathbf{e}_{2}$ are pooling, so $q=0$. When $k>1/2$, $\mathbf{k}$ is pooling, so $a^{1}<1$.} \\
\bottomrule
\end{tabular}
}
\end{table}

\end{document}